\documentclass[11pt]{article}

\usepackage[utf8]{inputenc}
\usepackage[letterpaper, margin=1in]{geometry}
\usepackage{amsmath,amssymb,amsthm,amsfonts}
\usepackage{mathtools}
\usepackage{graphicx} 
\usepackage{float}
\usepackage{subfigure} 
\usepackage{arydshln}
\usepackage{multirow}
\usepackage{bm}

\usepackage{hyperref}
\hypersetup{colorlinks,linkcolor=blue,urlcolor=blue,citecolor=blue}
\usepackage{url}
\usepackage{xcolor}
\usepackage[title]{appendix}
\usepackage{cite}
\usepackage{tikz}  %graphics
\usepackage{tikz-cd}  %commutative diagram
\usetikzlibrary{positioning}
\usepackage{tikz-qtree}
\usepackage{mathdots}
\usepackage{authblk}

\usepackage{qcircuit}
\newcommand{\ket}[1]{| #1 \rangle}
\newcommand{\bra}[1]{\langle #1|}
\newcommand{\braket}[1]{\langle #1 \rangle}
\newcommand{\bracket}[3]{\langle #1|#2|#3 \rangle}

\DeclareMathOperator{\polylog}{polylog}

\DeclareMathOperator{\sgn}{sgn}

\DeclareMathOperator{\ad}{ad}
\renewcommand{\Re}{\operatorname{Re}}
\renewcommand{\Im}{\operatorname{Im}}

\def \eps {\varepsilon}

\renewcommand{\d}{\mathrm{d}}
\renewcommand{\i}{{\mathrm i}}

\allowdisplaybreaks

\newtheorem{thm}{Theorem}[section]

\newtheorem{prop}[thm]{Proposition}

\newtheorem{lem}[thm]{Lemma}

\newtheorem{prob}[thm]{Problem}

\theoremstyle{definition}

\newtheorem{defn}[thm]{Definition}
\newtheorem{rmk}[thm]{Remark}

\numberwithin{equation}{section}

\newtheorem*{thm*}{Theorem}
\newtheorem*{lem*}{Lemma}
\newtheorem*{prop*}{Proposition}
\newtheorem*{defn*}{Definition}
\newtheorem*{prob*}{Problem}
\newtheorem*{ques*}{Question}

\newcommand{\mb}{\mathbb}

\newcommand{\mc}{\mathcal}

\newcommand{\norm}[1]{\left\lVert #1 \right\rVert}

\newcommand{\ltrinorm}{\left|\!\left|\!\left|}
\newcommand{\rtrinorm}{\right|\!\right|\!\right|}
\def\Tr{\mathrm{Tr}}

\def\b{{\bf b}}

\usepackage{algorithm}
\usepackage{algorithmic}

\makeatletter
\newenvironment{breakablealgorithm}
  {% \begin{breakablealgorithm}
   \begin{center}
     \refstepcounter{algorithm}% New algorithm
     \hrule height1pt depth0pt \kern3pt% \@fs@pre for \@fs@ruled »­Ïß
     \renewcommand{\caption}[2][\relax]{% Make a new \caption
       {\raggedright\textbf{\ALG@name~\thealgorithm} ##2\par}%
       \ifx\relax##1\relax % #1 is \relax
         \addcontentsline{loa}{algorithm}{\protect\numberline{\thealgorithm}##2}%
       \else % #1 is not \relax
         \addcontentsline{loa}{algorithm}{\protect\numberline{\thealgorithm}##1}%
       \fi
       \kern3pt\hrule\kern3pt
     }
  }{% \end{breakablealgorithm}
     \kern3pt\hrule\relax%
   \end{center}
  }
\makeatother

\newcommand{\be}{\begin{equation}}
\newcommand{\ee}{\end{equation}}
\newcommand{\bea}{\begin{eqnarray}}
\newcommand{\eea}{\end{eqnarray}}
\newcommand{\bes}{\begin{equation*}}
\newcommand{\ees}{\end{equation*}}
\newcommand{\beas}{\begin{eqnarray*}}
\newcommand{\eeas}{\end{eqnarray*}}

\title{Quantum singular value transformation without block encodings: Near-optimal complexity with minimal ancilla}

\author[1]{Shantanav Chakraborty\thanks{shchakra@iiit.ac.in}}
\author[1]{Soumyabrata Hazra\thanks{soumyabrata.hazra@research.iiit.ac.in}}
\author[2,3]{Tongyang Li\thanks{tongyangli@pku.edu.cn}}
\author[4]{Changpeng Shao\thanks{changpeng.shao@amss.ac.cn}}
\author[2,3]{Xinzhao Wang\thanks{wangxz@stu.pku.edu.cn}}
\author[4]{Yuxin Zhang\thanks{zhangyuxin@amss.ac.cn}}

\affil[1]{\small{CQST and CSTAR, International Institute of Information Technology Hyderabad, Telangana, India}}

\affil[2]{Center on Frontiers of Computing Studies, Peking University, Beijing,  China}
\affil[3]{School of Computer Science, Peking University, Beijing, China}
\affil[4]{\small{SKLMS, Academy of Mathematics and Systems Science, Chinese Academy of Sciences, Beijing, China}}

\begin{document}

\maketitle
\thispagestyle{empty}
\begin{abstract}

We develop new algorithms for Quantum Singular Value Transformation (QSVT), a unifying framework that encapsulates most known quantum algorithms and serves as the foundation for new ones.
Existing implementations of QSVT rely on block encoding, incurring an intrinsic $O(\log L)$ ancilla overhead and circuit depth $\widetilde{O}(L d\lambda )$ for polynomial transformations of a Hamiltonian $H=\sum_{k=1}^L H_k$, where $d$ is the polynomial degree and $\lambda=\sum_{k}\|H_k\|$.

We introduce a simple yet powerful approach that utilizes only basic Hamiltonian simulation techniques, namely, Trotter methods, to: (i) eliminate the need for block encoding, (ii) reduce the ancilla overhead to only a single qubit, and (iii) still maintain near-optimal complexity. 
Our method achieves a circuit depth of $\widetilde{O}(L(d\lambda_{\mathrm{comm}})^{1+o(1)})$,  without requiring any complicated multi-qubit controlled gates. 
Moreover, $\lambda_{\mathrm{comm}}$ depends on the nested commutators of the terms of $H$ and can be substantially smaller than $\lambda$ for many physically relevant Hamiltonians, a feature absent in standard QSVT. 
To achieve these results, we make use of Richardson extrapolation in a novel way, systematically eliminating errors in any interleaved sequence of arbitrary unitaries and Hamiltonian evolution operators, thereby establishing a general framework that encompasses QSVT but is more broadly applicable.

As applications, we develop end-to-end quantum algorithms for solving linear systems and estimating ground state properties of Hamiltonians, both achieving near-optimal complexity without relying on oracular access. 
Overall, our results establish a new framework for quantum algorithms, significantly reducing hardware overhead while maintaining near-optimal performance, with implications for both near-term and fault-tolerant quantum computing.
\end{abstract}

%{\bf Keywords:} Quantum algorithms; quantum singular value transformation; quantum signal processing; Hamiltonian simulation; quantum linear systems; ground state property estimation
%{\color{red}remove when submit}

\newpage

%\newgeometry{top=2.35cm,bottom=2.35cm,left=2.35cm,right=2.35cm}
\tableofcontents
\clearpage
%\restoregeometry

\newpage
%\pagenumbering{arabic}
%\setcounter{page}{1}

\section{Introduction}
\label{sec:intro}
\subsection{Brief overview}
\label{subsec:overview}
Quantum singular value transformation (QSVT) \cite{gilyen2019quantum, martyn2021grand} has emerged as one of the most powerful paradigms in quantum computing. The strength of QSVT lies in its ability to directly implement polynomial transformations on the singular values of an operator $A$, provided $A$ is embedded in the top-left block of a unitary $U_A$, known as the block encoding of $A$ \cite{low2019hamiltonian, chakraborty_et_al:LIPIcs.ICALP.2019.33}. 
By making queries to the block encoding, QSVT realizes an arbitrary polynomial transformation with query complexity scaling linearly in the polynomial degree, all while using only a single ancilla qubit. This versatility makes QSVT a powerful tool for designing quantum algorithms, allowing a wide range of problems to be formulated as polynomial transformations of block-encoded operators.

However, the reliance on block encoding is a fundamental obstacle to the practical and theoretical efficiency of QSVT-based algorithms. Constructing a block encoding for a general (Hermitian) operator — particularly one expressed as a linear combination of local terms — incurs significant overhead. 
This typically incurs circuit depth $O(L)$ for an operator composed of $L$ terms, along with $O(\log L)$ additional ancilla qubits and sophisticated controlled operations. Given the limitations of current quantum devices and the importance of QSVT, reducing ancilla usage while preserving shallow circuit depth is crucial to making QSVT-based methods feasible in near- and intermediate-term quantum devices.

In this work, we overcome this fundamental drawback by developing new quantum algorithms that implement QSVT without block encodings, using only the simplest Hamiltonian simulation techniques, based on product formulas, also known as Trotterization \cite{lloyd1996universal, childs2021theory}. 
More precisely, consider any Hermitian operator $H$ expressed as
$
H=\sum_{k=1}^{L} H_k,
$
where $H_k$ can be local operators (acting non-trivially on a small subset of qubits). 
Almost all physical systems in nature are described by Hamiltonians of this form, such as the transverse field Ising model, the quantum Heisenberg model, and the Hubbard model. 
Such Hamiltonians are also widely studied in problems related to quantum computing, including Hamiltonian learning~\cite{huang2023learning,bakshi2024learning} and thermal state preparation~\cite{bakshi2024high}. 
Our algorithm implements a polynomial of degree $d$ to any such $H$, using only a single ancilla qubit and a circuit of depth $\widetilde{O}(L(d\lambda_{\rm comm})^{1+o(1)})$,\footnote{Throughout this article, we use $\widetilde{O}(\cdot)$ to hide polylogarithmic factors.}
which nearly matches the $\widetilde{O}(L d\lambda)$ complexity of standard QSVT in terms of $d$, where $\lambda=\sum_k \|H_k\|$.
Moreover, the prefactor $\lambda_{\mathrm{comm}}$, which scales with the norm of the sum of nested commutators of $H_k$, can be significantly smaller than $\lambda$ for a wide range of physical systems, leading to a better dependence than standard QSVT. Thus, our method offers a practical and resource-efficient alternative for applying polynomial transformations in quantum algorithms, while retaining near-optimal asymptotic complexity.

To achieve our goals, we consider a highly general quantum circuit architecture: any circuit composed of an interleaved sequence of arbitrary unitaries and Hamiltonian evolution operators. This encompasses generalized quantum signal processing (GQSP) \cite{motlagh2024generalized} as well as QSVT, where the unitaries are single-qubit rotations, but extends well beyond them. 
Our central insight is that quantum algorithms based on such circuits can achieve near-optimal complexity even when Hamiltonian evolution operators are implemented via basic higher-order Trotter-Suzuki formulas, arguably the simplest Hamiltonian simulation technique. 
This is surprising, as Trotterization is typically considered suboptimal due to its poor scaling with inverse precision, which can increase the circuit depth required to achieve high accuracy. 

The key technical contribution enabling this result is a novel analytic framework, in which we prove that the expectation value of any observable, evaluated on the final quantum state prepared by such a circuit, can be expressed as a power series in the Trotter step sizes. 
This structure enables the use of classical extrapolation \cite{low2019well} to systematically eliminate Trotter errors and extrapolate to the zero-step-size limit, thereby recovering near-optimal asymptotic circuit depth. 
Importantly, our method is compatible with two optimal strategies for estimating the desired expectation value: (i) incoherent estimation: measure the observable in each run using a total of $1/\eps^2$ independent classical repetitions of the circuit; and (ii) coherent estimation: apply the iterative amplitude estimation technique of \cite{grinko2021iterative} to achieve a $1/\eps$ dependence in the overall circuit depth.\footnote{We use the terms ``incoherent'' and ``coherent'' estimation throughout the article to distinguish between these two methods of expectation value estimation.}

As concrete applications of our framework, we design end-to-end quantum algorithms for two fundamental tasks: (i) solving quantum linear systems and (ii) estimating properties of ground states of Hamiltonians. To the best of our knowledge, these are the first algorithms for either problem that simultaneously achieve:
\begin{itemize}
    \item[1.~]No oracular assumptions on access to the Hamiltonian $H$.
    \item[2.~]Constant ancilla overhead, and
    \item[3.~]Near-optimal end-to-end complexity, in terms of both circuit depth and total gate count.
\end{itemize}
Moreover, both algorithms are simply based on Hamiltonian simulation via Trotterization, along with classical extrapolation. Despite its simplicity, the quantum linear systems algorithm achieves near-linear scaling on the condition number of the underlying operator and polylogarithmic dependence on the inverse precision, nearly matching the complexity of state-of-the-art algorithms -- without relying on Szegedy-type quantum walks \cite{costa2022optimal} or complicated subroutines such as variable-time amplitude amplification \cite{ambainis2012variabletime, childs2017qls}. The ground state property estimation algorithm also achieves Heisenberg-limited scaling, which is optimal.

Overall, our methods provide a unified framework for developing resource-efficient, end-to-end quantum algorithms with near-optimal complexity for a wide range of problems.

\subsection{Main results}
\label{subsec:intro-main-results}
In this section, we summarize the main contributions of this work and provide pointers to the corresponding sections of the manuscript where formal statements and detailed derivations can be found.

\subsubsection{Interleaved sequences of unitaries and Hamiltonian evolution}
\label{subsubsec:interleaved-sequences}
We begin by introducing a novel framework for developing quantum algorithms that encompasses QSVT and also extends to a broader class of problems. This framework enables efficient estimation of expectation values of the form 
$\braket{\psi_0|W^{\dag}OW|\psi_0}$, for any initial state $\ket{\psi_0}$, where $W$ is a structured quantum circuit built from an interleaved sequence of arbitrary unitary operations and Hamiltonian evolutions. Our main result establishes a near-optimal complexity for implementing such circuits with only higher-order Trotterization, eliminating the need for block encoding, while using no ancilla qubits. As mentioned previously, the expectation value can be estimated, either by directly measuring $O$ (incoherent estimation) or by using iterative quantum amplitude estimation (coherent estimation). We now formally state the problem addressed in this work:
\begin{prob} \label{problem:interleaved}
Let \( H^{(1)}, H^{(2)}, \dots, H^{(M)} \) be Hermitian operators of the same dimension, where 
\begin{equation}
\label{eq:hamiltonian-interleaved-sequence}
    H^{(j)} = \sum_{\gamma=1}^{\Gamma_{j}} H^{(j)}_{\gamma},
\end{equation}
is a decomposition of \( H^{(j)} \) into terms \( H^{(j)}_{\gamma} \), such that each \( e^{\i H^{(j)}_{\gamma}} \) can be efficiently implemented (in constant time). Then define the interleaved quantum circuit $W$ as follows:
\begin{equation}
\label{eq:interleaved-sequence-circuit}
W:= V_0\cdot e^{\i H^{(1)}}\cdot V_1 \cdot e^{\i H^{(2)}}\cdots e^{\i H^{(M)}}\cdot V_M = V_0 \prod_{j=1}^M e^{\i H^{(j)}} V_{j},    
\end{equation}
where $\{V_j\}_{j=0}^{M}$ are arbitrary unitaries. Then, for any observable $O$, and initial state $\ket{\psi_0}$, estimate the expectation value $\braket{\psi_0|W^{\dag}OW|\psi_0}$ to within an additive precision of $\eps\|O\|$.
\end{prob}

We can directly replace each $e^{\i H^{(j)}}$ by a $p$-th order staged product formula \cite{childs2021theory} (See Sec.~\ref{subsec:prelim-trotter}). However, this results in a circuit with depth scaling as $1/\eps^{1/p}$, which is sub-optimal. We prove that for any quantum circuit $W$, the precision dependence can be improved to $\polylog(1/\eps)$ with the help of classical extrapolation techniques. In fact, our core technical contribution is to prove that the error between the true expectation value $\braket{\psi_0|W^{\dag}O W|\psi_0}$ and our approximation, which is obtained by replacing each $e^{\i H^{(j)}}$ with a $p$-th order Trotterization, can be expressed succinctly as a power series in the Trotter step sizes. This allows us to extrapolate to the zero-error limit via Richardson extrapolation \cite{low2019well} and subsequently obtain an exponentially better scaling.

Our result can be seen as a non-trivial generalization of recently used methods that make use of extrapolation to mitigate errors in particular Hamiltonian simulation algorithms \cite{low2019well, watson2024exponentially, watson2024randomly}. 
Crucially, a direct multiplicative combination of the expansions of error series from prior work \cite{watson2024exponentially} does not yield a controllable bound in our (more general) framework. Instead, we reinterpret the entire interleaved sequence circuit as being governed by a time-dependent (Floquet-like) effective Hamiltonian. This allows us to control the time-dependent error between the true expectation value and our approximation. 

Thus, the overall algorithm uses high-order Trotter methods to implement the Hamiltonian evolution with decreasing Trotter step sizes, estimates the observable at each step size, and extrapolates the measurement results to the zero-step-size limit, corresponding to the true value. The main results are proven in Theorem~\ref{thm:qsp-with-trotter-and-interpolation} and Theorem~\ref{thm:qsp-with-trotter-and-interpolation-coherent}. Here, we state the complexity of our algorithm via the following theorem:

\begin{thm}[Informal version of Thms. \ref{thm:qsp-with-trotter-and-interpolation} and \ref{thm:qsp-with-trotter-and-interpolation-coherent}]
\label{thm:interleaved-sequence-general-result}
Consider Hamiltonians $\{H^{(j)}:j\in [M]\}$ that can be written as in Eq.~\eqref{eq:hamiltonian-interleaved-sequence}. 
Define $\Gamma_{\mathrm{avg}}:=\sum_{j=1}^M\Gamma_j/M$, and suppose $\{V_j\}_{j=0}^M$ are unitaries with circuit depth $\tau_j$ in the circuit $W$ defined in Eq.~\eqref{eq:interleaved-sequence-circuit}. 
Furthermore, let $p \in \mb N$, $\tau_{\rm sum}=\sum_{j=0}^{M} \tau_j$, and $\varepsilon\in(0,1)$ be the precision parameter. 
Then there exists a quantum algorithm that, for any observable $O$ and efficiently preparable initial state $\ket{\psi_0}$, computes $\braket{\psi_0|W^{\dag}OW|\psi_0} \pm \varepsilon\|O\|$.
%with additive error at most $\varepsilon\|O\|$ and constant success probability. 

If $O$ is measured incoherently, the maximum quantum circuit depth is
$$\widetilde{O}\big(\Gamma_{\mathrm{avg}}\big(M \lambda_{\rm comm}\big)^{1+1/p}+\tau_{\rm sum}\big),$$ 
while the time complexity is\footnote{Throughout this article, the \textit{time complexity} of an algorithm refers to the number of elementary (single and two-qubit) gates required to implement it.} 
$$
\widetilde{O}\big({\big(\Gamma_{\mathrm{avg}}(M \lambda_{\mathrm{comm}})^{1+1/p}+\tau_{\rm sum}}\big)/{\varepsilon^2}\big).
$$ 
On the other hand, if the expectation value is coherently estimated, the time complexity and the maximum quantum circuit depth are
$$
\widetilde{O}\big(\big({\Gamma_{\mathrm{avg}}\big(M \lambda_{\mathrm{comm}}\big)^{1+1/p}+\tau_{\rm sum}}\big)/{\varepsilon}\big).
$$
\end{thm}

Here $\lambda_{\rm comm}$, which is formally defined in Eq.~\eqref{Lambda}, scales with the nested commutators of terms $H^{(j)}$, and can be upper bounded as (see Remark \ref{remark: A bound of Lambda})
\begin{equation}
\label{eq:nested-comm-scaling}
    \lambda_{\rm comm} \leq 4\max_{j\in [M]} \sum_{\gamma=1}^{\Gamma_{j}}\|H_{\gamma}^{(j)}\|.
\end{equation}
As in standard Hamiltonian simulation via Trotterization, the dependence of the complexity on the Trotter order $p$ scales as $O(p\cdot 5^p)$. Hence, for sufficiently large choices of $p$ (e.g. $p\sim \log\log(M\lambda_{\rm comm})$), this dependence can be absorbed into the $\widetilde{O}(\cdot)$ notation. Moreover, the dependence of the complexity on $M$ is nearly linear $M^{1+o(1)}$. 

Observe that there is a trade-off between the maximum circuit depth per coherent run and the total complexity. 
If $O$ is directly measured (incoherent estimation), we need independent classical runs of a shorter-depth circuit, which may be preferable in the intermediate term. The total time complexity, however, scales as $1/\varepsilon^2$. On the other hand, iterative quantum amplitude estimation \cite{grinko2021iterative} can also be applied to estimate the expectation value (coherent estimation). This quadratically reduces the dependence on precision, at the cost of increasing the circuit depth. However, if $O$ is, in general, non-unitary, quantum amplitude estimation would need access to a block encoding of $O$, which might involve additional ancilla qubits. Unless otherwise specified, for the results we state here, coherent estimation is always possible, leading to linear scaling in inverse precision. 

To see why circuit $W$ is extremely general, observe that there are no restrictions on the unitaries $V_0,\ldots, V_M$ --- each could represent an entirely different quantum algorithm. As such, any quantum procedure that invokes Hamiltonian evolution as a subroutine, even with different Hamiltonians at each step, can be cast as a special case of our interleaved circuit architecture. Our findings thus significantly expand the scope of classical extrapolation techniques and their potential to reduce circuit depth in generic quantum algorithms.

Indeed, we shall soon show that GQSP \cite{motlagh2024generalized} is a particular instance of the circuit in Eq.~\eqref{eq:interleaved-sequence-circuit}, with each $V_j$ being a single-qubit phase rotation, and each $H^{(j)}=H$. 
Incorporating time-evolution operators of different Hamiltonians $\{H^{(j)}\}$ enables us to construct a framework that goes beyond QSP or even QSVT. This becomes crucial for significantly simplifying recently developed optimal quantum linear systems algorithms inspired by adiabatic quantum computing \cite{subacsi2019quantum, costa2022optimal}. 
Unlike in \cite{costa2022optimal}, where a block encoding of the underlying matrix is used to implement Szegedy's quantum walk, we obtain a near-linear dependence on the condition number by directly simulating the Hamiltonians along the discretized adiabatic path via a high-order Trotter method. This is followed by a projection step implemented using QSVT. The entire procedure fits naturally within our interleaved circuit framework, making it a specific instance of $W$. Besides the applications considered here, simulation of time-dependent Hamiltonians \cite{Berry2020timedependent, watkins2024time-dependent}, the recently introduced framework of linear combination of Hamiltonian simulation \cite{an2023linear,an2023quantum,an2024laplace}, and variational algorithms \cite{cerezo2021variational} can also be seen as particular instances of $W$: it is possible to use our framework to develop simple resource-efficient quantum algorithms with near-optimal complexity, using only basic Trotterization, no oracular assumptions, and minimal ancilla overhead.

\subsubsection{QSVT using higher-order Trotterization}
We begin by stating the problem we intend to solve.

\begin{prob} \label{problem}
Consider a Hermitian operator $H$ that is a sum of local operators, expressed as 
\be \label{decom of H and lambda}
H=\sum_{j=1}^{L} H_j, \quad \text{~with~} \lambda=\sum_{j=1}^{L}\|H_j\|.
\ee
For any observable $O$ and initial state $\ket{\psi_0}$, accurately estimate the expectation value 
\be
\braket{\psi_0|f(H)^{\dag} O f(H)|\psi_0}
\ee
%with high probability, 
for any function $f(x)$ that is approximated by a bounded polynomial of degree $d$.
\end{prob}

As discussed earlier, QSVT provides a unified framework for obtaining such an estimate. 
The standard approach \cite{gilyen2019quantum, martyn2021grand} achieves this by assuming access to the unitary (which is known as a block encoding of $H$)
$$
U_H=\begin{bmatrix}
    H/\lambda & *\\
    * & *
\end{bmatrix},
$$
which encodes $H/\lambda$ in its top-left block. 
From the polynomial that approximates $f(x)$, one can find a sequence of single-qubit phase rotations $V_j=e^{\i \phi_j Z}$, such that an interleaved product of $V_j$ and $U_H$ is close to a block encoding of $f(H/\lambda)$. Typically, the complexity is measured in terms of the number of queries made to $U_H$. However, as argued before, constructing $U_H$ is expensive, which affects the implementability of QSVT. For instance, when $H_j$ are weighted unitaries, the block encoding $U_H$ can be constructed using the Linear Combination of Unitaries (LCU) technique \cite{childs2012hamiltonian}. Such a construction uses $O(\log L)$ ancilla qubits,
involves complicated controlled logical operations, and results in a quantum circuit of depth $O(L)$.\footnote{When only the Hermitian decomposition of $H$ is given, the LCU technique cannot be used directly, and the construction of block encodings of $H$ is more difficult.} Then QSVT outputs an $\eps\|O\|$-accurate estimate of the desired expectation value, using: (i) $O(1/\eps^2)$ repetitions of a quantum circuit of depth $\widetilde{O}(\lambda d L)$, for incoherent estimation, and (ii) a circuit depth of $\widetilde{O}(\lambda d L/\eps)$, if the expectation value is coherently measured using amplitude estimation. 

One might argue that frugal constructions of block encodings for such Hamiltonians $H$ are possible. After all, LCU is merely one method to construct $U_H$. Unfortunately, in Appendix \ref{sec-app:ancilla-lb-lcu}, we prove that $\Omega(\log L)$ ancilla qubits are necessary to construct such a block encoding, for a quite general class of quantum circuits, including LCU as a special case. Formally, we prove:
\begin{thm}[Lower bound of ancillas for LCU]
Assume that $H=\sum_{j=1}^L \lambda_j P_j$ is a linear combination of unitaries. Let 
\bes
U = (V_1 \otimes I) c_0\text{-}P_1 (V_2 \otimes I) c_0\text{-}P_2 \cdots 
(V_L \otimes I) c_0\text{-}P_L (V_{L+1} \otimes I),
\ees
where $V_1,\ldots,V_{L+1}$ act on some $m$ ancilla qubits and $c_0\text{-}P_j$ are controlled unitaries with control qubit 0.
If $U$ is an exact block encoding of $H$, that is,
$$
\left(\bra{0^{\otimes m}}\otimes I\right) U \left(\ket{0^{\otimes m}}\otimes I\right)=H,
$$
then the number of ancilla qubits, $m$, satisfies $m=\Omega(\log L)$.
\end{thm}
The substantial ancilla overhead required to block encode $H$ underscores why standard QSVT is impractical for near-term quantum computers. So, to implement QSVT with minimal ancilla usage, there are broadly two potential strategies. The first is to construct an approximate block encoding of $H$ using $O(1)$ ancilla qubits, while ensuring that the gate complexity of such a block encoding scales polylogarithmically with the dimension of $H$ and inverse precision. However, to the best of our knowledge, no such constructions are currently known. The second approach, taken in this work, is to bypass block encoding altogether and develop alternative methods for implementing QSVT without relying on block encodings.

To this end, we consider a formulation of QSVT as an interleaved sequence of single-qubit phase rotations and controlled Hamiltonian evolution operators, a recently introduced framework known as generalized quantum signal processing (GQSP) \cite{motlagh2024generalized, wang2023quantum}. 
Given access to the Hamiltonian evolution operator $e^{\i H}$, this method implements polynomial transformations of $e^{\i H}$ by an interleaved sequence of general single-qubit rotations and controlled-$e^{\i H}$ operations. 
The advantage of GQSP is that, unlike the standard QSVT, this does not treat real and complex polynomials separately. 
Overall, if the end-to-end complexity of the method in \cite{motlagh2024generalized} is considered, the optimal circuit depth (matching standard QSVT) is attained by assuming a block encoding access to $H$, involving $O(\log L)$ ancilla qubits. We show that this framework is also a particular instance of the interleaved unitaries circuit of Eq.~\eqref{eq:interleaved-sequence-circuit}, and develop an end-to-end quantum algorithm (without block encoding) for estimating $\braket{\psi_0|f(H)^{\dag}Of(H)|\psi_0}$ while using only a single ancilla qubit. 

A brief review of GQSP is provided below, with additional details in Sec. \ref{subsec:prelim-gqsp}. Let 
\begin{align*}
    R(\theta, \phi, \gamma)=\begin{bmatrix}
        e^{\i (\gamma+\phi)} \cos (\theta) & e^{\i  \phi} \sin (\theta) \\
        e^{\i  \gamma} \sin (\theta) & -\cos (\theta)
         \end{bmatrix} \otimes I 
\end{align*}
correspond to a single-qubit rotation of the ancilla qubit, parametrized by $(\theta,\phi, \gamma)$. 
Motlagh and Wiebe \cite{motlagh2024generalized} showed that there exists an interleaved sequence of $R(\theta_j, \phi_j, 0)$ and $c_0\text{-}e^{\i H}$ (controlled on $0$), $c_1\text{-}e^{-\i H}$ (controlled on $1$) of length $2d+1$, that can implement a block encoding of degree-$d$ Laurent polynomials (polynomials with both positive and negative powers) $P(e^{\i H})$ such that, $|P(x)|\le 1$ for all $x\in\mb T=\{ z\in \mathbb{C} : |z|=1 \}$. 
This framework allows for implementing the block encoding of any $f(H)$ that can be approximated by $P(e^{\i H})$. 

However, instead of assuming access to a block encoding of $H$, we can directly implement $c_0\text{-}e^{\i H}$ and $c_1\text{-}e^{-\i H}$, obtaining a circuit that is an instance of $W$ considered in Sec.~\ref{subsubsec:interleaved-sequences}. 
Indeed, a simple way to see this is that we can rewrite the controlled Hamiltonian evolution as $e^{\i  \tilde{H}}$ with $\tilde{H}={\rm diag}(H,0)$ or ${\rm diag}(0,-H)$ depending on the control qubit being $0$ or $1$. That is, if $H$ is an $n$-qubit Hamiltonian, then there exist $\{\theta_j\}_{j=0}^{2d}, \{\phi_j\}_{j=0}^{2d} \in \mathbb{R}^{2d+1}$, and $\gamma \in \mathbb{R}$ such that: 
\[
    \begin{bmatrix}
    P(e^{\i H}) & * \\
    * & *
    \end{bmatrix}=\Big(\prod_{j=1}^d R\left(\theta_{d+j}, \phi_{d+j}, 0\right) c_1\text{-}e^{-\i H} \Big)\Big(\prod_{j=1}^{d} R\left(\theta_j, \phi_j, 0\right) c_0\text{-}e^{\i H} \Big) R\left(\theta_0, \phi_0, \gamma\right),
\]
where the LHS is now an $(n+1)$-qubit unitary. 
Thus, all we need is to implement $e^{\i  \tilde{H}}$ using $p$-th order Trotterization and mitigate the errors using Richardson extrapolation to obtain an improved dependence on precision. 
In order to obtain a specific instance of the circuit in Eq.~\eqref{eq:interleaved-sequence-circuit}, we simply set each $V_j=R(\theta_j,\phi_j,\gamma_j)\otimes I$,  and $H^{(j)}=\tilde{H}$, with $\Gamma_{\rm avg}=L$. Moreover, each $V_j$ can now be implemented in constant time, and $M$ is replaced by the degree $d$ of the polynomial. 

This framework is already general enough to incorporate many problems that can be solved by standard QSVT. 
In fact, we significantly expand its applicability by rigorously finding the conditions under which any function $f(x)$ can be approximated by Laurent polynomials in $e^{\i x}$ of degree $d$. 
We prove that any function $f(x)$ that can be approximated in the interval $[-1,1]$ by a polynomial that remains bounded in an extended interval $[-1-\delta, 1+\delta]$ (e.g., $\delta\in (0,1/2]$), can also be approximated by an appropriate Laurent polynomial in $e^{\i x}$ for $x\in [-1,1]$. 
Overall, we obtain the following theorem, proven in Sec.~\ref{sec:qsvt-trotter}:

\begin{thm}[QSVT with Trotterization, informal version of Thm. \ref{thm_gqsp-trotter-suzuki}]
\label{thm:hsvt-gqsp-trotter}
Let $\eps\in (0,1)$, $p\in\mathbb{N}$, $B\in (1, \pi]$, and $H$ be a Hermitian operator that is a sum of $L$ local terms (or a linear combination of strings of Pauli operators) such that $\|H\|\le B$. 
Suppose that \( f(x) \) can be approximated by a Laurent polynomial in $e^{\i x}$, $P(e^{\i x})$, on $[-B,B]$ satisfying $|P(y)|\le 1$ for all $y\in \mb T=\{ z\in \mathbb{C} : |z|=1 \}$, of degree at most $d$, with additive precision $\eps/6$. 
Furthermore, let $\ket{\psi_0}$ be an efficiently preparable initial state and $O$ be an observable. 
Then there exists an algorithm that estimates $ \braket{\psi_0|f(H)^{\dag}Of(H)|\psi_0}$ with an additive accuracy $\eps\|O\|$ and constant success probability. 
The algorithm uses only a single ancilla qubit and has an overall time complexity 
\be
\widetilde{O} \left( L \big(d \lambda_{\rm comm}\big)^{1+o(1)}/\varepsilon^2 \right)
\ee
with the maximum quantum circuit depth
\be
\widetilde{O}\left( L \big(d \lambda_{\rm comm}\big)^{1+o(1)} \right).
\ee
Here, $\lambda_{\rm comm}$ scales with the nested commutators of the local terms of $H$ and can be upper bounded by $\lambda$, as given in Eqs.~\eqref{eq:nested-comm-scaling} and \eqref{decom of H and lambda}.
\end{thm}
The condition $\|H\| \le B$ ensures that the spectrum of $H$ lies within a single period of $P(e^{\i x})$, but this does not limit the scope of the method. For any general Hamiltonian, we can apply our theorem to its rescaled operator $H' = H/\|H\|$, approximating the function $f(x\|H\|)$. Although the approximation degree $d$ of $f(x\|H\|)$ scales linearly with $\|H\|$ in general, this is canceled by the inverse scaling of the $\lambda_{\rm comm}$ term for $H'$. As a result, the overall complexity, governed by the product $d \lambda_{\rm comm}$, remains independent of the operator norm and instead depends only on the commutator structure of $H$. This ensures that our approach remains efficient even for Hamiltonians with  $\|H\| >\pi$.

\begin{table}[ht!!]
\begin{center}
    \resizebox{\columnwidth}{!}{
    \renewcommand{\arraystretch}{1.5} 
    \begin{tabular}{cccc}
    \hline
    Algorithm & Ancilla & Circuit depth per coherent run & Classical repetitions \\ \hline\hline
     
  Standard QSVT \cite{gilyen2019quantum} & $\lceil \log_2 L \rceil+1 $ & $\widetilde{O}\left(Ld\lambda\right)$ & $\widetilde{O}\left(\varepsilon^{-2}\right)$ \\  
 
  QSVT with Trotterization (Thm.~\ref{thm:hsvt-gqsp-trotter}) & 1 & $\widetilde{O}\left(L(d\lambda_{\mathrm{comm}})^{1+o(1)}\right)$ & $\widetilde{O}\left(\varepsilon^{-2}\right)$ \\ 
 \hline
\end{tabular}}
  \caption{\small{Comparison of the complexity of standard QSVT with our algorithm for solving Problem \ref{problem}. 
\label{table:comparison-qsvt}}}
\end{center}
\end{table}

Theorem \ref{thm:hsvt-gqsp-trotter} provides a general framework to implement polynomial transformations (irrespective of whether it is real or complex) of a Hermitian operator, without using block encodings, while using only a single ancilla qubit. Table \ref{table:comparison-qsvt} compares the complexity of Theorem \ref{thm:hsvt-gqsp-trotter} with standard QSVT with block encoding.
Note that despite requiring only a single ancilla qubit and without assuming any oracular access to $H$, we obtain near-optimal time complexity. If $H$ is non-Hermitian, Theorem \ref{thm:hsvt-gqsp-trotter} would use an additional ancilla qubit. 
In some specific cases, it is possible to keep the total number of ancilla qubits to one, even when $H$ is non-Hermitian. 
This can be achieved by implementing complex polynomials using Hamiltonian Singular Value Transformation (HSVT) \cite{dong2022ground, lloyd2021hamiltonian}. 
While this framework is also subsumed by the circuit in Eq.~\eqref{eq:interleaved-sequence-circuit} and does not require implementing controlled Hamiltonian simulations, finding appropriate complex polynomials is challenging. 
For certain odd real polynomials, HSVT can be implemented using only one ancilla qubit by employing symmetric phase factors  \cite{wang2022energy}. 
However, in general, implementing HSVT with real polynomials uses two ancilla qubits, even when \( H \) is Hermitian.\footnote{See the arXiv version of this paper \cite{chakraborty2025quantum} for a more detailed discussion.}

Finally, in several applications, we are often interested in finding the expectation value of $O$ with respect to the normalized quantum state output by QSVT. We can define the problem formally as:
\begin{prob} \label{problem:normalized-qsvt}
Consider a Hermitian operator $H$ that is a sum of local operators, $H=\sum_{j=1}^{L} H_j$ with $\lambda=\sum_{j=1}^{L}\|H_j\|.$
For any observable $O$ and initial state $\ket{\psi_0}$, estimate the expectation value 
$$
\braket{\psi|O|\psi}\pm \eps\|O\|,
$$
%with high probability, 
where
\begin{equation}
    \label{eq:normalized-qsvt-state}
    \ket{\psi}=\dfrac{f(H)\ket{\psi_0}}{\|f(H)\ket{\psi_0}\|},
\end{equation}
for any function $f(x)$ that is approximated by a bounded polynomial of degree $d$.
\end{prob}
In such a case, Theorem \ref{thm:hsvt-gqsp-trotter} can be easily adapted to estimate $\braket{\psi|O|\psi}$. 
If $\|f(H)\ket{\psi_0}\|\geq \eta$, there are three ways in which this expectation value in Problem \ref{problem:normalized-qsvt} can be estimated, just like in standard QSVT, with the circuit depth and the overall time complexity also scaling similarly. We list them here, with the detailed costs summarized in Table \ref{table:comparison-our-qsvt-methods}:
\begin{itemize}
    \item First, it is possible to directly measure $O$ incoherently as before. Applying Hoeffding's inequality, one can show that the expectation value with respect to the normalized state can be estimated with $\widetilde{O}(\eps^{-2}\eta^{-2})$ classical repetitions.

    \item Second, it is possible to adapt the interleaved sequence $W$ to incorporate quantum amplitude amplification. 
For an initial state $\ket{\psi_0}$, this leads to a modified interleaved sequence circuit:
\begin{equation}
\label{eq:interleaved-sequence-qaa}
    \widetilde{W}=\prod_{j=1}^{K/2} W e^{\i \phi_j\ket{0,\psi_0}\bra{0,\psi_0}} W^{\dag} e^{\i \theta_j (\ket{0}\bra{0}\otimes I)},
\end{equation}

where $\phi_j,\theta_j\in [0,2\pi]$ are specified by the fixed-point amplitude amplification algorithm \cite{yoder2014fixed} and $K=O(\eta^{-1}\log(1/\eps))$. The unitary $e^{\i \phi_j\ket{0,\psi_0}\bra{0,\psi_0}}$ helps implementing a reflection about the initial state $\ket{\psi_0}$, controlled on a single ancilla qubit. The new circuit $\widetilde{W}$ is still an instance of Eq.~\eqref{eq:interleaved-sequence-circuit}, and we can again use Richardson extrapolation to mitigate the errors at the end. The circuit depth of each run, as well as the total running time, now increases by a factor of $\eta^{-1-1/p}$ as compared to Theorem \ref{thm:hsvt-gqsp-trotter}. 

\item Finally, it is also possible to estimate $\braket{\psi|O|\psi}$ using iterative quantum amplitude estimation \cite{grinko2021iterative}. The time complexity scales as $\widetilde{O}(L(d\lambda_{\rm comm}/\eta)^{1+o(1)}/\eps)$, but requires block encoding access to $O$.
\end{itemize}
\begin{table}[ht!!]
\begin{center}
    \resizebox{\columnwidth}{!}{
    \renewcommand{\arraystretch}{1.5} 
    \begin{tabular}{c|ccccc}
    \hline
    Algorithm & Variant & Ancilla & \vtop{\hbox{\strut ~~Circuit depth} \hbox{\strut per coherent run}} & Classical repetitions \\ \hline\hline
    
    \multirow{3}{*}{\large{Standard QSVT \cite{gilyen2019quantum}}}  & Without QAA or QAE & $O\left(\log L\right)$ & $\widetilde{O}\left(L d\lambda\right)$ & $\widetilde{O}\left({\varepsilon^{-2}\eta^{-2}}\right)$\\
    
    & QAA & $O(\log L)$ & $\widetilde{O}\left(L d\lambda/\eta\right)$ & $\widetilde{O}\left({\varepsilon^{-2}}\right)$ \\ 
    
   & QAE & $O\left(\log L\right)$ & $\widetilde{O}\left({L d\lambda}/{\eta\eps}\right)$ & -- \\
     \hline 
   
 \multirow{3}{*}{\large{This work}} & Without QAA or QAE & $1$ & $\widetilde{O}\left(L(d\lambda_{\mathrm{comm}})^{1+o(1)}\right)$ & $\widetilde{O}\left(\varepsilon^{-2}\eta^{-2}\right)$ \\
 
& QAA & 2 & $\widetilde{O}\left(L(d\lambda_{\mathrm{comm}}/\eta)^{1+o(1)}\right)$ & $\widetilde{O}\left(\varepsilon^{-2}\right)$ \\ 

 & Iterative QAE & 3 & $\widetilde{O}\left(L(d\lambda_{\mathrm{comm}}/\eta)^{1+o(1)}\cdot\eps^{-1}\right)$ & --\\
     \hline 
    \end{tabular}}
 \caption{\small{Comparison between standard QSVT and our method to solve Problem \ref{problem:normalized-qsvt}. The table compares the number of ancilla qubits, the circuit depth per coherent run, and the total number of classical repetitions. There are three ways the expectation value can be estimated, both for standard QSVT and our method. The first approach measures $O$ directly at the end of each run. The second technique uses quantum amplitude amplification (QAA) to prepare the normalized state, followed by incoherent measurements. Note that our method 
 runs an interleaved sequence circuit $\widetilde{W}$ which incorporates fixed point quantum amplitude amplification (See Eq.~\eqref{eq:interleaved-sequence-qaa}) but still measures $O$ at the end of each run incoherently. The third approach is to measure the desired expectation value using quantum amplitude estimation (QAE). In our method, we use the iterative quantum amplitude estimation scheme of \cite{grinko2021iterative}.
\label{table:comparison-our-qsvt-methods}}}
    \end{center}
\end{table}
\renewcommand{\arraystretch}{1}

Note that for all three methods in Table \ref{table:comparison-our-qsvt-methods}, we provide end-to-end complexities for estimating the required expectation value. Both the circuit depth per coherent run as well as the total time complexity nearly match the optimal complexity of standard QSVT, despite requiring significantly fewer resources.

\subsubsection{Applications}
\label{subsubsec:applications}

We apply our general algorithms for implementing QSVT without block encodings, to two problems of practical interest: (i) solving quantum linear systems, and (ii) ground state property estimation.
~\\~\\
\textbf{Quantum linear systems without block encodings:}
The quantum linear systems algorithm has been quite fundamental for developing quantum algorithms for regression \cite{chakraborty_et_al:LIPIcs.ICALP.2019.33, chakraborty2023quantum}, machine learning, and solving differential equations \cite{ childs2020quantum, liu2021efficient, Krovi2023improvedquantum}. Suppose $A\in \mathbb{C}^{N\times N}$ is a matrix with $\|A\|=1$, with singular values in the interval $[-1,-1/\kappa]\cup[1/\kappa, 1]$, where $\|A^{-1}\|=\kappa$ is the condition number. 
Given a procedure to efficiently prepare the quantum state $\ket{b}$, the goal is to estimate the expectation value $\braket{x|O|x}$ with $\eps$-additive accuracy, where $\ket{x}=A^{-1}\ket{b}/\|A^{-1}\ket{b}\|$. 
The complexity of this algorithm has been improved over the years, with broadly two approaches (both of which assume access to a block encoding $U_A$ of matrix $A$): 
The first one makes use of quantum linear algebra techniques such as Linear Combination of Unitaries \cite{childs2017qls, chakraborty_et_al:LIPIcs.ICALP.2019.33}, QSVT \cite{gilyen2019quantum, chakraborty2023quantum, low2024qls}, along with a complicated technique known as variable-time amplitude amplification (VTAA) \cite{ambainis2012variabletime} to obtain a complexity that depends linearly on $\kappa$ and logarithmically in $\kappa/\eps$. 

Recently, an approach that is quite distinct from these has gained prominence, as it avoids using VTAA. 
These methods \cite{subacsi2019quantum, an2022quantum, lin2020optimal} are inspired by adiabatic quantum computation: they consider a time-dependent Hamiltonian $H(s)$ (which is related to $A$), with $s\in [0,1]$. 
The algorithm starts from an easy-to-prepare $0$-eigenstate of $H(0)$ and evolves adiabatically to the $0$-eigenstate of $H(1)$, which encodes $\ket{x}$. 
The algorithm of Costa et al. \cite{costa2022optimal} achieves an optimal query complexity of $O(\kappa\log(1/\eps))$, in terms of the number of queries to the block encoding of $A$. 
The underlying idea is to construct a quantum walk operator $Q(s)$ (from a block encoding of $A$) for different values of $s\in \{k/T\}_{k=0}^{T-1}$, to discretize the adiabatic computation. 
They show that the discretized adiabatic evolution by $Q(s)$ approximates the ideal adiabatic evolution under $H(s)$ (up to a constant) for $T=O(\kappa)$, and the resulting state at this point has a constant overlap with $\ket{x}$. 
The optimal dependence on precision is attained by employing the eigenstate filtering technique of \cite{lin2020optimal}.

\begin{table}[ht!!]
\begin{center}
    \resizebox{\columnwidth}{!}{
    \renewcommand{\arraystretch}{1.5} 
    \begin{tabular}{cccc}
    \hline
    Algorithm & Ancilla & Circuit depth per coherent run & Classical repetitions \\ \hline\hline
     
  State-of-the-art \cite{costa2022optimal} & $\lceil \log_2 L \rceil+6 $ & $\widetilde{O}\left(L\lambda\kappa\right)$ & $\widetilde{O}\left(\varepsilon^{-2}\right)$ \\  
 Near-term methods \cite{wang2024qubit, chakraborty2024implementing} & $1$ & $\widetilde{O}(\lambda^2\kappa^2)$ & $\widetilde{O}\left(\kappa^4\eps^{-2}\right)$\\
 
  This work (Thm.~\ref{thm-intro:quantum linear system}) & 4 & $\widetilde{O}\left(L(\lambda\kappa)^{1+o(1)}\right)$ & $\widetilde{O}\left(\varepsilon^{-2}\right)$ \\ 
 \hline
\end{tabular}}
  \caption{\small{Comparison of the complexity of quantum linear systems algorithms. Our complexity is near-optimal while requiring fewer ancilla qubits.
\label{table1:comparison-qls}}}
\end{center}
\end{table}

Despite achieving optimal query complexity, block encoding of $A$ can be expensive in practice. Suppose $A$ is expressed as a linear combination of $L$ Pauli operators of $n$-qubits, i.e.\ $H=\sum_{j}\lambda_jP_j$, where $P_j\in \{I, X, Y, Z\}^{\otimes n}$, and $\lambda=\sum_{j}|\lambda_j|$. The circuit depth of the state-of-the-art algorithm \cite{costa2022optimal}, to estimate $\braket{x|O|x}$ to $\varepsilon\|O\|$-additive accuracy is $\widetilde{O}(L\lambda\kappa)$, while making $\widetilde{O}(1/\eps^2)$ classical repetitions. Moreover, the method uses $\lceil\log_2 L\rceil+6$ ancilla qubits in all. On the other hand, quantum algorithms tailored to early fault-tolerant quantum computers have recently been developed \cite{wang2024qubit, chakraborty2024implementing}. These randomized algorithms are hardware friendly (requiring only a single ancilla qubit), but are sub-optimal with a total time complexity scaling as $\widetilde{O}(\lambda^2\kappa^6/\eps^2)$. 

We present a fully end-to-end algorithm that entirely eliminates the need for (i) block encoding of the Hamiltonian, and (ii) the construction of quantum walk unitaries from the interpolating Hamiltonian $H(s)$. At the same time, (iii) our approach achieves near-optimal complexity while significantly reducing the number of ancilla qubits required (to a constant). The adiabatic Hamiltonian $H(s)$, constructed from $A$ and $U_b$ (the unitary preparing the state $\ket{b}$), is an $n+2$ qubit operator and also admits a natural decomposition as a linear combination of Pauli operators and low rank projectors. Subsequently, we proceed by simulating the time-dependent Hamiltonian $H(s)$ at discrete values of $s$, i.e.,
$$
B(s)=\prod_{m=0}^{sT-1} e^{\i H(m/T)}.
$$
This naturally fits into the interleaved circuit structure of Eq.~\eqref{eq:interleaved-sequence-circuit}, where each $H^{(j)}=H(j/T)$, and $V_j=I$. We can simulate each time-evolution operator using higher-order Trotterization. 
This works fine because the eigenvalue gap of the quantum walk unitary is asymptotically the same as the eigenvalue gap of $e^{\i H(s)}$. 

The operator $B(s)$ is then combined with a second interleaved circuit, that implements a trigonometric polynomial (Laurent polynomial in $e^{\i x}$) carefully chosen to approximate the filtering function introduced in \cite{lin2020optimal}. This filtering procedure effectively projects the evolved state onto the zero-eigenstate of $H(1)$, and can again be implemented using Theorem~\ref{thm:hsvt-gqsp-trotter}.
Altogether the entire algorithm becomes a particular instance of the interleaved unitary sequence $W$, followed by the extrapolation of the measurement results to the zero-error limit. This gives us the following theorem:

\begin{thm}[Quantum linear systems without block encodings, informal version of Thm. \ref{thm:qls-our-method}]
\label{thm-intro:quantum linear system}
    Suppose $A=\sum_{j=1}^{L}\lambda_j P_j$ is a matrix with $\|A\|=1$, with singular values in $[-1,-1/\kappa]\cup [1/\kappa,1]$ and $\|A^{-1}\|=\kappa$. Let $\varepsilon\in(0,1)$ be the precision parameter. Then, there exists an algorithm that, for any observable $O$, estimates $\braket{x|O|x}$ to within an additive error of at most $\varepsilon\|O\|$, with a constant success probability, using only four ancilla qubits. Moreover, the algorithm has total time complexity of $\widetilde{O}(L(\lambda\kappa)^{1+o(1)}/\varepsilon^2)$, and a maximum circuit depth $\widetilde{O}(L(\lambda\kappa)^{1+o(1)})$.
\end{thm}
The maximum circuit depth is $\widetilde{O}(L(\lambda\kappa)^{1+o(1)})$, which is nearly optimal: (i) the dependence on the condition number $\kappa$ is quasi-linear, and (ii) we also have a polylogarithmic dependence on inverse precision using only a few ancilla qubits, and just higher order Trotterization. As before, the desired expectation value can also be estimated using iterative quantum amplitude estimation with a quadratically reduced dependence on precision, resulting in a time complexity $\widetilde{O}(L(\lambda\kappa)^{1+o(1)}/\eps)$. Table \ref{table1:comparison-qls} compares our methods with  \cite{costa2022optimal, wang2024qubit, chakraborty2024implementing} in detail. More analysis is given in Sec. \ref{sec:applications}.
~\\~\\
\noindent
\textbf{Ground state property estimation:} Consider a Hamiltonian $H$ which is a sum of $L$ local terms, with unknown ground state $\ket{v_0}$. The goal is to estimate the expectation value $\braket{v_0|O|v_0}$, with arbitrary additive accuracy, provided we have (i) prior knowledge of the spectral gap $\Delta$ of $H$, and (ii) access to an initial (guess) state $\ket{\phi_0}$ with an overlap of at least $\gamma$ with the ground state. A particular instance of this problem is the estimation of the ground state energy $\xi_0$ of $H$. 
Both of these problems are highly relevant and have applications in many areas of physics. Almost all prior work on developing quantum algorithms for these problems (both for near-term quantum devices and fully fault-tolerant quantum computers), assume some form of oracular access to $H$ \cite{lin2020nearoptimalground, lin2022heisenberg, dong2022ground, zhang2022computingground}.

The state-of-the-art quantum algorithm by Lin and Tong \cite{lin2020nearoptimalground}, solves this problem while assuming access to a block encoding of $H$. The overall time complexity (using quantum amplitude estimation) is
$$
\widetilde{O}\left(\dfrac{L\lambda}{\Delta\gamma\eps}\right).
$$
However, the algorithm requires $O(\log L)$ ancilla qubits and sophisticated multi-qubit controlled gates, making it impractical to implement on near-term or early fault-tolerant quantum devices due to its significant resource overhead.

As a result, several quantum algorithms have been proposed that are specifically designed for early fault-tolerant quantum devices. Notably, recent randomized techniques \cite{wang2024qubit, chakraborty2024implementing} assume that the Hamiltonian $H$ is expressed as a linear combination of strings of Pauli operators, defined as $H=\sum_{j}\lambda_j P_j$, where now $\lambda=\sum_{j}|\lambda_j|$. These approaches require only a single ancilla qubit, making them hardware-friendly, but suffer from sub-optimal complexity scaling:
$$
\widetilde{O}\left(\dfrac{\lambda^2}{\Delta^2\gamma^4\eps^2}\right).
$$
In fact, any method that relies on sampling the Pauli terms $P_j$, according to their weight $\lambda_j/\lambda$, is fundamentally limited to worse 
scaling \cite{chakraborty2025quantum}. 

Another line of recent work \cite{lin2022heisenberg, dong2022ground, zhang2022computingground} has developed near-term quantum algorithms assuming oracular access to the Hamiltonian evolution operator $e^{\i H}$. For instance, \cite{dong2022ground} applies Hamiltonian Singular Value Transformation (HSVT) to achieve near-optimal query complexity using only a single ancilla qubit. Their circuit consists of an interleaved sequence of single-qubit rotations and Hamiltonian evolutions—precisely a special case of our general circuit structure $W$. However, in practice, implementing $e^{\i H}$ via higher-order Trotterization incurs a cost scaling as $1/\eps^{1+o(1)}$, which degrades the ideal Heisenberg scaling guaranteed under exact oracular access.

Building on this framework, we remove the oracle assumption and instead construct $e^{\i H}$ explicitly using higher-order Trotterization. Moreover, we apply Richardson extrapolation to effectively reduce the Trotter error. Thus, our general setting allows us to develop end-to-end quantum algorithms for this problem without assuming any oracular access and using only $O(1)$ ancilla qubits. We show that the shifted sign function can be well-approximated by a Laurent polynomial in $e^{\i x}$, as presented in Lemma \ref{lem:approx-shift-sign}, which allows us to use Theorem \ref{thm:hsvt-gqsp-trotter} to estimate $\braket{v_0|O|v_0}$ to the desired precision. 

Our algorithm achieves an exponential reduction in circuit depth compared to the incoherent ground state estimation approach of \cite{dong2022ground}. Furthermore, by leveraging iterative quantum amplitude estimation, we preserve Heisenberg scaling in precision, yielding a polynomial speedup over even the coherent variant of their algorithm. Remarkably, both our total runtime and maximum circuit depth closely approach the optimal complexity established by Lin and Tong \cite{lin2020nearoptimalground}. We summarize the result as follows, with the detailed proofs in Sec.~\ref{subsec:app-ground-state}:

\begin{thm}[Ground state property estimation, informal version of 
Thm. \ref{thm_ground state property estimation using HSVT with Trotter}]
\label{thm-intro:ground state}
Given a Hamiltonian $H=\sum_{j=1}^L H_j$, whose spectrum is contained in $[-1, 1]$. Suppose there exists a procedure to efficiently prepare the state $\left|\phi_0\right\rangle$ such that $\left|\braket{\phi_0|v_0} \right| \geq \gamma$ for $\gamma\in (0,1)$, and there is a spectral gap $\Delta$ separating the ground state energy from the rest of the spectrum. Let $\varepsilon\in(0,1)$ be the precision parameter. Then there exists a quantum algorithm that estimates $\braket{v_0|O|v_0}$ within an additive accuracy of $\eps\|O\|$, using only two ancilla qubits and with constant success probability. The algorithm has a maximum circuit depth of 
$
    \widetilde{O}( L({\lambda_{\rm comm}}/{\Delta\gamma})^{1+o(1)} )
$
and total time complexity of
    $
    \widetilde{O}(\eps^{-2} L({\lambda_{\rm comm}}/{\Delta\gamma})^{1+o(1)} ).
    $ If we have coherent access to the observable $O$, then we can reduce the time complexity to  $
    \widetilde{O}(\eps^{-1} L({\lambda_{\rm comm}}/{\Delta\gamma})^{1+o(1)} )
    $ using one additional ancilla qubit.
\end{thm}

Yet another advantage of using higher-order Trotterization is that for many physical Hamiltonians, $\lambda_{\rm comm}\ll \lambda$, which offers room for even more practical advantages. This includes broad classes of Hamiltonians, such as $k$-local Hamiltonians and the Heisenberg model, Hamiltonians with power law interactions, electronic structure Hamiltonians, and cluster Hamiltonians \cite{childs2021theory, aftab2024multi}. 
Despite removing the oracular assumption, the resulting quantum algorithms achieve asymptotic complexity comparable to the best-known results. 

\begin{table}[h!]
\centering
\resizebox{1.0\columnwidth}{!}{
\renewcommand{\arraystretch}{1.5} 
\begin{tabular}{c | c c c c} 
 \hline
 Algorithm & Variant & Ancilla  & Circuit depth per coherent run & Classical repetitions \\
 \hline\hline
 Lin and Tong~\cite{lin2020nearoptimalground}  & QAE & $\lceil \log_2L\rceil+4$ & $\widetilde{O}(L \lambda \Delta^{-1}\gamma^{-1}\eps^{-1})$ & $-$ \\ \hline
  Dong et al.~\cite{dong2022ground}  & \multirow{2}{*}{QAA} & 2 & $\widetilde{O}\big(L \big(\lambda_{\rm comm}' \Delta^{-1}\gamma^{-1}\big)^{1+o(1)}\varepsilon^{-o(1)}\big) $ & $\widetilde{O}(\varepsilon^{-2})$ \\ 
 This work (Thm.~\ref{thm-intro:ground state})  & & 2 & $ \widetilde{O}(L \big(\lambda_{\rm comm} \Delta^{-1}\gamma^{-1}\big)^{1+o(1)} ) $ & $\widetilde{O}(\varepsilon^{-2})$ \\ \hline
Dong et al.~\cite{dong2022ground}  & \multirow{2}{*}{Iterative QAE} & 3 & $\widetilde{O}\big(L \big(\lambda_{\rm comm}' \Delta^{-1}\gamma^{-1}\big)^{1+o(1)}\varepsilon^{-1-o(1)}\big) $ & $-$ \\
 This work (Thm.~\ref{thm-intro:ground state}) & & 3 & $ \widetilde{O}(L \big(\lambda_{\rm comm} \Delta^{-1}\gamma^{-1}\big)^{1+o(1)}\eps^{-1}) $ & $-$ \\  
 \hline
\end{tabular}
}
\caption{Comparison of different algorithms for ground state property estimation.}
\label{table:ground state problem}
\end{table}

Finally, in Table \ref{table:ground state problem} we compare our results with previous state-of-the-art results:  
\begin{itemize}
    \item The state-of-the-art algorithm of Lin and Tong \cite{lin2020nearoptimalground} obtains optimal complexity at the cost of using several ancilla qubits, and sophisticated multi-qubit controlled operators needed to construct the block encoding of $H$. In comparison, our method (using iterative QAE) attain Heisenberg scaling of $1/\eps$, and nearly matches the state-of-the-art time complexity. 
    
    \item In order to estimate the desired expectation value, one possibility is to measure $O$ directly, after each run of the algorithm. As compared to this variant of the Dong et al.~\cite{dong2022ground} algorithm, the circuit depth of our algorithm is also reduced by a factor of $\varepsilon^{-o(1)}$, yielding an exponential speedup. The incoherent algorithms achieve shorter quantum circuits and one fewer ancilla qubit. Although the overall cost of these algorithms is higher, they are more suitable for intermediate quantum devices.

    \item We can estimate the expectation value using iterative quantum amplitude estimation. As compared to Dong et al.~\cite{dong2022ground}, the use of extrapolation, results in a shorter circuit depth (by a factor of $\varepsilon^{-o(1)}$), yielding a polynomial speedup.
    
    \item Another advantage of our method over \cite{lin2020nearoptimalground} is that $\lambda_{\rm comm}\ll \lambda$. Note that this prefactor for \cite{dong2022ground}, $\lambda_{\rm comm}'$, is different from $\lambda_{\rm comm}$ but scales similarly. See Sec.~\ref{subsec:gspe-algorithm} for details.
\end{itemize} 

\subsection{Comparison with prior work}
%providing a systematic framework for designing and analyzing quantum algorithms 
Since QSVT unifies a wide range of quantum techniques \cite{martyn2021grand}, it has been central to the development of many fundamental algorithms, including Hamiltonian simulation \cite{low2017qsp, low2019hamiltonian}, quantum linear system solvers \cite{gilyen2019quantum, lin2020optimal, chakraborty2023quantum}, ground state preparation \cite{lin2020nearoptimalground}, quantum semidefinite program solvers \cite{brandao2017SDP, vanapeldoorn2019improvedsdp}, and quantum walks \cite{gilyen2019quantum, apers2021unified, chakraborty2024implementing}.

Recent efforts have focused on designing quantum algorithms that minimize resource requirements, motivated by the practical limitations of emerging small-scale quantum hardware—such as constrained circuit depth and limited availability of ancilla qubits \cite{cerezo2021variational, katabarwa2024early}. Many of these near-term algorithms are either heuristic in nature \cite{cerezo2021variational}, lacking rigorous guarantees, or tailored to specific problems \cite{lin2022heisenberg, zhang2022computingground}. A promising development in this direction is the randomized LCU framework introduced in \cite{chakraborty2024implementing, wang2024qubit}, which leverages a single ancilla qubit to construct resource-efficient quantum algorithms for tasks such as quantum linear system solving and ground state property estimation. However, these approaches still suffer from sub-optimal circuit depth and overall complexity. 

In contrast, our approach provides the first fully end-to-end implementation of quantum singular value transformation (QSVT) that eliminates oracular assumptions and requires only a single ancilla qubit, while retaining near-optimal asymptotic complexity. As a result, when applied to problems of practical interest, such as solving quantum linear systems or estimating ground state properties, our framework yields quantum algorithms that are provably efficient, hardware-friendly, and practically scalable on early fault-tolerant quantum computers.

Error mitigation via extrapolation has been extensively explored in the context of Hamiltonian simulation, particularly for algorithms such as Trotterization \cite{watson2024exponentially}, qDRIFT \cite{watson2024randomly}, and multi-product formulas \cite{low2019well, aftab2024multi}. These techniques are inspired by earlier strategies developed for suppressing errors in near-term quantum circuits \cite{endo2018practical, cai2023mitigation}. Recently, classical post-processing techniques such as Chebyshev interpolation have been proposed to reduce Trotter errors more efficiently \cite{rendon2024improved, rendon2024needtrotter, watson2024exponentially}. In particular, \cite{rendon2024needtrotter} examined the mitigation of Trotter errors in circuits composed of interleaved unitaries and Hamiltonian evolutions. To avoid fractional queries to Hamiltonian evolution, they applied separate cardinal sine interpolation to each segment, which results in sub-optimal scaling for sequence lengths $M \geq 2$. 

In contrast, our work advances extrapolation-based error mitigation by integrating it as a fundamental design principle for constructing near-optimal quantum algorithms, tailored to both near-term and fully fault-tolerant quantum architectures, and significantly expanding the scope and efficacy of extrapolation beyond prior techniques.

\section{Preliminaries}
\label{sec:prelim}
In this section, we first clarify the notation used throughout the article. Next, we provide a concise overview of the foundational concepts pertinent to this work. Each subsection is self-contained and does not depend on others, allowing readers to navigate the content flexibly. For those already familiar with specific topics, some subsections may be skipped without losing continuity. Our goal is to provide a thorough yet accessible overview, ensuring that all necessary background information is available to support the main results of this paper.

\subsection{Notation}
\label{sec:notation}

\textbf{Sets:~}We use $\mb N$ to denote the set of positive integers $\{1,2,\cdots\}$ and $\mathbb{T}= \{ z\in \mathbb{C} : |z|=1 \}$ to denote the complex unit circle. For any set $S\subseteq\mb R$, we use $S_{\geq p}$ to denote $S \cap [p,\infty)$.
\\~\\
\textbf{Complexity theoretic notation:} Throughout the work, we follow the standard complexity-theoretic notation. In particular, we use $\widetilde{O}(\cdot)$ to hide polylogarithmic factors, namely, $\widetilde{O}(f(n))=O(f(n)\polylog(f(n)))$. 
~\\~\\
\textbf{Norms:~}We use bold lowercase letters to denote vectors. For example, $\bm{a}$ represents the vector $(a_1, \ldots, a_n)$. We use $\|\cdot\|_1$ to denote the vector $\ell_1$-norm, i.e., $\|\bm{a}\|_1 = \sum_{i=1}^n |a_i|$.
Unless otherwise specified, $\|A\|$ will denote the spectral norm of the operator $A$. We denote two operators $A,B$ that are within $\eps$-close in spectral norm, i.e.,\ $\|A-B\|\leq \varepsilon$, by $A\approx_\varepsilon B$. We use $M_n(\mathbb{C})$ to denote the $n \times n$ complex matrices. For $A,B \in M_n(\mathbb{C})$, we define $\ad_B(A)=[B,A]=BA-AB$ and so $e^{\ad_B}(A) = e^B A e^{-B}$. Let $\mc E\colon M_n(\mathbb{C}) \rightarrow M_n(\mathbb{C})$ be a quantum channel, then the diamond norm of $\mc E$ is given by
\[
\|\mc E\|_{\diamond}:=  \max_{\rho: \|\rho\|_1 \leq 1} \left\|\left(\mathcal{E} \otimes I_{n}\right)(\rho)\right\|_1,
\]
where $\rho \in M_{n^2}(\mathbb{C})$ and $\|A\|_1 = \Tr[A^\dag A]$ is the trace norm or Schatten 1-norm for matrix $A$.
Let $f(x):[a,b] \rightarrow  \mathbb{C}$, we denote 
$$\|f\|_{[a,b]} = \max_{x\in[a,b]} |f(x)|.
$$  
~\\
\textbf{Trace, Expectation and Probability:~}The trace of an operator $A$ will be denoted by $\Tr[A]$. The expectation value of a random variable $V$ will be denoted by $\mathbb{E}[V]$. The probability of an event $E$ occurring will be denoted by $\Pr[E]$.
\\~\\
\textbf{Pauli matrices:~}The Pauli matrices are
\[
I = \begin{bmatrix}
    1 & 0 \\ 0 & 1
\end{bmatrix}, \quad
X = \begin{bmatrix}
     0 & 1 \\ 1 & 0 
\end{bmatrix}, \quad
Y = \begin{bmatrix}
     0 & -\i \\ \i & 0 
\end{bmatrix}, \quad
Z = \begin{bmatrix}
    1 & 0 \\ 0 & -1
\end{bmatrix}.
\]
Any matrix $A$ of dimension $2^n$ can be uniquely written as a linear combination of tensor products of Pauli matrices, namely, 
\[
A = \sum_{P\in\{I,X,Y,Z\}^{\otimes n}}\lambda_P P,
\qquad 
\lambda_P = \frac{1}{2^n} \Tr[A^\dag P].
\] 
If $A$ is Hermitian, then $\lambda_P\in\mathbb{R}$.
~\\~\\
\textbf{Logarithm of a unitary matrix:~}For any unitary matrix $U$, there exists a unique skew-Hermitian matrix $V$ with eigenvalues in $ (-\i\pi, \i\pi]$ such that $e^{V}=U$, and we denote $V$ by $\log(U)$.
~\\~\\
\textbf{Computational model:~}Throughout this article, by the \emph{time complexity} of a quantum algorithm, we refer to the number of elementary gates used to implement the algorithm.

\subsection{Brief overview of quantum singular value transformation}
\label{subsec-prelim:blk-encoding-qsvt}

In a nutshell, quantum singular value transformation (QSVT) applies polynomial transformations to the singular values of any given operator $H$ that is embedded in the top-left block of a unitary $U_H$. Before delving into QSVT, let us briefly define block encoding, which was introduced in a series of works \cite{low2019hamiltonian, chakraborty_et_al:LIPIcs.ICALP.2019.33, gilyen2019quantum}. Formally, we define a block encoding as follows:

\begin{defn}[Block encoding~\cite{chakraborty_et_al:LIPIcs.ICALP.2019.33}]
    \label{def:block_encoding}
Suppose that $H$ is an $s$-qubit operator, $\alpha \in \mathbb{R}^+$ and $a \in \mathbb{N} $, then we say that the $(s + a)$-qubit unitary $U_H$ is an $(\alpha, a, 0)$-block encoding of $H$, if

\begin{equation} \label{eq:block-encoding}
 H = \alpha (\bra{0}^{\otimes a} \otimes I)U_H(\ket{0}^{\otimes a} \otimes I). 
\end{equation}
That is,
$$
U_H=\begin{bmatrix}
   H/\alpha  & * \\ * & *
\end{bmatrix}.
$$
Usually, errors are allowed in \eqref{eq:block-encoding}.
\end{defn}

\begin{defn}[Singular value transformation]
\label{def_svt of matrix}
Let $f\colon \mb{R} \rightarrow \mathbb{C}$ be an even or odd polynomial. Let $A \in \mathbb{C}^{m \times n}$ be a matrix with singular value decomposition (SVD) $A=\sum_{i=1}^{n} \sigma_i\ket{u_i}\left\langle v_i\right|$ with $\sigma_i:= 0$ when $i> \min(m,n)$. We define singular value transformation corresponding to $f$ as
\[
f^{(\rm SV)}(A):= \begin{cases}  \vspace{.2cm}
\sum_{i=1}^{n}  f\left(\sigma_i\right) \ket{u_i} \bra{v_i} & \text { if } f \text{ is odd, }\\ 
\sum_{i=1}^{n}  f\left(\sigma_i\right) \ket{v_i} \bra{v_i} & \text { if } f \text{ is even. }
\end{cases}
\]
\end{defn}

Note that $ \sqrt{A^\dag A} = \sum_i \sigma_i \ket{v_i} \bra{v_i}$. So if $f(x)$ is even, then $f^{(\rm SV)}(A)=f(\sqrt{A^\dag A})$. If $f(x)=x \cdot f_1(x)$ is odd for some even polynomial $f_1(x)$, then $f^{(\rm SV)}(A)=A \cdot f_1(\sqrt{A^\dag A})$. If $f(x)$ is a polynomial with definite parity and $A$ is Hermitian, then $f^{(\rm SV)}(A)=f(A)$ is the standard polynomial of a matrix.

\begin{defn}[Definition 8 of \cite{gilyen2019quantum}] 
\label{def_alternating phase sequence}
    Let $\mathcal{H}_U$ be a finite dimensional Hilbert space and $U, \Pi, \widetilde{\Pi} \in \operatorname{End}\left(\mathcal{H}_U\right)$ be linear operators such that $U$ is unitary and $\Pi, \widetilde{\Pi}$ are orthogonal projectors. Let $\Phi=(\phi_1,\ldots,\phi_d)  \in \mathbb{R}^d$, we define the alternating phase modulation sequences $U_\Phi$ as follows
    \begin{equation}\label{eq_alternating phase circuit}
    U_{\Phi}:= \begin{cases}  \vspace{.2cm}
    e^{\i  \phi_1(2 \widetilde{\Pi}-I)} U 
    \prod_{j=1}^{(d-1) / 2} \left(e^{\i  \phi_{2 j}(2 \Pi-I)} U^{\dagger} e^{\i  \phi_{2 j+1}(2 \widetilde{\Pi}-I)} U\right), & \text { if } d \text { is odd, } \\
    \prod_{j=1}^{d / 2}\left(e^{\i  \phi_{2 j-1}(2 \Pi-I)} U^{\dagger} e^{\i  \phi_{2 j}(2 \widetilde{\Pi}-I)} U\right), 
    & \text { if } d \text { is even.}
    \end{cases}
    \end{equation}
\end{defn}

With the above preliminaries, we can now state a key result of QSVT. The following theorem provides a nice connection between singular value transformation and the alternating phase modulation sequences $U_\Phi$.

\begin{thm}[Theorem 10 of \cite{gilyen2019quantum}]
\label{thm for standard QSVT}
Assume that $f(x)\in \mathbb{C}[x]$ is an even or odd function of degree $d$ and $|f(x)| \leq 1$ for all $x\in[-1,1]$, then there exists $\Phi\in \mathbb{R}^{d}$ such that
\[
f^{(\rm SV)}(\widetilde{\Pi} U \Pi) = 
\begin{cases}
\widetilde{\Pi} U_\Phi \Pi & \text{if } d \text{ is odd}, \\
\Pi U_\Phi \Pi & \text{if } d \text{ is even}.
\end{cases}
\]
\end{thm}

Intuitively, if we set
$$
\Pi=\widetilde{\Pi}=\begin{bmatrix}
    I & 0 \\ 0 & 0
\end{bmatrix},
$$ 
then Theorem \ref{thm for standard QSVT} states that 
$$
U_\Phi = \begin{bmatrix}
   f^{(\rm SV)}(H/\alpha)  & * \\ * & *
\end{bmatrix}, \quad
\text{if } U=\begin{bmatrix}
   H/\alpha  & * \\ * & *
\end{bmatrix}.
$$
That is, if $U$ an $(\alpha, 1, 0)$-block encoding of $H$, then QSVT provides a generic framework to implement $f^{(\rm SV)}(H/\alpha)$, by making $\Theta(d)$ queries to $U$ and $U^{\dag}$, while using a single ancilla qubit. 

If $f(x)$ is real, then using the linear combination of unitaries (LCU) technique, we can still have a formula for $f^{(\rm SV)}(\widetilde{\Pi} U \Pi)$, see \cite[Corollary 11]{gilyen2019quantum}. Now, we need to introduce one more ancilla qubit because of LCU. 
Moreover, if $f(x)$ is a general polynomial (without a definite parity), we can decompose it into even and odd parts and apply these results to each part. 
Consequently, we can still obtain a concise formula for $f^{(\rm SV)}(\widetilde{\Pi} U \Pi)$. In the most general case, two additional ancilla qubits are used.

Despite its generality, a drawback of this framework is that it assumes access to a block encoding $U$ of the underlying operator $H$. The construction of a block encoding itself can be expensive and resource-consuming, which adds to the overhead of implementing any QSVT. For example, consider the case that $H$ is Hermitian and can be expressed as a linear combination of strings of $L$ Pauli operators, i.e.,
$$
H=\sum_{j=1}^{L} \lambda_j P_j, \quad \lambda_j>0
$$
where $P_j\in \{I,X,Y,Z\}^{\otimes n}$. Denote $\lambda=\sum_{j}\lambda_j$. Then the block encoding of such an $H$ can be constructed by using the LCU \cite{berry2015simulating}. This proceeds as follows: Set $\ell=\lceil \log_2L\rceil$. Define the following two unitaries:
$$
\texttt{PREP}~\ket{0}^{\otimes \ell}=\sum_{j=1}^{L}\sqrt{\lambda_{j}/\lambda} \, \ket{j}, \quad 
\text{and}\quad 
\texttt{SEL}_P=\sum_{j=1}^{L}\ket{j}\bra{j}\otimes P_j.
$$
Then, it is easy to show that the block encoding of $H$ is 
$$
U_H= \Big(\texttt{PREP}^{\dag}\otimes I \Big) \texttt{SEL}_P \Big(\texttt{PREP}\otimes I\Big).
$$
Indeed,
$$
\left(\bra{0}^{\otimes \ell}\otimes I\right)U_H\left(\ket{0}^{\otimes \ell}\otimes I\right)=H/\lambda,
$$
which is a $(\lambda, \lceil \log_2L\rceil, 0)$-block encoding of $H$. Observe that building $U_H$ typically uses $O(\log L)$ ancilla qubits, and a circuit depth of $L$ (the depth of the unitary $\texttt{SEL}_P$). Indeed, in the Appendix (Sec. \ref{sec-app:ancilla-lb-lcu}), we show that $\Omega(\log L)$ ancilla qubits are required to exactly construct a block encoding $H$ for a quite general circuit model including the technique of LCU.
Moreover, constructing $\texttt{SEL}_P$ can be challenging as it involves a sequence of complicated controlled operations. 

As a result, for any such Hamiltonian, from Theorem \ref{thm for standard QSVT}, QSVT implements $f^{({\rm SV})}(H/\lambda)$ with a circuit of depth $\widetilde{O}(Ld)$, and $O(\log L)$ ancilla qubits. Despite its unifying feature, this raises concerns regarding the applicability of QSVT for small-scale quantum computers. The same cost also holds when $H$ is a $k$-local Hamiltonian, written as $H=\sum_{j=1}^L H_j$. In this case, $\lambda=\sum_{j=1}^L \|H_j\|$. This appears in many physically relevant problems. 
In Sec.~\ref{sec:qsvt-trotter}, we propose algorithms for implementing QSVT without relying on block encoding access, while (nearly) preserving its optimality and using only a constant number of ancilla qubits.

\subsection{Generalized quantum signal processing} 
\label{subsec:prelim-gqsp}

Recall that in standard QSVT, the polynomials have to be either complex even or complex odd. While real polynomials can also be implemented, additional ancilla qubits are used. In the recently introduced framework of Generalized Quantum Signal Processing (GQSP) \cite{motlagh2024generalized}, if we are allowed to apply the controlled time-evolution operator, then we can implement any polynomial, even Laurent polynomials of the time-evolution operator, without introducing more ancilla qubits. Next, we outline some of their key results used in this work and refer the readers to \cite{motlagh2024generalized} for details. 

Suppose we can implement the controlled Hamiltonian simulation $U=e^{\i H}$ (and  $U^{\dagger}=e^{-\i H}$) as a black box. That is, the operators 
\be \label{control U:eq}
    c_0\text{-}U
    =\left[\begin{array}{cc}
        U & 0 \\
        0 & I
        \end{array}\right], \quad 
    c_1\text{-}U^\dag 
    =\left[\begin{array}{cc}
            I & 0 \\
            0 & U^{\dagger}
            \end{array}\right]
\ee
using one query to $U$ and one query to $U^{\dagger}$. Let 
\begin{align*}
    R(\theta, \phi, \lambda)=\left[\begin{array}{cc}
        e^{\i (\lambda+\phi)} \cos (\theta) & e^{\i  \phi} \sin (\theta) \\
        e^{\i  \lambda} \sin (\theta) & -\cos (\theta)
        \end{array}\right] \otimes I 
\end{align*}
be arbitrary U(2) rotation of the single ancilla qubit.  Motlagh and Wiebe \cite{motlagh2024generalized} showed that for any degree-$d$ polynomial $P$, satisfying $|P(x)|\le 1$ in $\mathbb{T}:= \{ x\in \mathbb{C} : |x|=1 \}$, there exists an interleaved sequence of $R(\theta_j, \phi_j, 0)$ and $c_0\text{-}U$, $c_1\text{-}U^\dag$, of length $2d+1$, that implements a block encoding of $P(U)$.  Moreover, their framework also holds for Laurent polynomials bounded on $\mathbb{T}$. We state this result as follows:

% \begin{lem}[Corollary 5 of \cite{motlagh2024generalized}]
%     For any $P \in \mathbb{C}[x]$ with $\operatorname{deg}(P)=d$, if
%     $|P(x)| \leq 1 $ for all $ x \in \mathbb{T}$,
%     then there exist $\Theta=(\theta_0,\ldots,\theta_d), 
%     \Phi=(\phi_0,\ldots,\phi_d) \in \mathbb{R}^{d+1}$ and $\lambda \in \mathbb{R}$ such that:
%     $$
%     \left[\begin{array}{cc}
%     P(U) & \cdot \\
%     \cdot & \cdot
%     \end{array}\right]=\left(\prod_{j=1}^d R\left(\theta_j, \phi_j, 0\right) c_0\text{-}U \right) 
%     R\left(\theta_0, \phi_0, \lambda\right) .
%     $$
% \end{lem}

% Their framework can also be extended to implement Laurent polynomials bounded on $\mathbb{T}$.

% \begin{lem}[Theorem 6 of \cite{motlagh2024generalized}]
% \label{lem:GQSP-for-laurent}
%     For any $\tilde{P}(x) = \sum_{j=-k}^{d-k} a_j x^j \in \mathbb{C}[x,x^{-1}]$, if
%     $|\tilde{P}(x)| \leq 1 $ for all $ x \in \mathbb{T}$, then there exist $\Theta=(\theta_0,\ldots,\theta_d), 
%     \Phi=(\phi_0,\ldots,\phi_d) \in \mathbb{R}^{d+1}$ and $\lambda \in \mathbb{R}$ such that
%     $$
%     \left[\begin{array}{ll}
%     \tilde{P}(U) & \cdot \\
%     \cdot & \cdot
%     \end{array}\right]=\left(\prod_{j=1}^k R\left(\theta_{d-k+j}, \phi_{d-k+j}, 0\right) c_1\text{-}U^\dag \right)\left(\prod_{j=1}^{d-k} R\left(\theta_j, \phi_j, 0\right) c_0\text{-}U \right) R\left(\theta_0, \phi_0, \lambda\right).
%     $$
% \end{lem}

\begin{thm}[Combining Corollary 5 and Theorem 6 of \cite{motlagh2024generalized}]
    \label{thm:generalized-QSP}
    For any Laurent polynomial $P(z) = \sum_{j=-d}^{d} a_{j} z^j$ such that $|P(z)| \le 1$ for all $z\in \mathbb{T}$, there exist $\Theta=(\theta_j)_j, \Phi=(\phi_j)_j \in \mathbb{R}^{2d+1}, \lambda \in \mathbb{R}$ such that 
    \begin{equation}
    \label{eq:gqsp}
        \begin{bmatrix}
            P(U) & * \\
            * & * 
        \end{bmatrix} = \Big(\prod_{j=1}^d R(\theta_{d+j}, \phi_{d+j}, 0) c_1\text{-}U^\dag \Big)\Big(\prod_{j=1}^{d} R(\theta_j, \phi_j, 0) c_0\text{-}U \Big) R(\theta_0, \phi_0, \lambda).
    \end{equation}
\end{thm}
In this work, we concretely provide the conditions under which a function $f(x)$ can be approximated by a Laurent polynomial in $e^{\i x}$ (See Sec.~\ref{sec:qsvt-trotter}). Given any $H$, this allows for implementing any such $f(H)$. 
%However, implementing GQSP relies on block encodings. 
We observe that GQSP is a particular instance of an interleaved sequence of unitaries and Hamiltonian evolution operators. So, it suffices to implement the controlled Hamiltonian evolution operators by higher-order Trotter. This allows us to develop an end-to-end quantum algorithm for QSVT, that is near optimal while using only a single ancilla qubit.

\subsection{Hamiltonian simulation by Trotter methods}
\label{subsec:prelim-trotter}
Simulating the dynamics of physical systems is widely considered to be one of the foremost applications of a quantum computer. Over the years, several quantum algorithms have been developed for this problem, leading to near-optimal complexities \cite{berry2015simulating, low2017hamiltonian, low2017qsp, low2019hamiltonian}. Trotter methods or product formulas \cite{lloyd1996universal}, inspired by the original proposal by Feynman for simulating physical systems \cite{feynman1982simulating}, are arguably the simplest techniques to implement $e^{\i tH}$. These methods have seen remarkable progress over the years, are easy to implement and demonstrate superior performance in practice \cite{childs2018toward,childs2021theory}.

Here, we introduce some results from the theory of product formulas \cite{childs2021theory} which we will be using in the subsequent sections. Let 
$$
H = \sum_{\gamma=1}^{\Gamma} H_{\gamma}
$$
be a Hamiltonian expressed as a sum of $\Gamma$ Hermitian terms. Then, a \emph{staged product formula} \cite{watson2024exponentially,childs2021theory} is an approximation of the unitary evolution operator \( e^{-\i H t} \) by a product of the form
\be \label{staged product formula}
\mathcal{P}(t):= \prod_{\nu=1}^{\Upsilon} \prod_{\gamma=1}^{\Gamma} e^{ \i t a_{(\nu, \gamma)} H_{\pi_\nu(\gamma)}},
\ee
where $\pi_\nu \in S_\Gamma$ is permutation and $a_{(\nu, \gamma)} \in \mb R$. 
A \emph{staged product formula} is of $p${\em -th order} if $e^{-\i H t} = \mathcal{P}(t)+O(t^{p+1})$, and {\em symmetric} if $\mathcal{P}(-t)=\mathcal{P}^{-1}(t)$. 

A well-known example is the $2k$-th order Trotter-Suzuki formula $S_{2k}(t)$ \cite{suzuki1990fractal}, which is recursively defined as
\be
    S_2(t):= \prod_{\gamma=\Gamma}^1 e^{-\i H_\gamma t / 2} \prod_{\gamma=1}^{\Gamma} e^{-\i H_\gamma t / 2} 
\ee
and 
\be \label{S2k}
    S_{2 k}(t):= [S_{2(k-1)}(u_k t)]^2 S_{2(k-1)}((1-4 u_k) t)[S_{2(k-1)}(u_k t)]^2,
\ee
where $u_k=1/(4-4^{1/(2k-1)})$.  The $2k$-th order Trotter-Suzuki formula is symmetric with $\Upsilon=2\cdot 5^{k-1}$ and $|a_{(\nu, \gamma)}|\le 2k/3^k$, as stated in \cite{wiebe2010higher}.  Define
\begin{equation*}
    \alpha_{\mathrm{comm}}^{(j)}:= \sum_{\gamma_1 ,\gamma_2, \ldots ,\gamma_j=1}^{\Gamma}\big\|\left[H_{\gamma_1} H_{\gamma_2} \ldots H_{\gamma_j}\right]\big\|
\end{equation*}
for $j\in \mathbb{N}$,
where $\left[X_1 X_2 \ldots X_n\right]$ refers to a right-nested $n$-commutator as
\begin{equation*}
    \left[X_1 X_2 \ldots X_n\right]:= \left[X_1,\left[X_2,\left[\ldots\left[X_{n-1}, X_n\right] \ldots\right]\right]\right].
\end{equation*} 
Childs et al. \cite{childs2021theory} showed that
$$
\|e^{-\i tH}-\mathcal{P}(t)\|=O\left(\alpha^{(p+1)}_{\rm comm}t^{p+1}\right).
$$
Thus, the complexity (circuit depth) of $p$-th order Trotter formula to approximate $e^{-\i tH}$ to an accuracy of $\eps$, with commutator scaling is given by
$$
O\left(\dfrac{\Gamma(\alpha^{(p+1)}_{\rm comm})^{1/p} t^{1+1/p}}{\eps^{1/p}}\right).
$$
As discussed in \cite{childs2021theory}, the prefactor $\alpha_{\rm comm}$ scales significantly better than the worst-case bound of $O(\|H\|_1)$ for many physical Hamiltonians. We take advantage of this scaling in our QSVT algorithm developed in Sec.~\ref{sec:qsvt-trotter}. 

Recently, it has been shown that the circuit depth can be exponentially improved to $O(\log (1/\eps))$ by utilizing extrapolation techniques \cite{watson2024exponentially}. Therein, it was proven that for sufficiently small $t$, $\mathcal{P}(t)$ can be written as an exponential of $-\i H$ plus a power series in $t$, using Baker-Campbell-Hausdorff (BCH) formula based on right-nested commutators \cite{arnal2021note}.
\begin{lem}[Lemma 2 of \cite{watson2024exponentially}]
    \label{lem:expansion}
    Let $\mathcal{P}(t)$ be a staged product formula \eqref{staged product formula} and $a_{\max}=\max_{\nu,\gamma}|a_{\nu,\gamma}|$. For any $t\in \mathbb{R}$, if there exists $J\in \mathbb{N}$ and $c \in \mathbb{R}^+$ such that
    \be \label{BCH convergence condition}
        \sup _{j \geq J} \alpha_{\mathrm{comm}}^{(j)} \cdot (a_{\max } \Upsilon|t|)^j \leq c, 
    \ee
    then $\mathcal{P}(t)$ can be written as
    %the exponential to an effective Hamiltonian 
    \begin{align*}
        \mathcal{P}(t)=e^{-\i t(H+\sum_{j=1}^{\infty} E_{j+1} t^j)},
    \end{align*}
    such that $\|E_{j}\|\le (a_{\max} \Upsilon)^j \alpha_{\mathrm{comm}}^{(j)}/ j^2$ for all $j$.
\end{lem}  
% \begin{lem}[Lemma 3 of \cite{watson2024exponentially}]
%     Let $\mathcal{P}(t)$ be a $p$-th order staged product formula, and suppose that condition \eqref{BCH convergence condition} holds. Then $E_{j+1}=0$ for all $j<p$. Moreover, for symmetric $\mathcal{P}(t)$, $E_{j+1}=0$ for all odd $j$.
% \end{lem} 
Note that condition \eqref{BCH convergence condition} ensures that the nested commutators do not grow too rapidly, which guarantees the convergence of the BCH formula. Using 
% the error series in 
Lemma~\ref{lem:expansion}, \cite{watson2024exponentially} constructed an error series of an operator evolved under a product formula. More precisely, for $s\in (0,1)$ and observable $O$, let 
$\tilde{O}_{T, s} := 
(\mathcal{P}^{\dagger}(s T))^{1/s}  O (\mathcal{P}(s T))^{1/s}$.
% be the observable evolved under the product formula $\mathcal{P}(t)$ for duration $T\in \mathbb{R}$ with step size $sT$ and let 
%and $\tilde{O}_{T,0} = \lim_{s\rightarrow 0} \tilde{O}_{T, s}$.
The next result describes the error between $\tilde{O}_{T, s}$
and the exact evolution $O_{T} := e^{\i  H T} O e^{-\i  H T}$.

\begin{lem}[Lemma 4 of \cite{watson2024exponentially}]
    \label{lem:trotter-power}
    Let $\mathcal{P}$ be a $p$-th order staged product formula. Suppose that condition \eqref{BCH convergence condition} holds for $t=sT$. Let $\sigma=2$ if $\mathcal{P}$ is symmetric, and 1 otherwise. Then, for any $K \in \mathbb{N}$,
    \begin{align*}
    \tilde{O}_{T, s} - O_{T} = \sum_{j \in (\sigma \mathbb{Z})_{ \geq p} } s^j \tilde{E}_{j+1, K}(T) + \tilde{F}_K(T, s),
    \end{align*}
    where $\tilde{E}_{j+1, K}(T)$ and $\tilde{F}_K(T, s)$ are superoperators with induced spectral norm bounded as
    \begin{align*}
        & \|\tilde{E}_{j+1, K}(T)\| \leq(a_{\max } \Upsilon T)^j \sum_{\ell=1}^{\min \{K-1,\lfloor j / p\rfloor\}} \frac{(a_{\max } \Upsilon T)^\ell}{\ell!} \sum_{\substack{j_1, \ldots, j_\ell \in 
        (\sigma \mathbb{Z})_{\geq p} \\
        j_1+\cdots+j_\ell=j}} ~~ 
        \prod_{\kappa=1}^\ell  \frac{2\alpha_{\mathrm{comm}}^{(j_\kappa+1)}}{(j_\kappa+1)^2} , \\
        & \|\tilde{F}_K(T, s)\| \leq \frac{(a_{\max } \Upsilon T)^K}{K!} \sum_{j \in (\sigma \mathbb{Z})_{ \geq K p} }(a_{\max } \Upsilon s T)^j \sum_{\substack{j_1, \ldots, j_K \in (\sigma \mathbb{Z})_{ \geq p} \\
        j_1+\cdots+j_K=j}} ~~ \prod_{\kappa=1}^K \frac{ 2 \alpha_{\mathrm{comm}}^{(j_\kappa+1)}}{(j_\kappa+1)^2} .
        \end{align*}
\end{lem}

Expressing the error terms as a power series in the Trotter step sizes allows us to use Richardson extrapolation (see Sec.~\ref{subsec:prelim-richardson}) to extrapolate to the zero-error limit. This leads to a circuit depth scaling with $\polylog(1/\eps)$ instead of $1/\eps^{1/p}$. 

In Sec.~\ref{sec:interleaved-sequence}, we provide a non-trivial generalization of the result of \cite{watson2024exponentially}, proving that even for very general quantum circuits (interleaved sequence of arbitrary unitaries and Hamiltonian evolution operators), the difference between the aforementioned expectation values can be expressed as a power series in the Trotter step sizes. Thus, this is at the heart of our end-to-end quantum algorithms for QSVT without block encodings.

\subsection{Richardson extrapolation}
\label{subsec:prelim-richardson}

Richardson extrapolation is a method to estimate 
$f(0)=\lim_{x\rightarrow 0} f(x)$ using value $f(x)$ for several values of $x\in\{s_1,s_2,\ldots,s_m\}$, via an appropriate linear combination of $f(s_1), f(s_2), \ldots, f(s_m)$. This method improves the accuracy of numerical solutions by effectively removing leading-order error terms, thereby increasing the order of accuracy without significantly increasing computational effort. Recently, it has been applied to mitigate errors in Hamiltonian simulation using Trotter methods, multi-product formulas, and qDRIFT. Formally, we state the relevant lemma below:

\begin{lem}[Richardson extrapolation, see \cite{watson2024randomly,low2019well}]
    \label{lem_richardson extrapolation}
    Suppose that function $f\colon \mb R \to \mathbb{R}$ has a series expansion
    \begin{align*}
        f(x) = f(0) + \sum_{i=1}^{\infty} c_i x^{i}.
    \end{align*}
    Choose an initial point $s_0\in (0,1)$ and set points $s_i=s_0/r_i$ with integer $r_i \neq 0$ for $i\in [m]$. Then, using these $m$ sample points $\{s_i: i\in[m]\}$, we can derive an $m$-term Richardson extrapolation
    \be
    \label{eq_richardson extrapolation}
        F^{(m)}(s_0)=\sum_{i=1}^m b_i f(s_i)
    \ee
    by choosing appropriate coefficients $b_i$ to cancel the first $m-1$ terms in the series expansion,
    such that
    \bes
        |F^{(m)}(s_0)-f(0)|\le \|\b\|_1 \left|R_{m}(s_0)\right|=O(s_0^{m}).
    \ees
    where $R_{m}(x)$ is a function that only has terms of order $O(x^{m})$ and above.
    In particular, let
    \be \label{eq:def-r}
        r_{i}=\left\lceil\frac{\sqrt{8} m}{\pi \sin (\pi(2 i-1) / 8 m)}\right\rceil^2,
    \qquad 
        b_i=\prod_{\ell \neq i} \frac{1}{1-r_\ell/r_i},
    \ee 
    then $\|\b\|_1=O(\log m), \max_i\{r_i\}=O(m^4)$ and $\max_i \{r_i/r_m\}=O(m^2)$.
 
\end{lem}
\begin{proof}
For a general power series, the solution to find the coefficients $b_i$ is given by the Vandermonde matrix equation
\[
\left[\begin{array}{cccc}
1 & 1 & \ldots & 1 \\
r_1^{-1} & r_2^{-1} & \cdots & r_m^{-1} \\
\vdots & \vdots & \ddots & \vdots \\
r_1^{-m+1} & r_2^{-m+1} & \ldots & r_m^{-m+1}
\end{array}\right]\left[\begin{array}{c}
b_1 \\
b_2 \\
\vdots \\
b_m
\end{array}\right]=\left[\begin{array}{c}
1 \\
0 \\
\vdots \\
0
\end{array}\right] .
\]
For the choice of $\{r_i\}$, it is easy to see that $r_i$ decreases as $k$ increases from $1$ to $m$, which implies that $\max_i\{r_i\}=r_1=O(m^4)$ and $\max_i \{r_i/r_m\}=r_1/r_m=O(m^2)$.
\end{proof}

Lemma \ref{lem_richardson extrapolation} provides a choice of coefficients $b_i$ for general power series. However, for specific series where many $c_i$ vanish, we can choose different coefficients $b_i$ to obtain even better estimates.
In particular, for symmetric Trotter methods, half of the coefficients $c_i$ in the power series are zero, so the following result from \cite{watson2024exponentially} will be useful for error analysis of our general framework.

\begin{lem}[Lemma 5 of \cite{watson2024exponentially}]
\label{lem_Richad for Trotter}
    Let $f(x) \in C^{2 m+2}([-1,1])$ be an even, real-valued function, and let $P_j$ and $R_j$ be the degree-$(j-1)$ Taylor polynomial and Taylor remainder, respectively, such that $f(x)=P_j(x)+R_j(x)$. Let
    \[
    F^{(m)}(s_0)=\sum_{i=1}^m b_i f(s_i)
    \]
    be the Richardson extrapolation of $f(x)$ at points $s_i=s_0/r_i, i\in[m]$ given by
    \be \label{eq:def-r-trotter}
        r_{i}=\left\lceil\frac{\sqrt{8} m}{\pi \sin (\pi(2 i-1) / 8 m)}\right\rceil,
    \qquad 
        b_i=\prod_{\ell \neq i} \frac{1}{1-r_\ell^2/r_i^2},
    \ee 
    Then
    \[
    F^{(m)}(s_0)=f(0)+\sum_{i=1}^m b_i R_{2 m}\left(s_i\right)
    \]
    with error bounded by $\|\b\|_1 \max_{i\in[m]} \left|R_{2m}(s_i)\right|$ and $\|\b\|_1=O(\log m)$.
\end{lem}
In this work, we show that Richardson extrapolation can be applied, almost as a black box, in much more general settings.

\section[Interleaved sequence of unitaries and Hamiltonian evolution operators]{Interleaved unitaries and Hamiltonian evolutions:~A unified framework}
\label{sec:interleaved-sequence}
In this section, we show how quantum circuits composed of interleaved sequences of arbitrary unitaries and Hamiltonian evolutions can be implemented with near-optimal complexity. Our approach replaces each Hamiltonian evolution with a higher-order Trotter–Suzuki approximation, and we demonstrate that the expectation value of any observable with respect to the final state can be expressed as a power series in the Trotter step size. This structure enables the use of classical extrapolation techniques, significantly improving the dependence on precision from $1/\eps^{o(1)}$ to $\polylog(1/\eps)$. As a result, we obtain a resource-efficient, near-optimal solution to Problem \ref{problem:interleaved}. 

Let
$$
H^{(1)}, H^{(2)}, \dots, H^{(M)},
$$
be Hermitian operators, where each 
$$ 
H^{(\ell)} = \sum_{\gamma=1}^{\Gamma_{\ell}} H^{(\ell)}_{\gamma},$$ 
is a decomposition of \( H^{(\ell)} \) into terms \( H^{(\ell)}_{\gamma} \), and the Hamiltonian evolution of each \( H^{(\ell)}_{\gamma} \) can be efficiently implemented. Let \( V_0, V_1, \dots, V_M \) be arbitrary unitary operators. Then, we define a quantum circuit that is an interleaved sequence of unitaries and Hamiltonian evolutions, as follows: 
\be
\label{eq:interleaved operator}
W = V_0 \prod_{\ell=1}^M e^{\i H^{(\ell)}} V_{\ell}.
\ee
We now address Problem \ref{problem:interleaved}, namely the estimation of \( \bra{\psi_0}W^\dag O W \ket{\psi_0}  \) for an arbitrary initial state \( \ket{\psi_0}\) and any observable \( O \), using no ancilla qubits, and with near-optimal complexity. Our algorithm proceeds by applying high-order Trotter methods to implement the Hamiltonian evolution with decreasing Trotter step sizes. At each step-size, we estimate the observable and then apply classical extrapolation on the measurement results to recover the zero-step-size limit, corresponding to the true expectation value. The main results are proven in Theorem~\ref{thm:qsp-with-trotter-and-interpolation} and Theorem~\ref{thm:qsp-with-trotter-and-interpolation-coherent}. To establish the error bound, we first derive an error series for the measurement result in terms of the inverse of the total number of Trotter steps as a non-trivial generalization of Lemma~\ref{lem:trotter-power} (see Lemma \ref{lem:interleaved} below).

Define 
\[
    \alpha_{\mathrm{comm}}^{(j,\ell)} := \sum_{\gamma_1 ,\gamma_2, \ldots ,\gamma_j=1}^{\Gamma} \big\| \,[H^{(\ell)}_{\gamma_1} H^{(\ell)}_{\gamma_2} \ldots H^{(\ell)}_{\gamma_j}] \, \big\|
\]
for \( j \in \mathbb N \) and \begin{align}
\label{eq:def-tilde-alpha}
\tilde{\alpha}_{\mathrm{comm}}^{(j)}:=\max_{\ell \in [M]} \alpha_{\mathrm{comm}}^{(j,\ell)}.
\end{align}
Let \( O \) be an observable, and define 
\be \label{O-2M+1}
    O_{2M+1} := W^{\dagger} O W 
\ee
as the observable evolved under the interleaved operator $W$ defined as in Eq.~\eqref{eq:interleaved operator}.
Suppose \[ \mathcal{P}^{(\ell)}(t) =\prod_{\nu=1}^{\Upsilon} \prod_{\gamma=1}^{\Gamma} e^{ \i t a_{(\nu, \gamma)} H^{(\ell)}_{\pi_\nu(\gamma)}},
\]
is a \( p \)-th order symmetric staged product formula for \( H^{(\ell)} \). For any \( s \in (0,1) \), define
\be \label{tilde-O-2M+1}
    \tilde{O}_{2M+1,s} := \Big( V_0\prod_{\ell=1}^M \mathcal{P}^{(\ell)}(sM)^{1/(sM)}V_{\ell} \Big)^{\dagger} 
    O \Big( V_0\prod_{\ell=1}^M \mathcal{P}^{(\ell)}(sM)^{1/(sM)}V_{\ell} \Big)
\ee
as the observable evolved under the staged product formula approximations with step size \( sM \).

\begin{lem}
\label{lem:interleaved}
Suppose condition \eqref{BCH convergence condition} holds for $\tilde{\alpha}_{\mathrm{comm}}^{(j)}$ with $t=sM$. Then
    \begin{align*}
        \tilde{O}_{2M+1,s}-O_{2M+1}=\sum_{j \in (\sigma \mathbb{Z})_{\ge p}} s^j \tilde{E}_{j+1, K}(M) +\tilde{F}_K(M, s) ,
    \end{align*}
where $\tilde{E}_{j+1, K}(M)$ and $\tilde{F}_K(M, s)$ are superoperators with induced spectral norm bounded as
\be
\begin{aligned}
\label{eq:bound-E-F} 
    \|\tilde{E}_{j+1, K}(M)\| &\leq& (a_{\max } \Upsilon M)^j \sum_{\ell=1}^{\min \{K-1,\lfloor j / p\rfloor\}} \frac{(a_{\max } \Upsilon M)^\ell}{\ell!} \sum_{\substack{j_1, \ldots, j_\ell \in (\sigma \mathbb{Z})_{\ge p} \\
        j_1+\cdots+j_\ell=j}} \prod_{\kappa=1}^\ell  \frac{2\tilde{\alpha}_{\mathrm{comm}}^{(j_\kappa+1)}}{(j_\kappa+1)^2} , \\
    \|\tilde{F}_K(M, s)\| &\leq& \frac{( a_{\max } \Upsilon M)^K}{K!} \sum_{j \in (\sigma \mathbb{Z})_{\ge K p}}(a_{\max } \Upsilon s M)^j \sum_{\substack{j_1, \ldots, j_K \in (\sigma \mathbb{Z})_{\geq p} \\
        j_1+\cdots+j_K=j}}~~ \prod_{\kappa=1}^K  \frac{2\tilde{\alpha}_{\mathrm{comm}}^{(j_\kappa+1)}}{(j_\kappa+1)^2} .
\end{aligned}
\ee
\end{lem}

\begin{proof}
    Let $\delta:= sM$ be the time step size of the staged product formula $\mathcal{P}^{(\ell)}$ and 
    \be
        \label{eq:def-E}
        E^{(\ell)}(\delta):= -\sum_{j\in (2 \mathbb{N})_{\ge p} }^{\infty} E^{(\ell)}_{j+1} \delta^j
    \ee 
    be the effective error Hamiltonian in Lemma~\ref{lem:expansion}, where the norm of $E^{(\ell)}_{j}$ is bounded as \begin{align}
        \|E^{(\ell)}_{j}\|\le \frac{(a_{\max}\Upsilon)^j}{j^2} \alpha_{\mathrm{comm}}^{(j,\ell)}\le \frac{(a_{\max}\Upsilon)^j}{j^2} \tilde{\alpha}_{\mathrm{comm}}^{(j)}.
    \end{align}
    By assumption, $\delta$ satisfies the convergence condition of Lemma~\ref{lem:expansion} for all $\ell$, so we can write $\mathcal{P}^{(\ell)}(\delta)$ as $
    \mathcal{P}^{(\ell)}(\delta)=e^{-\i\delta(H-E^{(\ell)}(\delta))}$. 
    We define the time-dependent Hamiltonian $H(t)$ and $\tilde{H}(t)$ as \begin{align*}
        H(t)=\begin{cases} \i \log(V_{\lfloor t\rfloor/2}) & \text{ if } \lfloor t\rfloor \text{ is even}, \\ -H^{(\lceil t/2\rceil)}  & \text{  if } \lfloor t\rfloor \text{ is odd}, \end{cases}  \quad \tilde{H}(t) =\begin{cases} \i \log(V_{\lfloor t\rfloor/2}) & \text{ if } \lfloor t\rfloor \text{ is even}, \\ -H^{(\lceil t/2\rceil)}+E^{(\lceil t/2\rceil)}(\delta) & \text{  if } \lfloor t\rfloor \text{ is odd}, \end{cases}  
    \end{align*}
for $t\in [0,2M+1]$. Here $\log(V_{\ell})$ is skew-Hermitian as $V_{\ell}$ is unitary. Denote $O_T, \tilde{O}_{T,s}$ as the time evolved observable under the Hamiltonian $H(t)$ and $\tilde{H}(t)$ respectively, which means \be
     \label{eq:def-O}
    \partial_t O_t = \i \ad_{H(t)}(O_t), \quad \partial_t  \tilde{O}_{t,s} = \i \ad_{\tilde{H}(t)}(\tilde{O}_{t,s})
    \ee and $O_0= O$, $\tilde{O}_{0,s}= O$, where $\ad_H(\cdot):= [H,\cdot]$, see Sec. \ref{sec:notation}. For $T = 2M+1$, we have \bes 
    O_{2M+1} = \Big( V_0\prod_{\ell=1}^M e^{\i H^{(\ell)}}V_{\ell} \Big)^{\dagger} O \Big(V_0\prod_{\ell=1}^M e^{\i H^{(\ell)}}V_{\ell} \Big),
    \ees
    and \bes
    \tilde{O}_{2M+1,s} = \Big( V_0\prod_{\ell=1}^M \mathcal{P}^{(\ell)}(sM)^{1/(sM)}V_{\ell} \Big)^{\dagger} 
    O \Big( V_0\prod_{\ell=1}^M \mathcal{P}^{(\ell)}(sM)^{1/(sM)}V_{\ell} \Big),
\ees
    which match the definitions of $O_{2M+1}$ and $\tilde{O}_{2M+1,s}$ in Eqs.~\eqref{O-2M+1} and \eqref{tilde-O-2M+1}.
    Let $\mathbf{1}_A$ be the indicator function of the set $A:=\{t\in[0,2M+1]:\lfloor t\rfloor \text{ is odd}\}$ and then the effective error Hamiltonian at time $t$ can be written as 
    \begin{align*}
    \hat{E}(t):= \tilde{H}(t) - H(t) = E^{(\lceil t/2\rceil)}(\delta) \cdot \mathbf{1}_A(t).
    \end{align*}

    Let $\Phi_H(T_2,T_1)=\mathcal{T}e^{\i~\int_{T_1}^{T_2}\textrm{ad}_{H(t)}\,\d t}$ and $\Phi_{\tilde{H}}(T_2,T_1)=\mathcal{T}e^{\i~\int_{T_1}^{T_2}\textrm{ad}_{H(t)}\,\d t}$ be the time-evolution operator of the differential equation defining $O_T$ and $\tilde{O}_{T,s}$ in Eq.~\eqref{eq:def-O} and then $O_T$ and $\tilde{O}_{T,s}$ can be written as 
    \bes
    O_{T}=\Phi_H(T,0)O, \quad \tilde{O}_{T,s}=\Phi_{\tilde{H}}(T,0)O
    \ees
    for $T\in [0,2M+1]$. Using a time-dependent variation-of-parameters formula 
    \[
    D_2(T,0)=D_1(T,0)+\int_0^T D_1(T,\tau)B(\tau)D_2(\tau,0)~\d\tau,
    \]
    where
    \[
    D_2(T_2, T_1)=\mathcal{T}e^{\int_{T_1}^{T_2}(A(t)+B(t))\,\d t} ~\text{and } D_1(T_2, T_1)=\mathcal{T}e^{\int_{T_1}^{T_2} A(t)\,\d t},
    \]
    the difference between $\Phi_{\tilde{H}}(T,0)$ and $\Phi_{H}(T,0)$ can be written as 
    \begin{align*}
        \Phi_{\tilde{H}}(T,0)=\Phi_{H}(T,0)+\int_0^T\Phi_H(T,\tau)\i~\textrm{ad}_{\hat{E}(\tau)}\Phi_{\tilde{H}}(\tau,0)~\d\tau.
    \end{align*}
    We can use the variation-of-parameters formula $K-1$ times to obtain
    \bea \label{eq:variation}
        && \tilde{O}_{T,s}-O_T \nonumber \\
        &=& (\Phi_{\tilde{H}}(T,0) - \Phi_H (T,0) )O\nonumber\\
        &=& \int_0^T\Phi_H(T,\tau_1)\i~\textrm{ad}_{\hat{E}(\tau_1)}(\tilde{O}_{\tau_1,s})~\d\tau_1\nonumber \\ %\textcolor{blue}{\hat{E}(\tau_1)\rightarrow \tilde{E}(t)}
        &=& \int_0^T\Phi_H(T,\tau_1)\i~\textrm{ad}_{\hat{E}(\tau_1)}(O_{\tau})~\d\tau_1 + 
        \int_0^T  \int_0^{\tau_1} \Phi_{H}(T,\tau_1) \i \ad_{\hat{E}(\tau_1)} \Phi_H(\tau_1, \tau_2) \i \ad_{\hat{E}(\tau_2)}(\tilde{O}_{\tau_2,s})\,\, \d \tau_2 \d \tau_1 \nonumber\\
        &=& \sum_{\ell=1}^{K-1} \int_0^T  \int_0^{\tau_1}  \cdots \int_0^{\tau_{\ell-1}}  \Phi_{H}(T,\tau_1) \i \ad_{\hat{E}(\tau_1)} 
        %\Phi_{H}(\tau_1,\tau_2) \i \ad_{\hat{E}(\tau_2)} 
        \cdots \Phi_{H}(\tau_{\ell-1},\tau_\ell) \i \ad_{\hat{E}(\tau_\ell)}(O_{\tau_\ell}) \,\, \d \tau_\ell   \cdots \d \tau_1  \nonumber \\ 
        &\quad+&  \int_0^T  \int_0^{\tau_1}  \cdots \int_0^{\tau_{K-1}}  \Phi_{H}(T,\tau_1) \i \ad_{\hat{E}(\tau_1)} 
        %\Phi_{H}(\tau_1,\tau_2) \i \ad_{\hat{E}(\tau_2)} 
        \cdots \Phi_{H}(\tau_{K-1},\tau_K) \i \ad_{\hat{E}(\tau_K)}(\tilde{O}_{\tau_K,s}) \,\, \d \tau_K   \cdots \d \tau_1. 
    \eea
    For any $\ell\in[K-1]$, we have 
    \begin{align*}
        &\quad  \int_0^T  \int_0^{\tau_1}  \cdots \int_0^{\tau_{\ell-1}}  \Phi_{H}(T,\tau_1) \i \ad_{\hat{E}(\tau_1)}  
        \cdots \Phi_{H}(\tau_{\ell-1},\tau_\ell) \i \ad_{\hat{E}(\tau_\ell)}(O_{\tau_\ell}) \,\, \d \tau_\ell   \cdots \d \tau_1  \\
        &= \int_0^T  \int_0^{\tau_1}  \cdots \int_0^{\tau_{\ell-1}}  \prod_{\kappa=\ell}^1 \Big(\Phi_H(\tau_{\kappa-1},\tau_{\kappa}) \i \ad_{E^{(\lceil \tau_k/2\rceil)}(\delta)} \mathbf{1}_A(\tau_{\kappa}) \Big)(O_{\tau_\ell}) \,\, \d \tau_\ell   \cdots \d \tau_1 \\
        &=\int_0^T  \int_0^{\tau_1}  \cdots \int_0^{\tau_{\ell-1}} 
        \prod _{\kappa=\ell}^1 \Big(\Phi_H(\tau_{\kappa-1},\tau_{\kappa}) \i \mathbf{1}_A(\tau_{\kappa})\sum_{j_{\kappa}\in (2 \mathbb{N})_{\ge p} }^{\infty} \ad_{E^{(\lceil \tau_k/2\rceil)}_{j_{\kappa+1}}} \delta^{j_\kappa}\Big)(O(\tau_\ell)) 
        \,\, \d \tau_\ell   \cdots \d \tau_1 \\
        & = \sum_{j\in (2 \mathbb{N})_{\ge p \ell} } \delta^j 
        \int_0^T  \int_0^{\tau_1}  \cdots \int_0^{\tau_{\ell-1}} 
        \sum_{\substack{j_1, \ldots, j_\ell \in (2 \mathbb{N})_{\geq p} \\ j_1+\cdots+j_\ell=j}} \prod _{\kappa=\ell}^1 \Big(\Phi_H(\tau_{\kappa-1},\tau_{\kappa}) \i \mathbf{1}_A(\tau_{\kappa}) \ad_{E^{(\lceil \tau_k/2\rceil)}_{j_{\kappa+1}}} \Big) (O_{\tau_\ell})
        \,\, \d \tau_\ell   \cdots \d \tau_1,
    \end{align*}
    with $\tau_0\equiv T$, where the second equation follows from Eq.~\eqref{eq:def-E}. We then substitute $\delta=sM$, $T=2M+1$, and take summation over $\ell\in[K-1]$. This gives 
    %{\footnotesize
    % \begin{align*}
    %     & \quad \sum_{\ell=1}^{K-1}\sum_{j\in (2 \mathbb{N})_{\ge p\ell} } (sM)^j \int_0^{2M+1}   \int_0^{\tau_{1}} \cdots \int_0^{\tau_{\ell-1}}  \sum_{\substack{j_1 \ldots j_\ell \in (2 \mathbb{N})_{\geq p} \\ j_1+\cdots+j_\ell=j}} \prod _{\kappa=\ell}^1 \Big(\Phi_H(\tau_{\kappa-1},\tau_{\kappa}) \i \mathbf{1}_A(\tau_{\kappa}) \ad_{E^{(\lceil \tau_k/2\rceil)}_{j_{\kappa+1}}} \Big) (O_{\tau_\ell})
    %     \,\, \d \tau_\ell   \cdots \d \tau_1\\
        % &=\sum_{j\in (2 \mathbb{N})_{\ge p}} (sM)^j \sum_{\ell=1}^{\min\{ K-1, \lfloor j/p \rfloor \}} \int_0^{2M+1}   \int_0^{\tau_{1}} \cdots \int_0^{\tau_{\ell-1}}  
        % \sum_{\substack{j_1 \ldots j_\ell \in (2 \mathbb{N})_{\geq p} \\ j_1+\cdots+j_\ell=j}} \prod _{\kappa=\ell}^1 \Big(\Phi_H(\tau_{\kappa-1},\tau_{\kappa}) \i \mathbf{1}_A(\tau_{\kappa}) \ad_{E^{(\lceil \tau_k/2\rceil)}_{j_{\kappa+1}}} \Big) (O_{\tau_\ell})
        % \,\, \d \tau_\ell   \cdots \d \tau_1 \\
    %     &=\sum_{j\in 2 \mathbb{N}\ge p} s^j \tilde{E}_{j+1,K}(M)(O),
    % \end{align*}
    %}
    \[
    \text{First line of Eq.} ~\eqref{eq:variation} = \sum_{j\in (2 \mathbb{N})_{\ge p} } s^j \tilde{E}_{j+1,K}(M)(O),
    \]
    where $\tilde{E}_{j+1,K}$ is defined as 
    \beas
    \tilde{E}_{j+1,K}(M)(O) &=& M^j \sum_{\ell=1}^{\min\{ K-1, \lfloor j/p \rfloor \}} \int_0^{2M+1}   \int_0^{\tau_{1}} \cdots \int_0^{\tau_{\ell-1}}  \\
    &&  \sum_{\substack{j_1, \ldots, j_\ell \in (2 \mathbb{N})_{\geq p} \\ j_1+\cdots+j_\ell=j}} \prod _{\kappa=\ell}^1 \Big(\Phi_H(\tau_{\kappa-1},\tau_{\kappa}) \i \mathbf{1}_A(\tau_{\kappa}) \ad_{E^{(\lceil \tau_k/2\rceil)}_{j_{\kappa+1}}}  \Big) (O_{\tau_\ell})
        \,\, \d \tau_\ell   \cdots \d \tau_1.
    \eeas
    The induced spectral norm of the superoperator $\tilde{E}_{j+1,K}(M)$ is bounded as 
    \beas
    && \|\tilde{E}_{j+1,K}(M)\| \\
    &\leq& M^j \sum_{\ell=1}^{\min\{ K-1, \lfloor j/p \rfloor \}} \int_0^{2M+1}   \int_0^{\tau_{1}} \cdots \int_0^{\tau_{\ell-1}}  \sum_{\substack{j_1, \ldots ,j_\ell \in (2 \mathbb{N})_{\geq p} \\ j_1+\cdots+j_\ell=j}} \prod _{\kappa=\ell}^1   \mathbf{1}_A(\tau_{\kappa}) \|\ad_{E^{(\lceil \tau_k/2\rceil)}_{j_{\kappa+1}}} \|    
        \,\, \d \tau_\ell   \cdots \d \tau_1 \\
    &=& M^j \sum_{\ell=1}^{\min\{ K-1, \lfloor j/p \rfloor \}} \int_0^{2M+1}   \int_0^{\tau_{1}} \cdots \int_0^{\tau_{\ell-1}}   \prod_{\kappa=\ell}^1  \mathbf{1}_A(\tau_{\kappa})  
     \sum_{\substack{j_1 ,\ldots, j_\ell \in (2 \mathbb{N})_{\geq p} \\ j_1+\cdots+j_\ell=j}} 
     \prod_{\kappa=\ell}^1   \|\ad_{E^{(\lceil \tau_k/2\rceil)}_{j_{\kappa+1}}} \|   
        \,\, \d \tau_\ell   \cdots \d \tau_1 \\
    &\le& M^j \sum_{\ell=1}^{\min\{ K-1, \lfloor j/p \rfloor \}} 
    \frac{M^\ell}{\ell!} \sum_{\substack{j_1, \ldots ,j_\ell \in (2 \mathbb{N})_{\geq p} \\ j_1+\cdots+j_\ell=j}} 
    \prod_{\kappa=\ell}^1 2 \tilde{\alpha}_{\mathrm{comm}}^{\left(j_\kappa+1\right)} \frac{\left(a_{\max } \Upsilon\right)^{j_\kappa+1}}{\left(j_\kappa+1\right)^2} \\
    &=& \left(a_{\max } \Upsilon M\right)^j \sum_{\ell=1}^{\min \{K-1,\lfloor j / p\rfloor\}} \frac{\left(a_{\max } \Upsilon M\right)^\ell}{\ell!} \sum_{\substack{j_1 ,\ldots ,j_\ell \in (2 \mathbb{N})_{\geq p} \\ j_1+\cdots+j_\ell=j}}  \,\, \prod_{\kappa=1}^\ell  \frac{2\tilde{\alpha}_{\mathrm{comm}}^{(j_\kappa+1)}}{\left(j_\kappa+1\right)^2} .
    \eeas
    % \begin{align*}
    %     \|\tilde{E}_{j+1,K}(M)\|&\le M^j \sum_{l=1}^{\min\{ K-1, \lfloor j/p \rfloor \}} \int_0^{2M+1} d \tau_1 \int_0^{\tau_1} d \tau_2 \ldots \int_0^{\tau_{l-1}} d \tau_l\sum_{\substack{j_1 \ldots j_l \in 2 \mathbb{N} \geq p \\ j_1+\cdots+j_l=j}}  \Big(\prod _{\kappa=l}^1 \Big(\mathbf{1}_A(\tau_{\kappa})\big\|\ad_{E^{(\lceil \tau_k/2\rceil)}_{j_{\kappa+1}}}\big\| \Big)\Big)\\
    %     &=M^j \sum_{l=1}^{\min\{ K-1, \lfloor j/p \rfloor \}} \Big( \int_0^{2M+1} d \tau_1 \int_0^{\tau_1} d \tau_2 \ldots \int_0^{\tau_{l-1}} d \tau_l \Big(\prod _{\kappa=l}^1 \mathbf{1}_A(\tau_{\kappa})\Big)\Big)\Big(\sum_{\substack{j_1 \ldots j_l \in 2 \mathbb{N} \geq p \\ j_1+\cdots+j_l=j}}  \Big(\prod _{\kappa=l}^1 \big\|\ad_{E^{(\lceil \tau_k/2\rceil)}_{j_{\kappa+1}}}\big\| \Big)\Big)\Big)\\
    %     &\le M^j \sum_{l=1}^{\min\{ K-1, \lfloor j/p \rfloor \}} \frac{M^l}{l!} \sum_{\substack{j_1 \ldots j_l \in 2 \mathbb{N} \geq p \\ j_1+\cdots+j_l=j}}\left(\prod_{\kappa=1}^l 2 \tilde{\alpha}_{\mathrm{comm}}^{(j_\kappa+1)} \frac{\left(a_{\max } \Upsilon\right)^{j_\kappa+1}}{\left(j_\kappa+1\right)^2}\right) \\
    %     &=\left(a_{\max } \Upsilon M\right)^j \sum_{l=1}^{\min \{K-1,\lfloor j / p\rfloor\}} \frac{\left(a_{\max } \Upsilon M\right)^l}{l!} \sum_{\substack{j_1 \ldots j_l \in 2 \mathbb{N} \geq p \\ j_1+\cdots+j_l=j}}\left(\prod_{\kappa=1}^l 2 \frac{\tilde{\alpha}_{\mathrm{comm}}^{(j_\kappa+1)}}{\left(j_\kappa+1\right)^2}\right).
    % \end{align*}
    The first inequality follows from $\|\Phi_H(\tau_{\kappa-1}, \tau_{\kappa})\|=1$ as $H(t)$ is Hermitian, and the second inequality follows from $\mathrm{vol}(A^{\otimes \ell})=\mathrm{vol}(A)^\ell=M^\ell$ and take the order of $\tau_1,\tau_2,\ldots,\tau_\ell$ into account, the integration
\[
\int_0^{2M+1}   \int_0^{\tau_{1}} \cdots \int_0^{\tau_{\ell-1}}   \prod_{\kappa=\ell}^1  \mathbf{1}_A(\tau_{\kappa})    
        \,\, \d \tau_\ell   \cdots \d \tau_1 = \frac{\mathrm{vol}(A^{\otimes \ell})}{\ell !}=\frac{M^\ell}{\ell!} .
\]
    % \begin{align*}
    %     \int_0^{2M+1} d \tau_1 \int_0^{\tau_1} d \tau_2 \ldots \int_0^{\tau_{l-1}} d \tau_l \Big(\prod _{\kappa=l}^1 \mathbf{1}_A(\tau_{\kappa})\Big)=\frac{M^l}{l!}.
    % \end{align*}
    By similar arguments, we can show that the second term in Eq.~\eqref{eq:variation} can be written as 
    \beas
    \tilde{F}(M,s)(O) &=& \sum_{j\in (2 \mathbb{N})_{\ge Kp}} (sM)^j \int_0^{2M+1}  \int_0^{\tau_1}  \cdots \int_0^{\tau_{K-1}}  \\
    && \sum_{\substack{j_1 ,\ldots ,j_K \in 
    (2 \mathbb{N})_{\geq p} \\ j_1+\cdots+j_K=j}}   \prod _{\kappa=K}^1 \Big(\Phi_H(\tau_{\kappa-1},\tau_{\kappa}) \i \mathbf{1}_A(\tau_{\kappa}) \ad_{E^{(\lceil \tau_k/2\rceil)}_{j_{\kappa+1}}} \Big) (\tilde{O}_{\tau_K,s})
    \,\, \d \tau_K \cdots \d \tau_1 .
    \eeas
    % \begin{align*}
    %     & \quad \int_0^T d \tau_1 \int_0^{\tau_1} d \tau_2 \ldots \int_0^{\tau_{K-1}} d \tau_K \Phi_{H}(T,\tau_1) \i \ad_{\hat{E}(\tau_1)} \Phi_{H}(\tau_1,\tau_2) \i \ad_{\hat{E}(\tau_2)} \ldots \Phi_{H}(\tau_{K-1},\tau_K) \i \ad_{\hat{E}(\tau_K)}(\tilde{O}(\tau_K,s)) \\
    %     &=\sum_{j\in 2 \mathbb{N}\ge Kp} (sM)^j \int_0^{2M+1} d \tau_1 \int_0^{\tau_1} d \tau_2 \ldots \int_0^{\tau_{K-1}} d \tau_K\sum_{\substack{j_1 \ldots j_K \in 2 \mathbb{N} \geq p \\ j_1+\cdots+j_K=j}}  \Big(\prod _{\kappa=K}^1 \Big(\Phi_H(\tau_{\kappa-1},\tau_{\kappa}) \i \mathbf{1}_A(\tau_{\kappa}) \ad_{E^{(\lceil \tau_k/2\rceil)}_{j_{\kappa+1}}} \Big)\Big)(\tilde{O}(\tau_K,s))\\
    %     &= \tilde{F}(M,s)(O).
    % \end{align*}    
    The operator norm of the operator $\tilde{F}_K(M,s)$ is bounded as \begin{align*}
        \|\tilde{F}_K(M,s)\| \leq 
        \frac{\left(a_{\max } \Upsilon M\right)^K}{K!} \sum_{j\in (2 \mathbb{N})_{\ge Kp} } \left(sa_{\max } \Upsilon M\right)^j   \sum_{\substack{j_1, \ldots, j_K \in (2 \mathbb{N})_{\geq p} \\ j_1+\cdots+j_K=j}} \prod_{\kappa=1}^K \,\, \frac{2\tilde{\alpha}_{\mathrm{comm}}^{(j_\kappa+1)}}{\left(j_\kappa+1\right)^2} .
    \end{align*}
This completes the proof.
\end{proof}

For any $s_0\in (0,1]$ with $1/s_0\in\mathbb{N}$, let 
\be
\label{tilde-O-2M+1:Richardson extrapolation}
\tilde{O}^{(m)}_{2M+1,s_0} = \sum_{k=1}^m b_k 
\tilde{O}_{2M+1,s_k} 
\ee
be the $m$-point Richardson extrapolation with $b_k,s_k$ given in Lemma~\ref{lem_Richad for Trotter}, where $\tilde{O}_{2M+1,s_k} $ is defined via Eq.~\eqref{tilde-O-2M+1}. 
Let 
\be
\label{Lambda-j-l}
\Lambda_{j, \ell}:= \max\Bigg\{\Bigg(\sum_{\substack{j_1, \ldots, j_\ell \in (2\mathbb{N})_{\geq p} \\ j_1+\cdots+j_\ell=j}} \,\, \prod_{\kappa=1}^\ell  \frac{2\tilde{\alpha}_{\mathrm{comm}}^{(j_\kappa+1)}}{\left(j_\kappa+1\right)^2}\Bigg)^{1 /(j+\ell)}, \quad (\tilde{\alpha}_{\mathrm{comm}}^{(j)})^{1/j}\Bigg\}
\ee
and 
\be \label{Lambda}
\lambda_{\rm comm} :=  \max_{\substack{j \in (2 \mathbb{N})_{\geq 2m} \\ \ell \in [K]}} \Lambda_{j, \ell}.
\ee

%\yz{It seems that we could remove Remark \ref{remark: A bound of Lambda} below, and replace all of $\Lambda$ in the complexity results with $\lambda_{\rm comm}$ given as follows?}

\begin{rmk}[Upper bounds of $\lambda_{\rm comm}$]
\label{remark: A bound of Lambda}

Here, we simplify the expression of $\lambda_{\rm comm}$ by providing some known upper bounds. 

\begin{itemize}
    \item For general Hamiltonians, according to \cite[Eqs.~(33) and (78)]{watson2024exponentially}, we have 
    \be \label{eq:para comm lambda}
    \lambda_{\rm comm}\leq 4 \max_{\ell\in[M]} \sum_{\ell=1}^{\Gamma_{\ell}}\|H_{j}^{(\ell)}\|.
    \ee
    \item Suppose that there exist $\alpha>0$ and $C_{\alpha} > 1$ such that 
    \begin{align}
        \label{eq:alpha-bound}
    \tilde{\alpha}_{\mathrm{comm}}^{(j)}\le C_{\alpha}\alpha^j/2
    \end{align}
    for all $j\in (2\mathbb{N})_{\ge p}$.
    Then, for all $j\in (2\mathbb{N})_{\ge 2m}$ and $\ell \in [K]$, we have $(\tilde{\alpha}_{\mathrm{comm}}^{(j)})^{1/j} \le C_{\alpha}^{1/p} \alpha$ and
    \beas 
    \Bigg(\sum_{\substack{j_1, \ldots, j_\ell \in (2\mathbb{N})_{\geq p} \\ j_1+\cdots+j_\ell=j}}  \prod_{\kappa=1}^\ell  \frac{2\tilde{\alpha}_{\mathrm{comm}}^{(j_\kappa+1)}}{\left(j_\kappa+1\right)^2}\Bigg)^{1 /(j+\ell)}
        &\le & \Bigg(\sum_{\substack{j_1, \ldots, j_\ell \in (2\mathbb{N})_{\geq p} \\ j_1+\cdots+j_\ell=j}} \prod_{\kappa=1}^\ell \frac{(C_{\alpha}^{1/(p+1)}\alpha)^{j_\kappa+1}}{\left(j_\kappa+1\right)^2}\Bigg)^{1 /(j+\ell)} \\
        &=& \alpha C_{\alpha}^{1/(p+1)} \cdot \Bigg(\sum_{\substack{j_1, \ldots, j_\ell \in (\mathbb{N})_{\geq p} \\ j_1+\cdots+j_\ell=j}} \prod_{\kappa=1}^\ell \frac{1}{\left(j_\kappa+1\right)^2}\Bigg)^{1 /(j+\ell)}\\
        &\le& \alpha C_{\alpha}^{1/(p+1)}  \Bigg( \sum_{\substack{j_1, \cdots j_\ell ,\in \mathbb{N} \\ j_1+\cdots+j_\ell=j}} 1\Bigg)^{1 /(j+\ell)} \\
        &=& \alpha C_{\alpha}^{1/(p+1)}\binom{j+\ell-1}{\ell-1}^{1 /(j+\ell)}\\
        &\le& \alpha C_{\alpha}^{1/(p+1)}(2^{j+\ell})^{1/(j+\ell)}=2C_{\alpha}^{1/(p+1)}\alpha,
    \eeas
where the first inequality follows from $C_{\alpha} >1$ and $j_{\kappa}\ge p$. 
This implies that 
\be
\lambda_{\rm comm}\le 2C_{\alpha}^{1/p}\alpha.
\ee
The condition in Eq.~\eqref{eq:alpha-bound} is satisfied by many physical Hamiltonians, as elucidated in \cite{childs2021theory, aftab2024multi}. We use some of these results to estimate $\lambda_{\rm comm}$.
\begin{itemize}
    \item[(i)] For Hamiltonian $H=\sum_{\gamma=1}^{\Gamma}H_{\gamma}$,     \cite{childs2021theory} showed that if $H$ is an electronic-structure Hamiltonians over $n$ orbitals, we have $\alpha_{\mathrm{comm}}^{(j)}=O(n^{j})$ and so 
    \begin{equation}
    \label{eq:comm-scaling-electronic-structure-ham}
        \lambda_{\rm comm}=O(n).
    \end{equation}

    \item[(ii)] For $k$-local Hamiltonians over $n$ qubits $H=\sum_{j_1, \ldots, j_k} H_{j_1, \ldots, j_k}$, we have 
    $$\alpha_{\mathrm{comm}}^{(j)}=O\left(\ltrinorm H \rtrinorm_1^{j-1}\|H\|_1\right),$$ 
    where, 
    $$\ltrinorm H \rtrinorm_1=\max _\ell \Big\{ \max _{j_\ell} \sum_{j_1, \ldots, j_{\ell-1}, j_{\ell+1}, \ldots, j_k}\left\|H_{j_1, \ldots, j_k}\right\| \Big\}$$ is the induced 1-norm. In this case, we have 
    \begin{equation}
    \label{eq:comm-scaling-k-local-ham}
        \lambda_{\rm comm} = O\left((\|H\|_1/\ltrinorm H \rtrinorm_1)^{1/p} \ltrinorm H \rtrinorm_1\right).
    \end{equation}

    \item[(iii)] For Hamiltonians with a power law decay, this scaling improves further. Let $\Lambda\subseteq\mathbb{R}^d$ be an $n$-qubit $d$-dimensional square lattice. Then, a power law Hamiltonian $H$ on $\Lambda$ with exponent $r$, is written as
    $$
    H=\sum_{\vec{i},\vec{j}\in\Lambda} H_{\vec{i},\vec{j}},
    $$
    where each $H_{\vec{i},\vec{j}}$ acts non-trivially on two qubits $\vec{i},\vec{j}\in\Lambda$ and $\|\vec{i}-\vec{j}\|=1$, if $i=j$, and $\|\vec{i}-\vec{j}\|=\|\vec{i}-\vec{j}\|^{-r}_2$, if $\vec{i}\neq \vec{j}$. Then, \cite{childs2021theory} proved that
    $$
    \ltrinorm H\rtrinorm_1=\begin{cases}
        O(n^{1-r/d}), & 0\leq r<d\\
        O(\log n), & r=d\\
        O(1), & r>d
    \end{cases}
    \qquad 
    \|H\|_1=\begin{cases}
        O(n^{2-r/d}), & 0\leq r<d\\
        O(n\log n), & r=d\\
        O(n), & r>d
    \end{cases}
    $$
    Substituting the above values of $\ltrinorm H\rtrinorm_1$  and $\|H\|_1$ into Eq.~\eqref{eq:comm-scaling-k-local-ham}, we obtain different scalings for $\lambda_{\rm comm}$.
\end{itemize}

    %\cs{XZ, revise this part, $\lambda_{\rm comm}$ or $\Lambda$? also in the following}\yz{Based on the current definition, all $\Lambda$ should be replaced by $\lambda_{\rm comm}$.}
\end{itemize}

\end{rmk}

We now state a lemma which upper bounds the error between $O_{2M+1}$ and $ \tilde{O}^{(m)}_{2M+1,s_0}$.

\begin{lem}
\label{lem:qsp-richardson}
Let $O_{2M+1} $ and $ \tilde{O}^{(m)}_{2M+1,s_0}$ be defined via Eqs. \eqref{O-2M+1} and \eqref{tilde-O-2M+1:Richardson extrapolation} respectively. Suppose that $s_m a_{\max}\Upsilon M\lambda_{\rm comm}\le 1/2$. Then for any initial state $\ket{\psi_0}$, the extrapolation error is bounded as follows
\beas
 \big|\bra{\psi_0}\tilde{O}^{(m)}_{2M+1,s_0} \ket{\psi_0}-\bra{\psi_0}O_{2M+1} \ket{\psi_0}\big| 
\leq  4 \|O\|\|\b\|_1 (s_m a_{\max}\Upsilon M\lambda_{\rm comm})^{2m} \max\{a_{\max} \Upsilon M\lambda_{\rm comm},1\}^K,
\eeas
where $K=\left\lceil{2m}/{p}\right\rceil$.
\end{lem}

\begin{proof}
    Define 
    \[\begin{aligned}
        f(s) &:=\bra{\psi_0}\tilde{O}_{2M+1,s}\ket{\psi_0} \quad \text{ for } s\in(0,s_m], \\
        f(0) &:= \bra{\psi_0}O_{2M+1} \ket{\psi_0}
    \end{aligned}\]
    be the function we want to extrapolate. 
    By the definition of $r_k$ in Eq.~\eqref{eq:def-r-trotter} and $s_k=s_0/r_k$, all $s_k$ are in the domain of $f(s)$. 
    As $s_m a_{\max}\Upsilon M\lambda_{\rm comm}\le 1/2$ and by the definition of $\lambda_{\rm comm}$, we have 
    \begin{align*}
    \tilde{\alpha}_{\mathrm{comm}}^{(j)} \cdot (a_{\max } \Upsilon|s_mM|)^j \leq (s_m a_{\max}\Upsilon M \lambda_{\rm comm})^j \le 1/2^j
    \end{align*}
    for all $j\ge 2m$. 
    So condition \eqref{BCH convergence condition} holds with $t=s_m M, C=1/2^{2m}, J=2m$. 
    By Lemma~\ref{lem:interleaved}, $f(s)$ can be  expanded as 
    \begin{equation}
        \begin{aligned}
            \label{eq:f-expansion}
            f(s)&=\bra{\psi_0}O_{2M+1} \ket{\psi_0}+\sum_{j \in 2 \mathbb{N}, j \ge p} s^j \bra{\psi_0}\tilde{E}_{j+1, K}(M)(O)\ket{\psi_0}+\bra{\psi_0}\tilde{F}_K(M, s)(O)\ket{\psi_0}\\
            &=f(0)+\sum_{\substack{j \in 2 \mathbb{N}\\ 2m>j\ge p}} s^j \bra{\psi_0}\tilde{E}_{j+1, K}(M)(O)\ket{\psi_0}+R_{2m}(s)
        \end{aligned}
    \end{equation}
    for all $s\in (0, s_m]$, where 
    \be
        \label{eq:reminder}
    R_{2m}(s):= \sum_{j\in 2\mathbb{N}, j\ge 2m} s^j \bra{\psi_0}\tilde{E}_{j+1, K}(M)(O)\ket{\psi_0}+\bra{\psi_0}\tilde{F}_K(M, s)(O)\ket{\psi_0} 
    \ee
    is the reminder term. Let 
    $F_m(s_0) := \bra{\psi_0}\tilde{O}^{(m)}_{2M+1,s_0}\ket{\psi_0}= \sum_{k=1}^m b_k f(s_k)
    $ be the Richardson extrapolation of $f(s)$ at points $s_1,s_2,\ldots,s_m$. By Lemma~\ref{lem_Richad for Trotter}, the error of the Richardson extrapolation is bounded as \be
        \label{eq:extrapolation-reminder}
        \big|\bra{\psi_0}\tilde{O}^{(m)}_{2M+1,s_0}\ket{\psi_0}-\bra{\psi_0}O_{2M+1} \ket{\psi_0}\big|=
        |F_m(s_0)-f(0)| 
        \le \|\b\|_1 \max_{k\in[m]} |R_{2m}(s_k)|.
    \ee
    As $j\ge 2m \ge K(p-1)$, by Eq.~\eqref{eq:bound-E-F}, the norm of $\tilde{E}_{j+1,K}(M)$ and $\tilde{F}_K(M, s)$ can be bounded as 
    \bea
    \|\tilde{E}_{j+1, K}(M)\| 
    &\leq& (a_{\max } \Upsilon M)^j \sum_{\ell=1}^{\min \{K-1,\lfloor j / p\rfloor\}} \frac{(a_{\max } \Upsilon M)^\ell}{\ell!} \sum_{\substack{j_1, \ldots, j_\ell \in (2 \mathbb{N})_{\ge p} \\
     j_1+\cdots+j_\ell=j}} \prod_{\kappa=1}^\ell  \frac{2 \tilde{\alpha}_{\mathrm{comm}}^{(j_\kappa+1)}}{(j_\kappa+1)^2}  \nonumber \\
    &=& (a_{\max } \Upsilon M)^j \sum_{\ell=1}^{K-1} \frac{(a_{\max } \Upsilon M)^\ell}{\ell!} \Lambda_{j,\ell}^{j+\ell} \nonumber \\
    &\le& \sum_{\ell=1}^{K-1} \frac{(a_{\max}\Upsilon M\lambda_{\rm comm})^{j+\ell}}{\ell!},
    \label{eq:bound-E-2}
    \eea
    and 
    \bea
        \|\tilde{F}_{K}(M,s)\| &\leq& \frac{(a_{\max } \Upsilon M)^K}{K!} \sum_{j \in 2\mathbb{N}, j \geq K p}(a_{\max } \Upsilon s M)^j\Big(\sum_{\substack{j_1 ,\ldots, j_K \in 2 \mathbb{N} \geq p \\ j_1+\cdots+j_K=j}}\Big(\prod_{\kappa=1}^K  \frac{2\tilde{\alpha}_{\mathrm{comm}}^{(j_\kappa+1)}}{(j_\kappa+1)^2}\Big)\Big) \nonumber \\
        &=&  \sum_{j \in (2 \mathbb{N})_{\geq K p} }s^j\frac{(a_{\max } \Upsilon M \Lambda_{j,K})^{j+K}}{K!} \nonumber \\
        &\le& \sum_{j \in (2 \mathbb{N} )_{\geq K p} } s^j\frac{(a_{\max } \Upsilon M \lambda_{\rm comm})^{j+K}}{K!} .
    \label{eq:bound-F-2}
    \eea
    By Eqs.~\eqref{eq:reminder}, \eqref{eq:bound-E-2}, and \eqref{eq:bound-F-2}, the reminder $R_{2m}(s)$ can be bounded as 
    \beas
        |R_{2m}(s)|& \le& \sum_{j\in 2\mathbb{N}, j\ge 2m} s^j \sum_{\ell=1}^{K-1} \frac{(a_{\max}\Upsilon M\lambda_{\rm comm})^{j+\ell}}{\ell!}+\sum_{j \in 2 \mathbb{N} , j\geq K p} s^j\frac{(a_{\max } \Upsilon M \lambda_{\rm comm})^{j+K}}{K!}\\
        &\le& \sum_{j\in 2\mathbb{N}, j\ge 2m} s^j \sum_{\ell=1}^{K} \frac{(a_{\max}\Upsilon M\lambda_{\rm comm})^{j+\ell}}{\ell!} \\
        &\le& \sum_{j\in 2\mathbb{N}, j\ge 2m} (sa_{\max}\Upsilon M\lambda_{\rm comm})^j \max\{a_{\max} \Upsilon M\lambda_{\rm comm},1\}^K\sum_{\ell=1}^{K} \frac{1}{\ell!} \\
        &\le& (sa_{\max} \Upsilon M \lambda_{\rm comm})^{2m} \sum_{j \ge 0} \Bigl(\frac{1}{2}\Bigr)^{2j} \max \{a_{\max} \Upsilon M \lambda_{\rm comm}, 1\}^K (e-1)\\
        &\le& 4(sa_{\max}\Upsilon M\lambda_{\rm comm})^{2m} \max\{a_{\max} \Upsilon M\lambda_{\rm comm},1\}^K,
    \eeas
    where the second inequality follows from $Kp \ge 2m$ and the 
    fourth inequality follows from the estimate $sa_{\max}\Upsilon M\lambda_{\rm comm}$ $\le s_m a_{\max}\Upsilon M\lambda_{\rm comm} \le 1/2$. By Eq.~\eqref{eq:extrapolation-reminder}, we can bound the extrapolation error as 
    \beas
     & &\big|\bra{\psi_0}\tilde{O}^{(m)}_{2M+1,s_0}\ket{\psi_0}-\bra{\psi_0}O_{2M+1} \ket{\psi_0}\big| \\
     &\le& \|O\|\|\b\|_1 \max_{k\in[m]} 4(sa_{\max}\Upsilon M\lambda_{\rm comm})^{2m} \max\{a_{\max} \Upsilon M\lambda_{\rm comm},1\}^K \\
     &=& 4 \|O\|\|\b\|_1 (s_m a_{\max}\Upsilon M\lambda_{\rm comm})^{2m} \max\{a_{\max} \Upsilon M\lambda_{\rm comm},1\}^K,
    \eeas
which is as claimed.
\end{proof}

%\subsection{Main theorems}

We now summarize our main theorems. Below, we aim to provide a time complexity analysis rather than expressing complexity in terms of the number of Trotter steps as in \cite{watson2024exponentially}. This results in an additional factor of $O(\Gamma \Upsilon)$ in the complexity compared to \cite{watson2024exponentially}, as each Trotter step costs $O(\Gamma \Upsilon)$ in the worst case. But for certain Hamiltonians, more efficient implementations of each Trotter step exist \cite{childs2021theory}.

\begin{thm}[Interleaved sequences of unitaries and Hamiltonian evolution with Trotter, incoherent estimation]
\label{thm:qsp-with-trotter-and-interpolation}
Let $\{H^{(\ell)}\}_{\ell=1}^{M}$ be a set of Hermitian operators, with decomposition 
$H^{(\ell)} = \sum_{\gamma=1}^{\Gamma_{\ell}} H^{(\ell)}_{\gamma}$, such that Hamiltonian evolution of each \( H^{(\ell)}_{\gamma} \) can be implemented with circuit depth at most $\tau_H$. Let 
\[
W := V_0 \prod_{\ell=1}^M e^{\i H^{(\ell)}} V_{\ell},
\]
where $\{V_\ell\}_{\ell=0}^{M}$ are unitaries with circuit depth $\tau_\ell$.

Let $p \in \mb N$, $\varepsilon\in(0,1)$ be the precision parameter, $\tau_{\mathrm{sum}}=\sum_{\ell=0}^{M} \tau_\ell$, $\Gamma_{\mathrm{avg}}=\sum_{\ell=1}^M \Gamma_{\ell}/M$, and let $\lambda_{\rm comm}$ be defined via \eqref{Lambda} with 
$m=p\lceil\log(1/\varepsilon)\rceil, K = 2\lceil\log(1/\varepsilon)\rceil$.
%\yz{Why is $m=O(p \log(1/\varepsilon))$ instead of $m = O(\log(1/\varepsilon))$ here? } \xz{Just to make sure that $m$ is a multiple of $p$ so that Eq.~\eqref{eq:extro-error} we have $K=2m/p$ rather than $K= \lfloor 2m/p\rfloor$}
Suppose there is a symmetric $p$-th order product formula with $\Upsilon$ stages and maximum coefficient $a_{\max}$ satisfying
$a_{\max}\Upsilon M \lambda_{\rm comm} \geq 1$.
Assume the initial state $\ket{\psi_0}$ can be prepared efficiently, then there is an algorithm (see Algorithm \ref{alg: HSVT with trotter}) that, for any observable $O$, estimates $\bra{\psi_0} W^{\dagger} O W \ket{\psi_0}$ with an error of at most $\varepsilon\|O\|$ and a success probability of at least $2/3$, without using any ancilla qubits. The maximum quantum circuit depth is 
\[
\widetilde{O}\Big(p\tau_H\Gamma_{\mathrm{avg}}\Upsilon^{2+1/p} \big(a_{\max}M \lambda_{\rm comm} \big)^{1+1/p}+\tau_{\mathrm{sum}}\Big),
\]
and the total time complexity is \[
\widetilde{O}\Big(\frac{p\tau_H\Gamma_{\mathrm{avg}}\Upsilon^{2+1/p} \big(a_{\max} M  \lambda_{\rm comm} \big)^{1+1/p}+\tau_{\mathrm{sum}}}{\varepsilon^2}\Big).
\]
\end{thm}

\begin{proof}
We follow the notation in Lemma~\ref{lem:qsp-richardson} and $\bra{\psi_0} W^{\dagger} O W \ket{\psi_0} = \bra{\psi_0}O_{2M+1} \ket{\psi_0}$ is the zero-step-size limit of the extrapolation. For $s_m a_{\max} \Upsilon M \lambda_{\rm comm}\le 1/2$, Lemma~\ref{lem:qsp-richardson} bounds the Richardson extrapolation error as 
\be
\label{eq:extro-error}
\big|\bra{\psi_0} \tilde{O}^{(m)}_{2M+1,s_0} \ket{\psi_0}-\bra{\psi_0}O_{2M+1} \ket{\psi_0}\big|
\le
% & 4 \|O\|\|\b\|_1 (s_m a_{\max}\Upsilon M\lambda_{\rm comm})^{2m} \max\{a_{\max} \Upsilon M\lambda_{\rm comm},1\}^K \nonumber \\
% &=& 
4 \|O\|\|\b\|_1 s_m^{2m}(a_{\max}\Upsilon M\lambda_{\rm comm})^{2m(1+1/p)}. 
\ee
By Lemma~\ref{lem_Richad for Trotter}, $\|\mathbf{b}\|_1 = O(\log m)$, leading to $\|\b\|_1 ^{-1/2m}
%=\Omega((\log m)^{-1/2m})
=\Omega(1)$, and $\varepsilon^{1/2m}=\varepsilon^{1/(2p\lceil \log(1/\varepsilon) \rceil)}=\Omega(1)$. 
Choosing
\begin{align*}
s_m=O\big(\|\b\|_1 ^{-1/2m} \varepsilon^{1/2m}(a_{\max} \Upsilon M \lambda_{\rm comm})^{-1-1/p}\big)=O\big((a_{\max} \Upsilon M \lambda_{\rm comm})^{-1-1/p})
\end{align*}
ensures that the extrapolation error is at most $\varepsilon\|O\|/2$, also ensuring $s_m a_{\max} \Upsilon M \lambda_{\rm comm}\le 1/2$.
   
Next, we analyze the complexity of estimating 
\[
\bra{\psi_0} \tilde{O}^{(m)}_{2M+1,s_0} \ket{\psi_0} = \sum_{k=1}^m b_k\bra{\psi_0} \tilde{O}_{2M+1,s_k} \ket{\psi_0}.
\]
The $p$-th order staged product formula $\mathcal{P}^{(\ell)}$ for $H^{(\ell)}$ requires \( 1/(s_k M) \in \mathbb{N} \) to allow implementation using \( 1/(s_k M) \) Trotter steps.
Then we apply the unitary evolution 
$
V_0 \prod_{\ell=1}^M \mathcal{P}^{(\ell)}(s_k M)^{1/(s_k M)} V_{\ell}
$
to \( \ket{\psi_0} \), and measure \( O \). This process employs
\begin{align*}N_{\mathrm{rep}}=O\left((\|\b\|_1/\varepsilon)^2\log(m)\right)=\widetilde{O}\left(1/\varepsilon^2\right)
% \\
% &=O(\log^2 m\log(m/\delta)/\varepsilon^2)\\
% &=O(\log^2(p\log(1/\varepsilon))\log(p\log(1/\varepsilon)/\delta)/\varepsilon^2)
\end{align*}
repetitions. Let \( \tilde{f}(s_k) \) denote the average of the measurement outcomes. Hoeffding's inequality ensures that $\tilde{f}(s_k)$ estimates 
\begin{align}
\label{eq:one-sample-estimate}
% \Tr\left(O \left( V_0\prod_{\ell=1}^M \mathcal{P}^{(\ell)}(s_k M)^{1/(s_k M)}V_{\ell} \right) \rho \left( V_0\prod_{\ell=1}^M \mathcal{P}^{(\ell)}(s_k M)^{1/(s_k M)}V_{\ell}  \right)^\dag \right)=
\bra{\psi_0}\tilde{O}_{2M+1,s_k} \ket{\psi_0}
\end{align}
up to error $\varepsilon\|O\|/(2\|\b\|_1)$ with probability at least $1-1/(3m)$. 
The union bound guarantees that $\sum_{k=1}^m b_k\tilde{f}(s_k)$ estimates $\bra{\psi_0}\tilde{O}^{(m)}_{2M+1,s_0}\ket{\psi_0}$ with error at most $\varepsilon\|O\|/2$ and probability at least $2/3$. 

Since $s_k=s_0/r_k$ with $r_k\in \mathbb{Z}$, we choose $1/s_0\in M\mathbb{N}$ to guarantee that $1/(s_k M)\in \mathbb{N}$ for all $k\in[m]$. This allows us to obtain estimates $\tilde{f}(s_k)$ for all $k$.
By definition,  $r_1=\Theta(m^2)$ and $r_m=\Theta(m)$. So we should ensure that $s_0=s_m r_m=O(m (a_{\max} \Upsilon M \lambda_{\rm comm})^{-1-1/p})$ and $1/s_0=\Omega(M^{1+1/p})=\Omega(M)$. 
Thus, for sufficiently large $M$, we can choose $1/s_0\in M \mathbb{N}$ to satisfy both requirements. 

The number of Trotter steps of estimating $\bra{\psi_0} \tilde{O}_{2M+1,s_k} \ket{\psi_0}$ is $1/s_k$.
Hence, the total number of Trotter steps is 
\beas
N_{\mathrm{rep}}\sum_{k=1}^m \frac{1}{s_k} 
&=& N_{\mathrm{rep}}\sum_{k=1}^m  \frac{r_k}{s_0} \\
% &=& \frac{N_{\mathrm{rep}}}{s_0}\sum_{k=1}^m \bigg\lceil \frac{\sqrt{8} m}{\pi \sin (\pi(2 k-1) / 8 m)}\bigg\rceil \\
&=& O\left(\frac{N_{\mathrm{rep}}}{s_0} \sum_{k=1}^m\frac{m}{k/m} \right)=O\left(\frac{N_{\mathrm{rep}}}{s_0}m^2\log m\right) = \widetilde{O}\left(p\frac{(a_{\max} \Upsilon M \lambda_{\rm comm})^{1+1/p}}{ \varepsilon^2}\right).
\eeas
% where the last equation follows from $s_0=O(m(a_{\max} \Upsilon M \lambda_{\rm comm} )^{-1-1/p})$, $m = O(p\log(1/\varepsilon))$, and $N_{\mathrm{rep}}=\tilde{O}(1/\varepsilon^2)$.
Each Trotter step involves $\Gamma_{\ell} \Upsilon $ Hamiltonian evolutions of Hermitian terms, each implemented with depth at most $\tau_H$. 
Combining these, the total time complexity is
\[
\widetilde{O}\left(\frac{p\tau_H\Gamma_{\mathrm{avg}}\Upsilon^{2+1/p}(a_{\max} M \lambda_{\rm comm})^{1+1/p}+\tau_{\mathrm{sum}}}{\varepsilon^2}\right),
\]
% \begin{align*}
%     \tilde{O}\Big(p\tau_H \sum_{\ell=1}^M \frac{\Gamma_{\ell}}{M}\Upsilon \frac{(a_{\max} \Upsilon M \lambda_{\rm comm})^{1+1/p}}{ \varepsilon^2}+N_{\mathrm{rep}} \tau_{\mathrm{sum}}\Big)=\tilde{O}\Big(\frac{p\tau_H\Gamma_{\mathrm{avg}}\Upsilon^{2+1/p}(a_{\max} M \lambda_{\rm comm})^{1+1/p}+\tau_{\mathrm{sum}}}{\varepsilon^2}\Big).
% \end{align*}  
and the maximum quantum circuit depth is 
\[
\tau_H \Gamma_{\mathrm{avg}} \Upsilon /s_1+\tau_{\mathrm{sum}}=\widetilde{O}(p\tau_H\Gamma_{\mathrm{avg}}\Upsilon^{2+1/p}(a_{\max}M\lambda_{\rm comm})^{1+1/p}+\tau_{\mathrm{sum}}).
\] 
%\cs{What is query here?}\xz{To be a general framework, $V_i$ might not be implemented with constant gates, so we can regard $V_i$ as queries and measure the query complexity? Then we need also change the statement of this theorem. }
This completes the proof.
\end{proof}

We now present a variant of Theorem~\ref{thm:qsp-with-trotter-and-interpolation}, which makes use of iterative quantum amplitude estimation \cite{grinko2021iterative} to coherently estimate the expectation value. This increases the maximum quantum circuit depth but reduces the overall time complexity to \( \widetilde{O}(1/\varepsilon) \). The iterative quantum amplitude estimation is formally introduced below:

 \begin{lem}[Iterative Quantum Amplitude Estimation  {\cite[Theorem 1]{grinko2021iterative}}]
        \label{lem:iqae}
        Given $\varepsilon, \delta\in (0,1)$ and an $(n+1)$-qubit unitary $W$ such that \begin{align*}
            W\ket{0^n}\ket{0} = \sqrt{a}\ket{\psi_0} \ket{0} +\sqrt{1-a}\ket{\psi_1}\ket{1},
        \end{align*}
        where $\ket{\psi_0}$ are $\ket{\psi_1}$ and arbitrary $n$-qubit states, there exists an algorithm that outputs an estimate of $a$ to within additive error $\varepsilon$ with probability at least  $1-\delta$. The algorithm uses
        \[
        O\left(\dfrac{1}{\varepsilon}\log\left(\dfrac{\log(1/\varepsilon)}{\delta}\right)\right)
        \]
        applications of $U$ and $U^{\dagger}$ and requires one ancilla qubit besides the $n+1$ qubits required by $W$.
    \end{lem}

Now we are in a position to state the theorem formally:

\begin{thm}[Interleaved sequences of unitaries and Hamiltonian evolution with Trotter, coherent estimation]
    \label{thm:qsp-with-trotter-and-interpolation-coherent}
    Under the same assumption as in Theorem~\ref{thm:qsp-with-trotter-and-interpolation}, if we further assume $O/\gamma$ is block-encoded via a unitary $U_O$ for some $\gamma$ satisfying $\gamma = \widetilde{\Theta}(\|O\|)$, then there exists a quantum algorithm that estimates $\bra{\psi_0} W^{\dagger} O W \ket{\psi_0}$ with an error of at most $\varepsilon\|O\|$ and a success probability of at least $2/3$ using two ancilla qubits. 
    The maximum quantum circuit depth and total time complexity are both 
    \[\widetilde{O}\left(\frac{p\tau_H\Gamma_{\mathrm{avg}}\Upsilon^{2+1/p} \big(a_{\max} M  \lambda_{\rm comm} \big)^{1+1/p}+\tau_{\mathrm{sum}}}{\varepsilon}\right).
    \]
    The number of queries to $U_O$ is $\widetilde{O}(p/\varepsilon)$.
\end{thm}
\begin{proof}
    %We follow a similar analysis to \cite[Section E]{rendon2024improved}. 
    We use the same sample points as Theorem~\ref{thm:qsp-with-trotter-and-interpolation} for extrapolation, so the extrapolation error is bounded by $\|O\|\varepsilon/2$. For any \( k \in [M] \), the Hadamard test circuit produces an amplitude of 
    \[ \sqrt{(1 + \bra{\psi_0} \tilde{O}_{2M+1, s_k} \ket{\psi_0} / \gamma ) / 2}, \] 
    using one query to \( U_O \) and \(U_{\psi}\), and \( 1/s_k \) Trotter steps. Then we can use the Iterative Quantum Amplitude Estimation (IQAE) \cite{grinko2021iterative} to estimate the amplitude. By Lemma~\ref{lem:iqae}, \begin{align*}
        N_{\mathrm{rep}} = O\left(\dfrac{\gamma\|\mathbf{b}\|_1}{\varepsilon \|O\|}\log\left(\frac{1}{m}\log\left(\frac{\gamma\|\mathbf{b}\|_1}{\varepsilon \|O\|}\right)\right)\right) = \widetilde{O}\left(\frac{\gamma}{\varepsilon \|O\|}\right) = \widetilde{O}(1/\varepsilon)
    \end{align*}
    queries to the Hadamard circuit suffice to estimate the probability to within accuracy $\varepsilon\|O\|/(4\gamma\|\mathbf{b}\|_1)$ with success probability at least $1-1/(3m)$, which gives an $\varepsilon\|O\|/(2\|\mathbf{b}\|_1)$-approximation of $\bra{\psi_0} \tilde{O}_{2M+1, s_k} \ket{\psi_0}$. Linearly combining the estimates of \( \bra{\psi_0} \tilde{O}_{2M+1,s_k} \ket{\psi_0} \), weighted by \( b_k \), yields an \( \varepsilon \|O\|/2 \)-approximation of \( \bra{\psi_0} \tilde{O}^{(m)}_{2M+1,s_0} \ket{\psi_0} \) with success probability at least $2/3$. Therefore, the total error is bounded by $\|O\|\varepsilon$ and the maximum quantum circuit depth and total time complexity are both 
    \[\widetilde{O}\left(\frac{p\tau_H\Gamma_{\mathrm{avg}}\Upsilon^{2+1/p} \big(a_{\max} M  \lambda_{\rm comm} \big)^{1+1/p}+\tau_{\mathrm{sum}}}{\varepsilon}\right).
    \]
    The number of queries to $U_O$ is $\widetilde{O}(m/\varepsilon) = \widetilde{O}(p/\varepsilon)$. 

    Besides the ancilla needed for the block encoding of $O$ by the unitary $U_O$, the Hadamard test uses one additional ancilla qubit. To implement the reflection around \( \ket{\psi_0} \) in IQAE, a Toffoli gate is needed, which uses one additional ancilla qubit as the target qubit. The Toffoli gate can be implemented using \( O(n^2) \) elementary gates \cite{saeedi2013linear} without ancilla qubits, or \( O(n) \) elementary gates with one ancilla qubit \cite[Corollary 7.4]{barenco1995elementary}. %\xz{check if some ancillas can be reused.} 
\end{proof}

\begin{breakablealgorithm}
\caption{Interleaved Hamiltonian evolution with Trotter [Theorems \ref{thm:qsp-with-trotter-and-interpolation} and \ref{thm:qsp-with-trotter-and-interpolation-coherent}]}
\begin{algorithmic}[1]
\REQUIRE $M$ Hermitian operator $H^{(\ell)}=\sum_{\gamma=1}^{\Gamma_{\ell}}H_{\gamma}^{(\ell)}$. 

\qquad  An initial state $\ket{\psi_0}$  and an observable $O$.

\qquad $M+1$ unitaries $V_0,\ldots,V_M$.

\qquad $W :=  V_0\prod_{\ell=1}^M e^{\i H^{(\ell)}}V_{\ell}$.

\qquad A $p$-th order symmetric staged product formula $\mathcal{P}^{(\ell)}$ for each $H^{(\ell)}$ with $\Upsilon$ stages, maximum 

\qquad coefficient  $a_{\max}$ satisfying $a_{\max}\Upsilon M \lambda_{\rm comm} \ge 1$.

\qquad $m=p\lceil\log(1/\varepsilon)\rceil, K = 2\lceil\log(1/\varepsilon)\rceil$, $\lambda_{\rm comm}$ is given via \eqref{Lambda}.

\ENSURE $\bra{\psi_0} W^{\dagger} O W \ket{\psi_0} \pm \varepsilon\|O\|$.

\STATE Choose $s_0 =O(m (a_{\max} \Upsilon M \lambda_{\rm comm})^{-1-1/p})$ such that $1/s_0\in M \mathbb{N}$.

\STATE Compute $b_1,\ldots,b_m, s_1,\ldots,s_m$ via Lemma \ref{lem_Richad for Trotter}.

\STATE For $k=1,\ldots,m$

Compute $\tilde{f}(s_k) := \bra{\psi_0} \big(V_0\prod_{\ell=1}^M \mathcal{P}^{(\ell)}(s_k M)^{1/(s_k M)}V_{\ell}\big)^{\dagger} O \big(V_0\prod_{\ell=1}^M \mathcal{P}^{(\ell)}(s_k M)^{1/(s_k M)}V_{\ell} \big)\ket{\psi_0}$ on a quantum computer.

\STATE Output $\sum_{k=1}^m b_k \tilde{f}(s_k)$.
\end{algorithmic}
\label{alg: HSVT with trotter}
\end{breakablealgorithm}

Finally, Algorithm \ref{alg: HSVT with trotter} summarizes the method encompassed in Theorem \ref{thm:qsp-with-trotter-and-interpolation} and Theorem \ref{thm:qsp-with-trotter-and-interpolation-coherent}. The framework introduced here is very general. As mentioned previously, $V_0$ could be a quantum algorithm that precedes the circuit $W$. Similarly, $W$ could be used as input to another quantum algorithm, with circuit $V_M$. This flexibility allows for delaying the measurement of the observable $O$ (whether coherent or incoherent) until the end of the entire algorithm, of which $W$ may be a subroutine. In the next section, we develop methods to implement QSVT without block encodings -- the key observation is that it is possible to formulate QSVT as a particular instance of $W$. However, the applications of this framework extend beyond QSVT.

\section{QSVT using higher order Trotterization}
\label{sec:qsvt-trotter}

In this section, we prove in detail our quantum algorithms to implement QSVT using higher-order Trotterization. Theorem \ref{thm:qsp-with-trotter-and-interpolation} provides a general framework to estimate the expectation values of observables with respect to the output state of any quantum circuit that is an interleaved sequence of arbitrary unitaries and Hamiltonian evolution operators. Here, we adapt the framework of Generalized Quantum Signal Processing (GQSP) to our setting. Recall from Sec.~\ref{subsec:prelim-gqsp},  the basic building blocks are $c_0\text{-}U$ and $c_1\text{-}U^\dag$ (see Eq.~\eqref{control U:eq}). As $U = e^{\i H}$, we can rewrite $c_0\text{-}U = e^{\i \tilde{H}_0}, c_1\text{-}U^\dag = e^{\i \tilde{H}_1}$ with $\tilde{H}_0 = {\rm diag}(H,0), \tilde{H}_1 = {\rm diag}(0,-H)$. Thus, the circuit in Lemma \ref{thm:generalized-QSP} can be reformulated as an interleaved sequence of unitaries and Hamiltonian evolution operators --- our general framework in Sec.~\ref{sec:interleaved-sequence}. Indeed, we simply replace $V_j=R(\theta_j,\phi_j,\gamma_j)\otimes I$, and $H^{(j)}=H$, for all $j$. This completely eliminates the need for any block encoding access to $H$ and allows us to use Theorem \ref{thm:qsp-with-trotter-and-interpolation}. 

Consider any $n$-qubit Hamiltonian $H$ that can be written as
$$
H=\sum_{j=1}^{L} H_j.
$$
Here, $H$ can be a sum of local terms, or it can be a linear combination of $n$-qubit strings of Pauli operators, i.e.,
$H_j=\lambda_j P_j$. We define $\lambda:=\sum_{j=1}^{L}\|H_j\|$. Thus, given any such $H$ , initial state $\ket{\psi_0}$, and observable $O$, we can estimate the expectation value 
$$
\braket{\psi_0|f(H)^{\dag} O f(H)|\psi_0},
$$ 
to arbitrary accuracy, without using any block encodings and while using only a single ancilla qubit. Naturally, this method can implement any function $f(x)$ that can be approximated by a Laurent polynomial $P(e^{ix})$ (bounded in $\mathbb{T}$). We simply use the $p$-th order Trotter method to implement (controlled Hamiltonian evolution). Formally, we have the following theorem, where we have considered $2k$-order Suzuki Trotter formula:

\begin{thm}[Generalized QSP with Trotter, incoherent estimation]
\label{thm_gqsp-trotter-suzuki}
Given $B\in (0, \pi]$, let $H=\sum_{j=1}^L H_j$ be a Hermitian decomposition of $H$ such that $H_j$ is local and $\|H\|\le B$. Let $k\in \mb N$, $\varepsilon\in(0,1/2)$ be the precision parameter, and $f(x)$ be a function on $[-B,B]$. 
Suppose that there exists a degree-$d$ Laurent polynomial $P(z)=\sum_{j=-d}^d a_j z^j$ satisfying $|f(x)-P(e^{\i x})|\le \eps/6$ for all $x\in [-B,B]$ and $|P(z)|\le 1$ for all $|z|=1$. Assume the initial state $\ket{\psi_0}$ can be prepared efficiently, then there exists a quantum algorithm that, for any observable $O$, estimates $\bra{\psi_0} f(H)^\dag O f(H) \ket{\psi_0}$ with an error of at most $\varepsilon\|O\|$ and a success probability of at least $2/3$, using one ancilla qubit. The maximum quantum circuit depth is
\be
\widetilde{O}\left((25/3)^k k^2 L \big(d \lambda_{\rm comm}\big)^{1+1/2k} \right),
\ee
and the total time complexity is
\be
\widetilde{O} \left((25/3)^k k^2 L \big(d \lambda_{\rm comm}\big)^{1+1/2k}/\varepsilon^2 \right),
\ee
where $\lambda_{\mathrm{comm}}$ defined in Eq.~\eqref{Lambda} for all $H^{(\ell)}=H$ scales with the nested commutators of $H$.
\end{thm}
\begin{proof}
    Since $\|H\| \le B$ and $|f(x)-P(e^{\i x})|\le \eps/2$ for all $x\in [-B,B]$, we have $\|f(H)-P(e^{\i H})\|\le \eps/6$. The operator norm of $P(e^{\i H})$ is upper bounded by $1$ since $|P(z)|\le 1$ for $|z|=1$. From Theorem~\ref{thm:distance-expectation}, we have \[
        |\bra{\psi_0} P(e^{\i H})^\dag O P(e^{\i H}) \ket{\psi_0}-\bra{\psi_0} f(H)^\dag O f(H) \ket{\psi_0}| \le \varepsilon\|O\|/2.
    \]
    Therefore, it suffices to estimate $\bra{\psi_0} P(e^{\i H})^\dag O P(e^{\i H}) \ket{\psi_0}$ to within $\eps\|O\|/2$.

    By Theorem~\ref{thm:generalized-QSP}, there exist $\Theta=(\theta_j)_j, \Phi=(\phi_j)_j \in \mathbb{R}^{2d+1}, \lambda \in \mathbb{R}$ such that \begin{align*}
        \begin{bmatrix}
            P(e^{\i H}) & * \\
            * & * 
        \end{bmatrix} = \Big(\prod_{j=1}^d \Big( R(\theta_{d+j}, \phi_{d+j}, 0) \cdot c_1\text{-}e^{-\i H} \Big) \Big)\Big(\prod_{j=1}^{d}  \Big(R(\theta_j, \phi_j, 0) \cdot c_0\text{-}e^{\i H} \Big) \Big) \cdot R(\theta_0, \phi_0, \lambda).
    \end{align*}
    Since $c_0\text{-}e^{\i H}=e^{\i (\ket{0}\bra{0}\otimes H)}$ and $c_1\text{-}e^{-\i H}=e^{\i (-\ket{1}\bra{1}\otimes H)}$ are Hamiltonian evolution operators, the above GQSP circuit can be regarded as a interleaved sequence of Pauli rotations and Hamiltonian evolutions. The LHS of the above equation is an $(n+1)$-qubit unitary.
    
    We use the $2k$-th order Trotter-Suzuki formula $S_{2k}(t)$ to approximately implement the Hamiltonian evolution. The formula $S_{2k}(t)$ has $\Upsilon=2\cdot 5^{k-1}$ stages and maximum parameter $a_{\max}= 2k/3^k$ \cite{wiebe2010higher}. Note that each Pauli rotation operator can be implemented using $\widetilde{O}(1)$ gates. This implies that $\tau_H$ and $\tau_{\mathrm{sum}}$ in  Theorem~\ref{thm:qsp-with-trotter-and-interpolation} are $\widetilde{O}(1)$ and $\widetilde{O}(d)$, respectively. 
    Simplifying the complexity in Theorem \ref{thm:qsp-with-trotter-and-interpolation} for $S_{2k}(t)$ gives the claimed results.
\end{proof}
We only consider Hamiltonians with norms bounded by $\|H\| \le \pi$ in Theorem \ref{thm_gqsp-trotter-suzuki} for simplicity. This condition ensures that the spectrum of $H$ lies within a single period of $P(e^{\i x})$. However, this method is not limited to small-norm Hamiltonians. For a general Hamiltonian $H$ with $\|H\| > \pi$, we employ a rescaling strategy. Let $\lambda_{\mathrm{comm}}$ be the commutator bound for $H$. We define a rescaled Hamiltonian $H' = H/\|H\|$, a corresponding rescaled function $g(x) = f(x\|H\|)$, and then apply our algorithm to $H'$ and $g(x)$.
This rescaling has two competing effects on the complexity. First, the commutator bound for $H'$ becomes $\lambda'_{\mathrm{comm}} = \lambda_{\mathrm{comm}}/\|H\|$. Second, the degree of the Laurent polynomial required to approximate the more rapidly oscillating function $g(x)$ increases, typically scaling linearly with $\|H\|$. Let $d$ be the degree for approximating the original function $f(x)$; the new degree $d'$ for $g(x)$ is proportional to $d \cdot \|H\|$ (see Appendix \ref{sec-app:function-laurent} for examples).
The overall complexity of estimating $\bra{\psi_0}f(H)^{\dag} O f(H) \ket{\psi_0}$ scales with the product of these terms. By applying Theorem \ref{thm_gqsp-trotter-suzuki} to $H'$ and $g(x)$, the complexity scales as 
$$(d' \lambda'_{\mathrm{comm}})^{1+1/2k} \propto ((d\|H\|)(\lambda_{\mathrm{comm}}/\|H\|))^{1+1/2k} = (d\lambda_{\mathrm{comm}})^{1+1/2k}.$$ 
Thus, the $\|H\|$ dependence cancels out, and we recover the same complexity scaling for any general $H$. The comparison of the cost of our method with standard QSVT can be seen in Table \ref{table:comparison-qsvt}.

Often, for most applications, we are interested in estimating the expectation value of the observable $O$ with respect to the normalized state
$$
\ket{\psi}=\dfrac{f(H)\ket{\psi_0}}{\|f(H)\ket{\psi_0}\|}.
$$
Indeed, for any observable $O$, we can estimate $\braket{\psi|O|\psi}$ by using fixed point amplitude amplification \cite{yoder2014fixed} to obtain the normalized state $\ket{\psi}$, and then measure $O$ (incoherently). Note that with this modification, we obtain a new circuit that is an interleaved sequence of unitaries and Hamiltonian evolution operators. Thus, with only one additional ancilla qubit, we can estimate the desired expectation value. Formally, we have the following theorem:

\begin{thm}[Generalized QSP with Trotter and amplitude amplification]
    \label{thm:gqsp-trotter-suzuki-aa}
    %Under the same assumptions as in Theorem~\ref{thm_gqsp-trotter-suzuki}.
    Given $B\in (0, \pi]$, let $H=\sum_{j=1}^L H_j$ be a Hermitian decomposition of $H$ such that $H_j$ is local and $\|H\|\le B$. Let $k\in \mb N$, $\varepsilon\in(0,1/2)$ be the precision parameter, and $f(x)$ be a function on $[-B,B]$. Let $\ket{\psi} = f(H) \ket{\psi_0}/\|f(H) \ket{\psi_0}\|$ and assume $\eta\leq \|f(H)\ket{\psi_0}\|$.
    Suppose that there exists a degree-$d$ Laurent polynomial $P(z)=\sum_{j=-d}^d a_j z^j$ satisfying $|f(x)-P(e^{\i x})|\le \eps/(12\eta)$ for all $x\in [-B,B]$ and $|P(z)|\le 1$ for all $|z|=1$.  Assume the initial state $\ket{\psi_0}$ can be prepared efficiently, then there exists a quantum algorithm that, for any observable $O$, estimates $\bra{\psi} O \ket{\psi}$ with an error of at most $\varepsilon\|O\|$ and a success probability of at least $2/3$ using two ancilla qubits. The maximum quantum circuit depth is
    \be
    \widetilde{O}\left((25/3)^k k^2 L \big(d \lambda_{\rm comm}/\eta\big)^{1+1/2k} \right),
    \ee
    and the total time complexity is
    \be
    \widetilde{O} \left((25/3)^k k^2 L \big(d \lambda_{\rm comm}/\eta\big)^{1+1/2k}/\varepsilon^2 \right),
    \ee
    $\lambda_{\mathrm{comm}}$ defined in Eq.~\eqref{Lambda} for all $H^{(\ell)}=H$ scales with the nested commutators of $H$.
\end{thm}
\begin{proof}
   Suppose the initial state \( \ket{\psi_0} \) is a state of \( n \) qubits. 
   Since $\|H\| \le B$ and $|f(x)-P(e^{\i x})|\le \eps/(12\eta)$ for all $x\in [-B,B]$, we have $\|f(H)-P(e^{\i H})\|\le \eps/(12\eta)$. Let the normalized $P(e^{\i H})\ket{\psi_0}$ be $$
    \ket{\tilde{\psi}} = \dfrac{P(e^{\i H}) \ket{\psi_0}}{\| P{(e^{\i H})} \ket{\psi_0} \|}.
    $$ The distance between the two normalized states is \[
   \|\ket{\psi}-\ket{\tilde{\psi}}\|=\Bigg\|\frac{f(H)\ket{\psi_0}}{\|f(H)\ket{\psi_0}\|}-\frac{P(e^{\i H})\ket{\psi_0}}{\|P(e^{\i H}\ket{\psi_0}\|}\Bigg\|\le \frac{2\|f(H)\ket{\psi_0}-P(e^{\i H})\ket{\psi_0}\|}{\eta}\le \frac{\eps}{6}.
   \]
   From Theorem~\ref{thm:distance-expectation}, we have \[
        \left |\bra{\tilde{\psi}} O P(e^{\i H}) \ket{\tilde{\psi}}-\bra{\psi}O \ket{\psi}\right | \le \varepsilon\|O\|/2.
    \] Therefore, it suffices to estimate $\bra{\tilde{\psi}}O\ket{\tilde{\psi}}$ to within $\eps\|O\|/2$. 
    
   By Theorem~\ref{thm:generalized-QSP}, there exists a length-$(2d+1)$ interleaved circuit sequence comprising of single qubit $U(2)$ rotations and Hamiltonian evolution, \( U_P \), such that
    \begin{align*}
       (\bra{0} \otimes I_n) U_P \ket{0}\ket{\psi_0} = P(e^{\i H})\ket{\psi_0}.
    \end{align*}
    Using fixed-point amplitude amplification from \cite{yoder2014fixed}, we can prepare an \( \varepsilon/2 \)-approximation of $\ket{\tilde{\psi}}$
    using an interleaved sequence of operations 
   $$
    U_P(I_2\otimes U_{\psi_0}) e^{\i \alpha_j \ket{0^{n+1}}\bra{0^{n+1}}} (I_2\otimes U_{\psi_0}^{\dagger}) U_P^{\dagger} \quad \text{and} \qquad e^{\i \beta_j \ket{0}\bra{0} \otimes I_n},
    $$
    for \( O(\log(1/\varepsilon) / \eta) \) values of \( \alpha_j, \beta_j \in \mathbb{R} \). 
    
    The resulting quantum circuit is still an interleaved Hamiltonian sequence with length 
    $$
    O(d \log(1/\varepsilon) / \eta),
    $$ as \( U_P \) is a length \( 2d+1 \) interleaved Hamiltonian evolution sequence. % The circuit uses \( O(\log(1/\varepsilon) / \eta) \) queries to \( U_{\psi_0} \) and \( U_{\psi_0}^{\dag} \). 
    
    Therefore, we can use Algorithm~\ref{alg: HSVT with trotter} to estimate \( \bra{\psi} O \ket{\psi} \) to within error \( \varepsilon \|O\| \), and simplifying the complexity in Theorem~\ref{thm_gqsp-trotter-suzuki} yields the results. To implement  
    $$
    e^{\i \beta_j \ket{0}\bra{0} \otimes I_n}\qquad \text{and} \qquad e^{\i \alpha_j \ket{0^{n+1}}\bra{0^{n+1}}},
    $$
    we can reuse one additional ancilla qubit as the target qubit for all the \((n+2)\)-qubit Toffoli gates by uncomputing the target qubit after each reflection. Thus, we use only two ancilla qubits overall.
\end{proof}
Finally, it is also possible to estimate the expectation value coherently using quantum amplitude estimation. So, analogous to Theorem \ref{thm:qsp-with-trotter-and-interpolation-coherent}, we obtain

\begin{thm}[Generalized QSP with Trotter, coherent estimation]
\label{thm:gqsp-trotter-coherent measurement}
    Under the same assumption as in Theorem \ref{thm_gqsp-trotter-suzuki}, if we further assume $O/\gamma$ is block-encoded via a unitary $U_O$ for some $\gamma$ satisfying $\gamma = {\Theta}(\|O\|)$, there exists a quantum algorithm that estimates $\bra{\psi_0}f(H)^\dag  O f(H) \ket{\psi_0}$ with an error of at most $\varepsilon\|O\|$ and a success probability of at least $2/3$ using three ancilla qubits. The time complexity and maximum quantum circuit depth are both
    \[
    \widetilde{O}\left((25/3)^k k^2 L \big(d \lambda_{\rm comm}\big)^{1+1/2k}/\varepsilon \right),
    \]
    $\lambda_{\mathrm{comm}}$ defined in Eq.~\eqref{Lambda} for all $H^{(\ell)}=H$ scales with the nested commutators of $H$.
    The number of queries to \( U_O \) and its inverse is $\widetilde{O}(k/\varepsilon)$. 
    \end{thm}
    
\begin{proof}
Choosing the symmetric staged product formula $\mathcal{P}(t)$ as $S_{2k}(t)$ and simplifying the complexity in Theorem \ref{thm:qsp-with-trotter-and-interpolation-coherent} gives the claimed results.   %\xz{Can we reuse the ancilla of GQSP and Toffoli gates?}
\end{proof}

So, overall, our method allows us to implement functions that can be approximated by a Laurent polynomial in $e^{ikx}$. This naturally raises the question: Which functions can be efficiently approximated in this form? In Appendix~\ref{sec-app:function-laurent}, we show that many commonly used functions—including the shifted sign function, rectangle function, inverse function, and exponential function—can indeed be approximated by Laurent polynomials (bounded in $\mathbb{T}$) in $e^{\i kx}$ for some constant $k$.

In contrast, standard QSVT implements functions expressed as polynomials in $x$ that are bounded on $[-1,1]$. We identify precise conditions under which such a polynomial \( Q(x) \) can be approximated by a  Laurent polynomial in $e^{\i k x}$, \( P(e^{\i k x}) \), bounded in $\mathbb{T}$, for some constant \( k \), with the degree of \( P \) scaling linearly with the degree of \( Q \). 

In particular, we prove that if a function \( f(x) \) can be approximated by a $d$-degree polynomial $Q$ on \([-1,1]\) that remains bounded by 1 on the extended interval \([-1 - \delta, 1 + \delta]\) for some \( \delta \in (0, 1/2] \), then there exists a Laurent polynomial in $e^{\i \pi x / (2(1 + \delta))}$, \( P(e^{\i \pi x / (2(1 + \delta))}) \), of degree $\widetilde{O}(d/\delta^2)$, that approximates \( f(x) \) on $[-1,1]$ and is bounded by 1 on the complex unit circle \( \mathbb{T} \). Formally, we prove the following:

\begin{prop}
    \label{lem:gqsp-poly2tri}
    Let $\varepsilon \in (0,1)$ and $\delta \in (0, 1/2]$. Suppose $f \colon \mathbb{R} \to \mathbb{R}$ is approximated by a degree-$d$ polynomial $Q$ satisfying
    \begin{align*}
        &|f(x) - Q(x)| \le \varepsilon/2 &&\text{for all } x \in [-1,1], \\
        &|Q(x)| \le 1 &&\text{for all } x \in [-1 - \delta, 1 + \delta].
    \end{align*}
    Then there exists an Laurent polynomial $P(z)$ of $\deg(P) =O(d\log(d/\varepsilon)/\delta^2)$ such that \begin{align*}
        &|P(z)|\le 1 && \text{for all } z\in \mathbb{T},\\
        &|P(e^{\i \pi x/(2(1+\delta))})-f(x)|\le \varepsilon && \text{for all } x \in [-1,1].
    \end{align*}
\end{prop}

\begin{proof}
    Notice that for any polynomial $\tilde{P}(y)$ in $y=\sin (\pi x/(2(1+\delta)))$, we can substitute 
    $$\sin \left(\dfrac{\pi x}{2(1+\delta)}\right) = \dfrac{1}{2\i}\left(e^{\i \pi x/(2(1+\delta))} - e^{-\i \pi x/(2(1+\delta))}\right),$$ 
    into the polynomial to obtain a Laurent polynomial $P(z)$ in $z=e^{\i \pi x/(2(1+\delta))}$, of the same degree as $\tilde{P}$. Moreover, $|P(z)|\le 1$ for all $z\in \mathbb{T}$ is equivalent to $|\tilde{P}(y)|\le 1$ for all $y\in [-1, 1]$. Following the proof of \cite[Proposition B5]{chakraborty2025quantum}, we can construct such a polynomial $\tilde{P}$ satisfying \begin{align*}
        &|\tilde{P}(y)|\le 1 && \text{ for all } y \in [-1,1] , \\
        &|\tilde{P}(\sin(\pi x/(2(1+\delta))))-f(x)|\le \varepsilon && \text{ for all } x\in [-1,1], 
    \end{align*}
    of degree $O(d\log(d/\varepsilon)/\delta^2)$. This yields the desired Laurent polynomial \(P(z)\).
\end{proof}

By imposing stronger smoothness conditions on the polynomial \( Q \), and after an appropriate rescaling, the following result from \cite{van2017quantum} demonstrates that a quadratic improvement in the degree dependence on \( \delta \) can be achieved. 

\begin{prop}[Lemma 37 of \cite{van2017quantum}, rescaled version]
    \label{lem:bounded-1-norm-gqsp}
    Let $\varepsilon, \delta\in (0,1)$. Suppose $f(x) \colon \mathbb{R}\to\mathbb{R}$ admits a polynomial approximation  $Q(x)=\sum_{k=0}^K a_k x^k$ satisfying 
    \begin{align*}
        |f(x)-Q(x)|\le \varepsilon/2  \quad \text{ for all } x\in [-1,1].
    \end{align*}
    Let $B = \sum_{k=0}^{K} |a_{k}|(1+\delta)^{k}$. Then there exists a Laurent polynomial $P(z)$ of $\deg(P) = O(\log(B/\varepsilon)/\delta)$  such that \begin{align*}
        &|P(z)|\le B && \text{for all } z\in \mathbb{T},\\
        &|P(e^{\i \pi x/(2(1+\delta))})-f(x)|\le \varepsilon && \text{for all } x \in [-1,1].
    \end{align*}
\end{prop}

Together, Propositions~\ref{lem:gqsp-poly2tri} and~\ref{lem:bounded-1-norm-gqsp} establish a formal connection between standard QSVT and our Laurent polynomial-based approach, highlighting how bounded polynomial transformations in the QSVT framework can be translated into our setting.

Table~\ref{table:comparison-our-qsvt-methods} presents a detailed comparison between the complexities of our method and standard QSVT for estimating the expectation value of an observable $O$ with respect to the normalized quantum state
$$
\ket{\psi}=\dfrac{f(H)\ket{\psi_0}}{\|f(H)\ket{\psi_0}\|}.
$$
Both frameworks offer multiple strategies to perform this estimation. Notably, our approach achieves comparable asymptotic complexity to standard QSVT while avoiding the need for block-encodings of $H$, and requiring only a constant number of ancilla qubits. Below, we outline the complexity of each approach (assuming an appropriate choice of a large enough $k$) and highlight its key resource trade-offs:

\begin{itemize}
    \item \textbf{Direct estimation without amplitude amplification or estimation:}~This approach runs the interleaved circuit composed of controlled Hamiltonian evolutions and single-qubit rotations, followed by measurement of the observable $O$. The circuit depth per coherent run is determined by Theorem~\ref{thm_gqsp-trotter-suzuki}, and the number of independent classical repetitions is $\widetilde{O}(\eta^{-2}\eps^{-2})$. Both the depth and total time complexity nearly match standard QSVT. However, unlike standard QSVT, which requires $\lceil \log_2 L\rceil +1$ ancilla qubits (and complicated controlled logic), this method uses only a single ancilla qubit. This makes it especially appealing for early fault-tolerant quantum devices, where running many shallow-depth circuits is more practical than maintaining deep coherence.

    \item \textbf{Using fixed-point amplitude amplification:}~This variant composes the fixed-point amplitude amplification circuit with the one described above. The per-run circuit depth is given by Theorem~\ref{thm:gqsp-trotter-suzuki-aa}, and the number of repetitions reduces to $\widetilde{O}(\eps^{-2})$. While the complexity nearly matches standard QSVT, our method again requires only two ancilla qubits, a significant reduction compared to the logarithmic ancilla overhead in standard QSVT.

    \item \textbf{Iterative quantum amplitude estimation:~}This approach, formalized in Theorem~\ref{thm:gqsp-trotter-coherent measurement}, further reduces the number of repetitions by using coherent amplitude estimation, at the cost of increased circuit depth. Despite the deeper circuits, this method offers a reduced overall time complexity. As before, the ancilla overhead of our method remains minimal (three), offering similar advantages over standard QSVT.
\end{itemize}

Thus, we obtain a simple and resource-efficient method for implementing polynomial transformations, while retaining near-optimal complexity. In the following section (Sec.~\ref{sec:applications}), we apply these results to both of our target applications.

\section{Applications}
\label{sec:applications}
\subsection{Quantum linear systems via discrete adiabatic theorem without block encodings}
\label{subsec:qls-solver}

The quantum linear systems algorithm can be stated as follows: Consider a Hermitian matrix $A\in\mathbb{C}^{N\times N}$ with $\|A\|=1$, whose eigenvalues lie in the interval $[-1, -1/\kappa]\cup [1/\kappa, 1]$. Here, $\|A^{-1}\|=\kappa$ is the condition number of $A$. 
Given access to a procedure that prepares an initial state $\ket{b}$, the quantum linear systems algorithm prepares the quantum state $\ket{x}=A^{-1}\ket{b}/\|A^{-1}\ket{b}\|$. In many applications, simply preparing the state $\ket{x}$ may not be sufficient. Instead, one is often interested in extracting useful information from this state, such as estimating the expectation value of an observable $O$, i.e., \ $\braket{x|O|x}$.

Since the first quantum algorithm for this problem by Harrow, Hassidim, and Lloyd \cite{harrow2009quantum}, the quantum linear systems algorithm has been widely studied. 
The complexity of this algorithm has been progressively improved through a series of results based on quantum linear algebra techniques \cite{childs2017qls, chakraborty_et_al:LIPIcs.ICALP.2019.33, gilyen2019quantum, low2024qep, low2024qls}. 
The disadvantage of these algorithms is that they require a block encoding access to $A$ and achieve an optimal linear dependence on the condition number $\kappa$ only with the help of complicated subroutines such as variable time amplitude amplification (VTAA) \cite{ambainis2012variabletime}. 
More recently, adiabatic-inspired approaches have also been explored \cite{subacsi2019quantum, lin2020optimal}, which have managed to optimally solve this problem \cite{costa2022optimal}, without the use of VTAA. However, these techniques continue to assume block encoding access to $A$. 

Among these, the algorithm of Costa et al. \cite{costa2022optimal} achieves optimal query complexity. More precisely, given access to a block encoding of $A$, the algorithm makes use of the discrete adiabatic theorem \cite{dranov1997discrete} to prepare a state that is $\eps$-close to $\ket{x}$, with a query complexity of $O(\kappa \log (1/\eps))$.

In this work, we simplify this algorithm by eliminating the need for block encoding and reducing the ancilla qubit requirement to a constant. Despite this, the complexity has a near-optimal (quasi-linear) dependence on $\kappa$ and a polylogarithmic dependence on the inverse precision. We achieve this by demonstrating that the simulation of the discretized adiabatic theorem can be seen as a particular instance of the interleaved sequence circuit described in Sec.~\ref{sec:interleaved-sequence}. Moreover, unlike \cite{costa2022optimal}, we do not require the construction of a qubitized quantum walk operator. We achieve near-optimal complexity by simply using higher-order Trotterization. We begin by providing a brief overview of discrete adiabatic computation, stating a key result from \cite{costa2022optimal}, which we make use of. 

\subsubsection{A brief overview of the discrete adiabatic theorem}

In adiabatic quantum computation \cite{farhi2000quantum, farhi2001adiabatic2}, one considers a time‐dependent Hamiltonian $H(s)$ with $s\in [0,1]$, where typically, the ground state of \(H(0)\) is easy to prepare while the ground state of \(H(1)\) encodes the solution to the problem. 
The system begins in the ground state of $H(0)$, and \emph{slowly (adiabatically) evolves}, over a sufficiently long time $T$, such that by the end of the evolution, the system remains $\eps$-close to the ground state of $H(1)$. 
The precise meaning of ``a sufficiently long time $T$" is given by the rigorous version of the adiabatic theorem \cite{ambainis2004elementary, jansen2007bounds, elagrt2012note}. 
Informally, if the time of evolution $T$ is larger than the inverse of the minimum spectral gap of $H(s)$ for any $s\in [0,1]$, the system will remain in the instantaneous ground state throughout the adiabatic evolution. 
The adiabatic theorem holds for any eigenstate of the underlying Hamiltonian, with the key parameter being the energy gap between the eigenstate and the next closest energy level. 

Discrete adiabatic computation (DAC), introduced by \cite{dranov1997discrete}, is a discretization of the continuous-time adiabatic evolution. This is useful for simulating any adiabatic quantum computation in the circuit model. In the DAC framework, there are $T$ unitary operators 
$$
B(0/T), \, B(1/T), \, \ldots, \, B((T-1)/T),
$$ 
where each $B(j/T)$ approximates the ideal (continuous-time) adiabatic evolution at $s=j/T$. 
Overall, define the unitary 
\begin{align*}
    U(s)=\prod_{m=0}^{s T-1} B(m / T),
\end{align*}
with $U(0)=I$. If $U_A$ is the ideal adiabatic evolution for some time $T$, the DAC framework provides a concrete bound on $\|U_A-U(s)\|$. 
In order to state this difference, we define a series of useful quantities. 
First, define the $k$-th difference of $B(s)$ as \begin{align*}
D^{(k)} B(s) & := D^{(k-1)} B\Big(s+\frac{1}{T}\Big)-D^{(k-1)} B(s), \\
D^{(1)} B(s) & := B\Big(s+\frac{1}{T}\Big)-B(s) ,
\end{align*} 
where $k\in\mb N$.
Then, assume that there exist functions $c_1(s), c_2(s)$ such that $D^{(1)}, D^{(2)}$ are bounded as (see \cite{an2022quantum})
\begin{align}
    \label{eq:smooth-W}
    \|D^{(1)} B(s)\| \leq \frac{c_1(s)}{T}, \quad \|D^{(2)} B(s)\| \leq \frac{c_2(s)}{T^2} 
\end{align}
for all $T>0$. Also define $\hat{c}_k(s)$ as
\begin{align}
    \label{eq:def-c_k}
    \hat{c}_k(s):= \max _{s^{\prime} \in\{s-1 / T, s, s+1 / T\} \cap[0,1-k / T]} c_k(s^{\prime})
\end{align} for $k=1,2$, which takes the neighboring steps into account. 

Suppose eigenvalues of $B(s)$ are partitioned into two parts: the eigenvalues of interest $\sigma_\Pi(s)$ and the remaining ones $\sigma_{\widetilde{\Pi}}(s)$. 
Define $\Pi(s)$ to be the projection onto the eigenspace of $B(s)$ associated with the eigenvalues in $\sigma_\Pi(s)$. 
Since $B(s)$ is unitary, all its eigenvalues lie on the unit circle, which allows us to obtain the gaps between the different eigenvalues. 
For a non-negative integer $k$, define $\Delta_k(s)$ as the minimum angular distance between arcs $\sigma_\Pi^{(k)}$ and $\sigma_{\widetilde{\Pi}}^{(k)}$, which satisfy
\begin{align*}
\sigma_\Pi^{(k)} \supseteq \bigcup_{l=0}^k \sigma_\Pi(s+l / T), \quad \sigma_{\widetilde{\Pi}}^{(k)} \supseteq \bigcup_{l=0}^k \sigma_{\widetilde{\Pi}}(s+l/T) .
\end{align*}
Define the gap 
\begin{align}
    \label{eq:def-gap}
    \Delta(s)= \begin{cases}\Delta_2(s), & 0 \leq s \leq 1-2 / T \\ \Delta_1(s), & s=1-1 / T \\ \Delta_0(s), & s=1\end{cases}
\end{align}
to be the minimum gap among eigenvalues in the three successive steps, except at the boundary cases. 
Finally, define 
\begin{align}
    \label{eq:def-neighbor-gap}
    \check{\Delta}(s)=\min _{s^{\prime} \in\{s-1 / T, s, s+1 / T\} \cap[0,1]} \Delta(s^{\prime}),
\end{align}
as the minimum gap of neighboring steps. 
As mentioned previously, if the adiabatic evolution operator corresponds to $U_A(s)$, the discrete adiabatic theorem can be formally stated as follows: 
\begin{thm}[Theorem 3 of \cite{costa2022optimal}]
    \label{thm:dqc}
    Suppose that the operators $W(s)$ satisfy $\|D^{(k)} W(s)\| \le c_k(s) / T^k$ for $k=1,2$, as in Eq.~\eqref{eq:smooth-W}, and $T \geq \max _{s \in[0,1]}\{4 \hat{c}_1(s) / \check{\Delta}(s)\}$. 
    Then, for any time $s$ such that $s T \in \mathbb{N}$, we have
    \begin{align*}
    \|U(s)-U_A(s)\| \leq & \;\frac{12 \hat{c}_1(0)}{T \check{\Delta}(0)^2}+\frac{12 \hat{c}_1(s)}{T \check{\Delta}(s)^2}+\frac{6 \hat{c}_1(s)}{T \check{\Delta}(s)} +305 \sum_{n=1}^{s T-1} \frac{\hat{c}_1(n / T)^2}{T^2 \check{\Delta}(n / T)^3} \\
    &\; +44 \sum_{n=0}^{s T-1} \frac{\hat{c}_1(n / T)^2}{T^2 \check{\Delta}(n / T)^2}  +32 \sum_{n=1}^{s T-1} \frac{\hat{c}_2(n / T)}{T^2 \check{\Delta}(n / T)^2},
    \end{align*}
    where $\hat{c}_k(s)$ and $\check{\Delta}(s)$ are defined in Eqs.~\eqref{eq:def-c_k} and \eqref{eq:def-neighbor-gap}.
\end{thm}

\subsubsection{The adiabatic Hamiltonian for quantum linear systems}

All quantum linear systems algorithm inspired by the adiabatic quantum computation framework first construct the adiabatic Hamiltonian $H(s)$, given $A\in \mathbb{C}^{N\times N}$ \cite{subacsi2019quantum, an2022quantum, costa2022optimal}. 
Indeed, in all these methods, the adiabatic Hamiltonian $H(s)$ is defined as
\begin{align}
\label{eq:def-H(s)}
    H(s) = \begin{bmatrix}
        0 & A(f(s)) Q_b \\
        Q_b A(f(s)) & 0 
    \end{bmatrix},
\end{align}
where \begin{align*}
    A(f) :=  \begin{bmatrix}
        (1-f) I & f A \\
        f A^{\dagger} & -(1-f ) I
    \end{bmatrix}
\end{align*}and $Q_b:= I_{2N}-\ket{0,b}\bra{0,b}$ is a projection, and \begin{align*}
    f(s):= \frac{\kappa}{\kappa-1}\left(1-\dfrac{1}{(1+s(\sqrt{\kappa}-1))^{2}}\right)
\end{align*}
is a schedule function with $f(0)=0$ and $f(1)=1$. Then $\ket{0,0,b}$ is a zero-eigenstate of $H(0)$ and $\ket{0,1,x}$ is a zero-eigenstate of $H(1)$, where
$$
\ket{x}=\dfrac{A^{-1}\ket{b}}{\|A^{-1}\ket{b}\|},
$$
is the solution to the quantum linear systems.

By choosing $\sigma_\Pi=\{1\}$ and $\sigma_{\widetilde{\Pi}}$ to be the remaining eigenvalues, the ideal adiabatic evolution with even steps takes the state the $0$-eigenstate of $H(0)$ (i.e., $\ket{0,0,b}$) to the $0$-eigenstate of $H(1)$ (i.e., $\ket{0,1,x}$), which encodes the solution. 

In \cite{costa2022optimal}, the authors assume access to a block encoding of $A$. Using this, they construct a block encoding of $H(s)$, and subsequently, a qubitized quantum walk operator, which corresponds to $B(s)$ in the previous subsection. Furthermore, they prove that these quantum walk operators $B(s)$ satisfy the smoothness condition in Eq.~\eqref{eq:smooth-W} with 
\begin{align}
    c_1(s)=2 T(f(s+1 / T)-f(s))
\end{align}
and \begin{align}
c_2(s)= \begin{cases}
2 \max _{\tau \in\{s, s+1 / T, s+2 / T\}}(2|f^{\prime}(\tau)|^2+|f^{\prime \prime}(\tau)|), & 0 \leq s \leq 1-2 / T, \\ 
2 \max _{\tau \in\{s, s+1 / T\}}(2|f^{\prime}(\tau)|^2+|f^{\prime \prime}(\tau)|), & s=1-1 / T,
\end{cases}
\end{align}
and the gap condition in Eq.~\eqref{eq:def-neighbor-gap} with 
\begin{align}
    \check{\Delta}(s)= \begin{cases}
    (1-f(s+3 / T)+f(s+3 / T) / \kappa), & 0 \leq s \leq 1-3 / T, \\ 
    1 / \kappa, & s=1, 1-2 / T, 1-1 / T.
    \end{cases}
\end{align} 
Finally, by Theorem~\ref{thm:dqc}, \cite{costa2022optimal} shows that the error of the adiabatic evolution can be bounded as 
\begin{align}
    \label{eq:adiabatic-error}
    \|U(s)-U_A(s)\| \leq 44864 \frac{\kappa}{T}+\mathcal{O}\Big(\frac{\sqrt{\kappa}}{T}\Big)
\end{align}
for $T \geq \max (\kappa, 39 \sqrt{\kappa})$.

\subsubsection{Our algorithm}
\label{subsubsec:our-approach}

In this work, we consider $B(s)=e^{\i H(s)}$ instead of the quantum walk operators, which is the first major simplification. Suppose there exists a procedure $U_b$ to prepare the state $\ket{0,0,b}$, then by choosing $T=O(\kappa)$, we will show that applying 
$$
B(0/T), \, B(1/T), \, \ldots, \, B(1-1/T)
$$ 
to the initial state $\ket{0,0,b}$ yields a constant approximation of the state $\ket{0,1,x}$. Thus far, we simply have a sequence of Hamiltonian evolutions with different (time-dependent) Hamiltonians. This is a particular instance of the circuit $W$ in Sec.~\ref{sec:interleaved-sequence}, with $H^{(j)}=H((j-1)/T)$ and $V_j=I$.

Then, we use the framework in Sec.~\ref{sec:qsvt-trotter} to approximate the eigenstate filter polynomial of \cite{lin2020optimal} by a Laurent polynomial in $e^{\i x}$. This projects the aforementioned output state onto the zero-eigenstate $\ket{0,1,x}$ of $H(1)$. Overall, the entire quantum circuit incorporates Hamiltonian evolution operators and Pauli rotations, which is again an interleaved sequence circuit $W$. So, we can use Algorithm~\ref{alg: HSVT with trotter} to estimate $\braket{x|O|x}$. Before proving our result formally, let us assume that we have an $n$-qubit operator $A$ that can be expressed as a linear combination of strings of Pauli operators. That is,
\begin{equation}
\label{eq:qls-matrix}
    A=\sum_{j=1}^{L}\lambda_j P_j,
\end{equation}
where $P_j\in \{I, X, Y, Z\}^{\otimes n}$, and $\lambda=\sum_{j}|\lambda_j|$. Alternatively, one could also consider any $A$ that is a sum of local terms that are easy to simulate. Indeed, any $2^n$ dimensional operator can be expressed in this form. The first advantage of this is that each term $P_j$ is now efficiently simulable allowing us to use near-term Hamiltonian simulation techniques for solving this problem. The other advantage is that (as discussed previously) $A$ naturally has this form in physically relevant settings (for instance, Hamiltonians for a wide range of physical systems). This is particularly relevant in scientific computation, where problems involving the inversion of operators of this form, are ubiquitous. This includes Green's function estimation \cite{tong2021fast, wang2024qubit}, solving linear and partial differential equations \cite{berry2017quantum, Childs2021highprecision, linden2022quantum, Krovi2023improvedquantum}, and processing eigenvalues of non-Hermitian systems \cite{low2024qep}.

It can be shown that if \( A \) can be expressed as a linear combination of Pauli operators, then \( H(s) \) can be decomposed into a sum of Hermitian terms, each of which generates a Hamiltonian evolution that can be implemented using \( \widetilde{O}(1) \) time and queries. To handle complex coefficients $\lambda_j$, let us define the operator $Q_j \coloneqq (\Re(\lambda_j) X - \Im(\lambda_j) Y)/|\lambda_j|$, which is both Hermitian and unitary. Indeed, for any $f\in [0,1]$, $A(f)$ can be decomposed as 
$$
    A(f) = \sum_{j=1}^L \Big(|\lambda_j|f(s) Q_j \otimes P_j + \lambda_j(1-f(s)) Z \otimes I \Big).
$$
Then for any $s\in [0,1]$, $H(s)$ admits the decomposition 
    \bea
    H(s) &=& \sum_{j=1}^L \Big(|\lambda_j|f(s) X \otimes Q_j \otimes P_j + \lambda_j(1-f(s)) X\otimes Z \otimes I \Big) \nonumber \\
    && - \sum_{j=1}^L |\lambda_j|f(s)
    \begin{bmatrix}
        0 & (Q_j \otimes P_j) \ket{0,b} \bra{0,b} \\
        \ket{0,b} \bra{0,b} (Q_j \otimes P_j) & 0 
    \end{bmatrix} \label{eq:sec-H(s)} \\
    && - \sum_{j=1}^L \lambda_j (1-f(s))
    \begin{bmatrix}
        0 & (Z \otimes I) \ket{0,b} \bra{0,b} \\
        \ket{0,b} \bra{0,b} (Z \otimes I) & 0 
    \end{bmatrix}. \label{eq:thd-H(s)}
    \eea

%\textcolor{purple}{Above expression is true only when $A^{\dagger}=A$, else considering $A^{\dagger}=\sum_j\lambda_j^*P_j$ won't let us write in the form $X\otimes P_j$.} \cs{Can we assume $A$ is Hermitian in the beginning? What is the problem with this assumption? 1 more ancilla?}  
We shall demonstrate that $e^{\i H(s)}$ can be efficiently implemented using (higher-order) Trotterization. Formally, we have the following theorem, using which we prove both the correctness and the complexity of our algorithm:
\begin{thm}
\label{thm:qls-our-method}
    Let $Ax=b$ be a system of linear equations, where $A$ is an $N$ by $N$ matrix with $\|A\|=1$ and $\|A^{-1}\|=\kappa$, and let $\varepsilon\in(0,1/2)$ be the precision parameter. Given a unitary oracle $U_b$ that prepares the state $\ket{b}$ and a Pauli decomposition of $A=\sum_{j=1}^L \lambda_j P_j$, Algorithm \ref{alg: HSVT with trotter}, for any observable $O$, estimates 
    $$
   \braket{x|O|x},
    $$ 
    to within an additive error of $\varepsilon\|O\|$, using only $\log(N)+4$ qubits. Let $\lambda = \sum_{j=1}^L |\lambda_k|$. Then the maximum quantum circuit depth is $\widetilde{O}(L(\lambda \kappa)^{1+o(1)})$, while the total time complexity and the number of queries to $U_b$ and $U_b^{\dag}$ are $\widetilde{O}(L(\lambda \kappa)^{1+o(1)}/\varepsilon^2)$. 
    
    If we further assume $O/\gamma$ is block-encoded via a unitary $U_O$ for some $\gamma$ satisfying $\gamma = \Theta(\|O\|)$, the maximum quantum circuit depth, total time complexity, and the number of queries to $U_b$ and $U_b^{\dag}$ are $\widetilde{O}(L(\lambda \kappa)^{1+o(1)}/\varepsilon)$. The number of queries to $U_O$ is $\widetilde{O}(1/\varepsilon)$. Aside from the ancillas used in the block encoding of $O$, we need one additional ancilla qubit.
\end{thm}

\begin{proof} 
We first prove the correctness and then estimate the complexity of our algorithm.
~\\~\\
\textbf{Correctness:~}As mentioned previously, we simply run Algorithm \ref{alg: HSVT with trotter}, as the entire algorithm is a particular instance of the interleaved sequence $W$. We can divide the analysis into two parts. For the first part, we consider $B(s)=e^{\i H(s)}$, and implement the unitary
$$
U=\prod_{j=0}^{T-1} B(j/T).
$$
The $0$-eigenstates of $H(s)$ are now translated to the $1$-eigenstates of $B(s)$. So $\ket{0,0,b}$ is an $1$-eigenstate of $B(0)$ and $\ket{0,1,x}$ is an $1$-eigenstate of $B(1)$. Although there is another zero eigenstate $\ket{1,0,b}$ of $H(0)$ and $H(1)$, it is orthogonal to the eigenstates of interest $\ket{0,0,b}$ and $\ket{0,1,x}$, and the ideal adiabatic evolution operator $U_A(s)$ still maps $\ket{0,0,b}$ to $\ket{0,1,x}$. 

Now, we bound the adiabatic error by showing that the smooth parameters $c_1(s),c_2(s)$ and the gap parameter $\check{\Delta}(s)$ of our Hamiltonian evolution operators $W(s)$ have the same asymptotic scaling as the quantum walk operators considered in \cite{costa2022optimal}. The quantum walk operators have eigenvalues scaling as $e^{\pm\i \arcsin(H(s))}$ while in our case, the Hamiltonian evolution operators are $e^{\i H(s)}$. The eigenvalue gap of $\arcsin(H)$ and $H$ around zero are asymptotically the same, so the gap parameter $\check{\Delta}(s)$ is asymptotically the same for the $B(s)$. By applying the variation-of-parameters formula, the derivative of the Hamiltonian evolution operator $B(s)=e^{\i H(s)}$ is
    $$
    \int_{0}^1 e^{\i u H(s)} \frac{\mathrm{d} H(s)}{\mathrm{d} s} e^{
        \i(1-u) H(s)} \, \mathrm{d}u
    $$ whose norm is bounded by \begin{align*}
        \left\|\frac{\mathrm{d}H(s)}{\mathrm{d}s}\right\|=\left\| \frac{\mathrm{d}H(f)}{\mathrm{d}f}\right\|\left|f'(s)\right|\le \left\| \frac{\mathrm{d}A(f)}{\mathrm{d}f}\right\|\left|f'(s)\right|\le \sqrt{2}\left|f'(s)\right|,
    \end{align*}
    where the second inequality follows from 
    $$\Big\|\Big(\frac{\mathrm{d}A(f)}{\mathrm{d}f}\Big)^{\dagger}\frac{\mathrm{d}A(f)}{\mathrm{d}f}\Big\|=1+\|A\|^2=2.$$ The second derivative of $H(s)$ can be bounded by $\sqrt{2}|f''(s)|$ in the same way. 
    
We can bound the difference in $B(s)$ by 
\begin{align*}
        \|B(s+1/T)-B(s)\|&\le \int_{s}^{s+1/T} \left\|\frac{\mathrm{d}e^{\i H(\tau)}}{\mathrm{d}\tau}\right\| \, \mathrm{d}\tau\\
        &\le \sqrt{2}\int_{s}^{s+1/T} \left|f'(\tau)\right| \, \mathrm{d}\tau\\
        &=\sqrt{2}\left(f(s+1/T)-f(s)\right),
    \end{align*}
    where the equality follows from the fact that $f(s)$ is a monotonically increasing function. Therefore, we can take $$c_1(s)=\sqrt{2}T(f(s+1/T)-f(s))$$
    to satisfy the smoothness condition. The norm of the second derivative of $W(s)$ is bounded by 
    \begin{align*}
         &\int_{0}^1 \left\|\frac{\mathrm{d}e^{\i u H(s)}}{\mathrm{d}s} \frac{\mathrm{d}H(s)}{\mathrm{d}s}e^{\i(1-u) H(s)}\right\|+ \left\|e^{\i u H(s)} \frac{\mathrm{d^2}H(s)}{\mathrm{d}s^2} e^{\i(1-u) H(s)}\right\|+ \left\|e^{\i u H(s)} \frac{\mathrm{d}H(s)}{\mathrm{d}s}\frac{\mathrm{d}e^{\i (1-u)H(s)}}{\mathrm{d}s}\right\| \, \mathrm{d}u\\
        \le\, &\int_{0}^1( 2u |f'(s)|^2+\sqrt{2}|f''(s)|+2(1-u) |f'(s)|^2) \, \mathrm{d}u\\
        =\, &2|f'(s)|^2+\sqrt{2}|f''(s)|.
    \end{align*}
To bound the second difference of $B(s)$, we use Taylor's theorem to write 
    \begin{align*}
        &\left\|B(s+2/T)-2B(s+1/T)+B(s)\right\|\\
        =\, & \Big\| \int_{s+1 / T}^{s+2 / T}(s+2 / T-\tau) B^{\prime \prime}(\tau) d \tau +\int_s^{s+1 / T}(\tau-s) B^{\prime \prime}(\tau) d \tau \Big\| \\
        \le\, &\frac{1}{T^2} \max_{\tau \in [s, s+2 / T]} \left\|B^{\prime \prime}(\tau)\right\| \\
        \le\, & \frac{1}{T^2} \max_{\tau \in \{ s, s+1 / T, s+2 / T\}}(2|f'(\tau)|^2+\sqrt{2}|f''(\tau)|),
    \end{align*}
    where the last inequality follows from the fact that $|f'(\tau)|$ and $|f''(\tau)|$ are monotonically decreasing functions. So, we can set $$
    c_2(s)=\max_{\tau \in \{s, s+1 / T, s+2 / T\}}\left(2\left|f'(\tau)\right|^2+\sqrt{2}\left|f''(\tau)\right|\right),
    $$ to satisfy the smoothness condition. In conclusion, we choose asymptotically the same smooth parameters $c_1(s),c_2(s)$ and the gap parameter $\check{\Delta}(s)$ as the quantum walk operators in \cite{costa2022optimal}. Therefore, by Theorem~\ref{thm:dqc}, the error in the adiabatic evolution is bounded as
    \begin{align*}
        \|U(s)-U_A(s)\| = O(\kappa/T)
    \end{align*}
    in the same way as Eq.~\eqref{eq:adiabatic-error}. 
    We can choose $T=O(\kappa)$ to ensure that $U(1)\ket{0,0,b}$  has constant overlap with the target state $U_A(1)\ket{0,0,b}=\ket{0,1,x}$. This completes the first part.

    For the second part, we project the output state onto the zero eigenstate $\ket{0,1,x}$ of $H(1)$. For this,  we approximately implement the filter function of $H(1)$, by approximating it with a Laurent polynomial in $e^{\i x}$, allowing us to use Theorem \ref{thm_gqsp-trotter-suzuki}. Let $P(z)$ be the Laurent polynomials in Lemma~\ref{lem:tri-approx-filter} with parameters $(\Delta, \varepsilon)=(1/\kappa, \varepsilon)$. Then we have $P(1)\in [1-\varepsilon,1]$ and $P(e^{\i x})\in [-\varepsilon/2, \varepsilon/2]$ for all $x\in [-1,-1/\kappa]\cup[1/\kappa, 1]$, and $\deg(P) = O(\kappa\log(1/\varepsilon))$. Since all non-zero eingenvalues of $H(1)$ are in $[-1,-1/\kappa]\cup[1/\kappa, 1]$, $P(e^{\i H(1)})$ is an $\varepsilon$-approximation of the projection operator onto the zero eigenstates of $H(1)$. Therefore, we can apply Theorem \ref{thm_gqsp-trotter-suzuki} with the Laurent polynomial $P(z)$ to project the output state of the adiabatic evolution onto the zero eigenstate $\ket{0,1,x}$ of $H(1)$ with constant success probability. 
     Both the parts overall consist of two circuits, which is a particular case of the interleaved sequence circuit $W$ in Sec.~\ref{sec:interleaved-sequence}. When combined, they form a new circuit, which is still in the form of $W$. After running this circuit, we can measure the state conditioned on the success of the projection to estimate the expectation value of the observable $O$ with respect to $\ket{x}$. Thus, overall, the quantum circuit only involves Hamiltonian evolution operators and Pauli rotations, and, we can use Algorithm~\ref{alg: HSVT with trotter} to estimate $\braket{x|O|x}$ to within error $\varepsilon\|O\|$. 

Now, we move on to analyzing the complexity of our algorithm.
~\\~\\
\textbf{Complexity:~}The entire interleaved Hamiltonian evolution sequence consists of two parts. The first part is an adiabatic evolution composed of \( T = O(\kappa) \) Hamiltonian evolutions:  
$$
B(0)=e^{\i H(0)}, \,
B(1/T)=e^{\i H(1/T)}, \,
\ldots, \,
B(1-1/T)=e^{\i H(1-1/T)}.
$$ 
The second part is an eigenstate filtering stage, which involves \( 2\deg(P) = O(\kappa \log(1/\varepsilon)) \) interleaved single-qubit Pauli rotations and Hamiltonian evolutions of the form  
$$
e^{\i \ket{0}\bra{0} \otimes H(1)} \qquad\text{and}\qquad e^{-\i \ket{1}\bra{1} \otimes H(1)}.
$$ 
We can indeed implement the Hamiltonian evolution using Trotter methods, as $H(s)$ can be written as a sum of simple terms as shown previously.  

Note that for any two states $\ket{y},~\ket{z}$ we have 
$$
\begin{bmatrix}
    0 & \ket{z} \bra{y} \\ \ket{y} \bra{z} & 0
    \end{bmatrix}
    =\begin{bmatrix}
        \ket{z} \\ \ket{y}
    \end{bmatrix} \begin{bmatrix}
        \bra{z} & \bra{y}
    \end{bmatrix}
    - \begin{bmatrix}
        \ket{z} \bra{z} & 0 \\ 0 & \ket{y} \bra{y} 
    \end{bmatrix}.
$$ 
So, the Hermitian terms in Eq.~\eqref{eq:sec-H(s)} and Eq.~\eqref{eq:thd-H(s)} can be decomposed in this manner with $\ket{z} = (Q_j\otimes P_j)\ket{0,b}, (Z\otimes I)\ket{0,b}$ and $\ket{y}=\ket{0,b}$.

The Hamiltonian evolution 
\begin{align*}
    &\exp\left(\i t
    \begin{bmatrix}
    (Q_j\otimes P_j)\ket{0,b} \\
    \ket{0,b}
    \end{bmatrix}
    \begin{bmatrix}
    \bra{0,b} (Q_j\otimes P_j) & \bra{0,b}
    \end{bmatrix}
    \right) \\
    =\ &\begin{bmatrix}
        (Q_j\otimes (P_jU_b)) &  0\\
        0 & I\otimes U_b
    \end{bmatrix}\exp\left(2\i t \ket{+} 
    \ket{0^{n+1}}
    \bra{+}\bra{0^{n+1}}
    \right)\begin{bmatrix}
        (Q_j\otimes (U_b^{\dag} P_j)) &  0\\
        0 & I\otimes U_b^{\dag}
    \end{bmatrix}
    \end{align*}
    can be implemented using one query to $U_b$ and $U_b^{\dag}$, along with one $(n+3)$-qubit Toffoli gate and a constant number of elementary gates. The Hamiltonian evolution \begin{align*}
        &\exp\left(\i t \begin{bmatrix}
        (Q_j\otimes P_j)\ket{0,b}\bra{0,b}(Q_j\otimes P_j)  & 0 \\ 0 & \ket{0,b}\bra{0,b}
    \end{bmatrix}\right) \\
    =\ & \begin{bmatrix}
        (Q_j\otimes (P_jU_b)) &  0\\
        0 & I\otimes U_b
    \end{bmatrix}\exp\left(\i t (I\otimes \ket{0^{n+1}}\bra{0^{n+1}})
    \right)\begin{bmatrix}
        (Q_j\otimes (U_b^{\dag} P_j)) &  0\\
        0 & I\otimes U_b^{\dag}
    \end{bmatrix}
    \end{align*} can be implemented using one query to $U_b$ and $U_b^{\dag}$, along with one $(n+2)$-qubit Toffoli gate and a constant number of elementary gates. Using the $2k$-th order Trotter-Suzuki formula $S_{2k}$ to approximately implement the Hamiltonian evolution, by Theorem~\ref{thm:qsp-with-trotter-and-interpolation}, the maximum quantum circuit depth is \begin{align*}
        \widetilde{O}\left((25/3)^k k^2 L \big(\kappa \lambda_{\rm comm}\big)^{1+1/2k} \right),
    \end{align*}
    while the total time complexity and number of queries to $U_b$ and $U_b^{\dagger}$ are 
    \begin{align*}
        \widetilde{O}\left((25/3)^k k^2 L \big(\kappa \lambda_{\rm comm}\big)^{1+1/2k}/\varepsilon^2 \right).
    \end{align*} 
    By Eq.~\eqref{eq:para comm lambda}, we have $\lambda_{\rm comm}=O(\lambda)$. 
    Choosing a large enough value of $k$ (e.g., $k = \log(\log(\kappa/\varepsilon))$) gives the claimed results. We need $n+2$ qubits to encode $H(s)$, one qubit for implementing the filter function using GQSP, and one qubit to implement the Toffoli gates, so the total number of qubits is $n+4$. 
    
    On the other hand, if $O$ is measured coherently, simplifying the complexity in Theorem~\ref{thm:qsp-with-trotter-and-interpolation-coherent} gives the claimed results. Since the ancilla qubit for Toffoli gates can be reused, we only need one more ancilla qubit. 
\end{proof}

We compare the complexity of our algorithm with the state-of-the-art procedure of \cite{costa2022optimal} in Table \ref{table1:comparison-qls}. In \cite{costa2022optimal}, the block encoding of $A$ alone uses $\lceil\log_2 L\rceil$ ancilla qubits. In addition to this, their algorithm needs six more ancilla qubits: (i) four ancilla qubits for constructing the quantum walk operator out of the block encoding of $A$, and (ii) two ancilla qubits for an LCU-based method to implement the eigenstate filter. In order to estimate $\braket{x|O|x}$ to within an $\eps$-additive accuracy, the maximum quantum circuit depth is $\widetilde{O}(L\lambda\kappa)$, and $O(\|O\|^2/\eps^2)$ classical repetitions, if $O$ is measured incoherently, while using $\lceil\log_2 L\rceil+6$ ancilla qubits.

Our method has a quasi-linear dependence on $\kappa$ and a polylogarithmic dependence on $1/\eps$, without assuming any block encoding access, or using quantum walks. Moreover, we only need four ancilla qubits in all. To the best of our knowledge, this is the first quantum algorithm for quantum linear systems that satisfies the following three properties: (i) does not assume block encoding access of $A$, (ii) achieves a quasi-linear dependence on $\kappa$,  and (iii) uses only a constant number of ancilla qubits. Note that randomized quantum linear systems algorithms have recently been developed, tailored to near-term quantum devices. These methods require a single ancilla qubit but has a worse total time complexity of $\widetilde{O}(\lambda^2\kappa^6/\eps^2)$.

In many applications, such as Green's function estimation, linear regression, and solving differential equations, quantum linear systems are used as a subroutine. That is, the state $\ket{x}$ is the input to a different part of a larger quantum circuit. Our algorithm can also be used as a subroutine, as it basically implements an interleaved sequence circuit of unitaries and Hamiltonian evolution operators. Note that the final unitary in Eq.~\eqref{eq:interleaved operator} can be a new quantum algorithm altogether. Thus, the measurement of the observable (and consequently Richardson extrapolation) can be delayed until the end of the full circuit. 
%Finally, instead of higher-order Trotterization, we could also use qDRIFT and Algorithm \ref{alg: HSVT with qDRIFT} described in the next section to estimate $\braket{x|O|x}$, leading to a complexity of $\widetilde{O}(\lambda^2\kappa^2)$.

\subsection{Ground state property estimation}
\label{subsec:app-ground-state}

Consider a Hamiltonian $H$ with spectral decomposition $H=\sum_{i} \xi_i \ket{v_i}\bra{v_i}$, where the eigenvalues $\xi_i$ are arranged in increasing order. The {\em ground state} of $H$ is the eigenstate $\ket{v_0}$ corresponding to the smallest eigenvalue $\xi_0$, known as the {\em ground state energy}. Preparing the ground state and estimating the ground state energy of a Hamiltonian are fundamental tasks in quantum computing and quantum many-body physics. In general, these problems are QMA-hard \cite{kempe2006complexity}, but they become tractable when additional information is available. Specifically, one assumes: (i) knowledge of the spectral gap $\Delta$ of the Hamiltonian, and (ii) access to a procedure that prepares an initial guess state $\ket{\phi_0}$ that has at least $\gamma$ overlap with the ground state $\ket{v_0}$, i.e.\ $|\braket{\phi_0|v_0}|\geq \gamma$. 

Furthermore, we assume that any Hamiltonian $H$ of interest can be expressed as a sum of $L$ local terms, namely, \ $H=\sum_{k=1}^L H_k$, where each $H_k$ is a local Hamiltonian that can be exponentiated in constant time. We now formally state the ground state property estimation problem:

\begin{prob}[Ground state property estimation]
\label{prob:ground state}
    Let $H=\sum_{k=1}^L H_k$ be a Hamiltonian whose spectrum is contained in $[-1, 1]$, and $\lambda=\sum_{k}\|H_k\|$. Suppose $\varepsilon\in(0,1/2)$ and that we can efficiently prepare an initial guess state $\left|\phi_0\right\rangle$ such that it has an overlap of at least $\gamma$ with the ground state $\left|v_0\right\rangle$ of $H$, i.e.\ 
    $\left|\braket{\phi_0|v_0} \right| \geq \gamma$, for $\gamma\in (0,1)$. Furthermore, assume that there is a spectral gap $\Delta$ separating the ground state energy $\xi_0$ from the rest of the spectrum. Then, the goal is to compute $\bracket{v_0}{O}{v_0}$ up to additive error $\varepsilon\|O\|$ for any given observable $O$. 
\end{prob}
Before discussing our algorithms, we briefly discuss some prior results for this problem.

\subsubsection{Prior work}
The state-of-the-art algorithms for this problem utilize LCU or QSVT, but assume block encoding access to the underlying Hamiltonian \cite{ge2019faster, lin2020nearoptimalground}. However, these algorithms are only optimal in terms of the number of queries made to the underlying block-encoded operator. 
For instance, the algorithm in \cite{lin2020nearoptimalground} requires a circuit depth of $\widetilde{O}(L\lambda\Delta^{-1}\gamma^{-1}\eps^{-1})$
to estimate $\braket{v_0|O|v_0}$, while using $O(\log L)$ ancilla qubits and several sophisticated controlled operations. The ancilla requirement of \cite{ge2019faster} is even higher. 

Constructing a block encoding of $H$ can be prohibitively costly in many settings. As a result, several recent results on ground state property estimation have explored alternatives that avoid this overhead. These approaches, however, assume access to a different oracle: the time-evolution operator $U=e^{\ iH}$, and measure cost in terms of the number of queries to $U$ \cite{lin2022heisenberg, dong2022ground, wan2022randomized,Wang2023quantumalgorithm, zhang2022computingground, lin2022heisenberg}. This framework is commonly referred to as the Hamiltonian-evolution access model, and most of these algorithms require only $O(1)$ ancilla qubits. To convert such oracle-based procedures into fully end-to-end quantum algorithms without incurring additional ancilla overhead, one can approximate $U$ by higher order Trotterization. using higher-order Trotterization. However, this substitution generally worsens the precision dependence, from the optimal Heisenberg scaling of  $1/\eps$ to a suboptimal $1/\eps^{1+o(1)}$.

Our approach builds on the algorithm proposed in Ref.\cite{dong2022ground}, which introduces quantum eigenvalue transformation using unitaries (QET-U)—a framework for implementing real polynomial transformations of unitary matrices, assuming oracle access to $e^{iH}$. Notably, the QET-U circuit structure can be viewed as an interleaved sequence of single-qubit $Z$ rotations and Hamiltonian evolutions, which corresponds to a specific instance of the more general circuit architecture $W$ (introduced in Sec.~\ref{sec:interleaved-sequence}). This connection suggests that end-to-end quantum algorithms for ground state property estimation and related problems can be naturally realized within our framework, where classical extrapolation techniques allow us to recover the optimal Heisenberg scaling.
% Here, we demonstrate that the results established in Sec.~\ref{sec:qsvt-trotter} can be used to solve ground state property estimation effectively. 
% In particular, Theorem \ref{thm:gqsp-trotter-coherent measurement} enables us to achieve near-optimal complexity without assuming any oracular access, while still using only $O(1)$ ancilla qubits. The ground state energy estimation algorithm from Dong et al. \cite{dong2022ground} only requires estimating an observable's expectation value to a constant accuracy. Extrapolation techniques offer no advantage in this case, so we only focus on the property estimation problem. 
% We now proceed to discuss our methods in detail.

\subsubsection{Our algorithm}
\label{subsec:gspe-algorithm}
The ground state property estimation algorithm of \cite{dong2022ground} implements a polynomial approximation to the shifted sign function using the QET-U framework. In particular, they assume knowledge of some $\mu>0$ such that 
 \[
    \xi_0 \leq \mu-\Delta / 2<\mu+\Delta / 2 \leq \xi_1.
\]
Then, the shifted sign function is defined as
 \be
    \label{eq_shifted sign}
    \theta(x)= \begin{cases}1, & x \leq \mu, \\ 0, & x>\mu. \end{cases}
    \ee
Applying a polynomial approximation to the shifted sign function effectively filters out excited-state components, projecting the initial state $\ket{\phi_0}$ close to the ground state $\ket{v_0}$. 

For our algorithm, we show that the shifted sign function can be approximated by a Laurent polynomial in $e^{\i x}$ that remains bounded in $\mathbb{T}$ (see Lemma \ref{lem:approx-shift-sign}). This key observation enables the application of Theorem~\ref{thm:gqsp-trotter-suzuki-aa}, allowing us to estimate the expectation value $\braket{v_0|O|v_0}$. This leads to the following result:

\begin{thm}[Ground state property estimation using GQSP with Trotter]
\label{thm_ground state property estimation using HSVT with Trotter}
    Assume that $\|H\|\le 1$ and there is a spectral gap $\Delta$ separating the ground state energy $\xi_0$ from the first excited state energy $\xi_1$ such that
    \[
    \xi_0 \leq \mu-\Delta / 2<\mu+\Delta / 2 \leq \xi_1,
    \]
    for some given $\mu$. Let $k$ be a positive integer. Then there is a quantum algorithm for the ground state property estimation problem \ref{prob:ground state} using only two ancilla qubits. The maximum circuit depth is
    \[
    \widetilde{O} \Big({L} \Big(\frac{\lambda_{\rm comm}}{ \Delta\gamma}\Big)^{1+o(1)} \Big),
    \]
    and the time complexity is \[
    \widetilde{O} \left(\frac{L}{\varepsilon^2} \left(\frac{\lambda_{\rm comm}}{ \Delta\gamma}\right)^{1+o(1)} \right),
    \]
    where $\lambda_{\mathrm{comm}}$ defined in Eq.~\eqref{Lambda} for all $H^{(\ell)}=H$ scales with the nested commutators of $H$. If we further assume $O/\eta$ is block-encoded via a unitary $U_O$ for some $\eta$ satisfying $\eta = {\Theta}(\|O\|)$, we can reduce the time complexity to \[
        \widetilde{O} \left(\frac{L}{\varepsilon} \left(\frac{\lambda_{\rm comm}}{ \Delta\gamma}\right)^{1+o(1)} \right),
    \] using one additional ancilla qubit. %where $\lambda = \sum_{j=1}^L |\lambda_j|$.}
    %where $\alpha = \sum_{j=1}^L |\alpha_j|$. 
\end{thm}

\begin{proof}
    According to Lemma \ref{lem:approx-shift-sign}, there is a Laurent polynomial $P(z)$ of degree $O(\Delta^{-1} \log((\gamma\varepsilon)^{-1}))$ such that $\|P(z)\|_{\mb T} \leq 1$, and\footnote{Recall that $\mathbb{T}=\{ x\in \mathbb{C} : |x|=1 \}$.}
    \be
    \label{eq_ground state property}
    \begin{cases}
        |P(e^{\i x})-1| \leq \gamma\varepsilon/12, & \quad x\in\left[-1,\mu-\Delta / 2 \right], \\
        |P(e^{\i x})| \leq \gamma \varepsilon/12, & \quad x\in\left[\mu+\Delta / 2,1\right].
    \end{cases}
    \ee
    Therefore, we obtain a bounded polynomial $P$ such that $|P(e^{\i x})-\theta(x)| \leq \gamma\varepsilon$ for all $[-1,\mu-\Delta / 2] \cup [\mu+\Delta / 2, 1]$. Then we have \begin{align*}
        \|P(e^{\i H})\ket{\phi_0}-\braket{\psi_0|\phi_0}\ket{\psi_0}\| &= \| P(e^{\i H})\ket{\phi_0}-\theta(H) \ket{\phi_0} \| \le \gamma \eps,
    \end{align*}
    where the first equation follows from $\theta(H) \ket{\phi_0} =\braket{\psi_0|\phi_0}\ket{\psi_0}$ and the first inequality follows from that $H$ has no eigenvalues in $(\mu-\Delta/2, \mu+\Delta/2)$. Let $\ket{\psi} = P(e^{\i H})\ket{\phi_0}/\|P(e^{\i H})\ket{\phi_0}\|$. The distance between the two states after normalization is bounded as \[
    \bigg\|\ket{\psi}-\frac{\braket{\psi_0|\phi_0}\ket{\psi_0}}{|\braket{\psi_0|\phi_0}|}\bigg\| \le \frac{2\gamma\varepsilon}{12|\braket{\psi_0|\phi_0}|} \le \frac{\varepsilon}{6},
    \]
    and hence \[
    \|\bra{\psi}O\ket{\psi}-\bra{\psi_0}O\ket{\psi_0}\|\le \|O\|\eps/2
    \]
    By Theorem~\ref{thm:gqsp-trotter-suzuki-aa}, we can estimate $\bra{\psi}O\ket{\psi}$ to accuracy $\varepsilon\|O\|/2$ with maximum circuit depth
    \bes
    \widetilde{O}\big((25/3)^k k^2 L \big( \lambda_{\rm comm}\Delta^{-1}\gamma^{-1}\big)^{1+1/2k} \big),
    \ees
    and total time complexity 
    \bes
    \widetilde{O}\big((25/3)^k k^2 L \big(\lambda_{\rm comm}\Delta^{-1}\gamma^{-1}\big)^{1+1/2k}/\eps^2 \big),
    \ees
    which gives an estimate of $\bra{\psi_0}O\ket{\psi_0}$ to within an additive accuracy of $\|O\|\eps$. 

    Following Theorem~\ref{thm:qsp-with-trotter-and-interpolation-coherent}, it is also possible to measure $O$ using IQAE to achieve optimal complexities. This would increase the quantum circuit depth by a factor of $1/\eps$, but reduce the total time complexity to $\widetilde{O}(1/\eps)$, and we need one more ancilla qubit for the Hadamard test circuit to encode the expectation value into an amplitude. For both coherent and incoherent estimation, choosing a large enough $k$ results in the stated complexities.    
\end{proof}

We have summarized the complexity of our algorithm alongside prior work in Table~\ref{table:ground state problem}. For a sufficiently large choice of Trotter order $k$, our algorithm achieves nearly the same time complexity as the state-of-the-art method of \cite{lin2020nearoptimalground}, despite their stronger assumption of quantum access to a block encoding of $H$, which can be expensive to construct. For example, consider Hamiltonians expressed as a linear combination of unitaries: $H=\sum_k \lambda_k H_k$, where $\|H_k\|=1$, and each unitary $H_k$ is easy to implement. Constructing a block encoding of such an $H$ using the LCU method requires $\lceil \log_2L\rceil$ ancillas, along with sophisticated multi-qubit controlled operations. Additionally, estimating expectation values via iterative quantum amplitude estimation demands at least four more ancilla qubits, further increasing the hardware overhead.

We now compare the complexity of Dong et al.\cite[Theorem 11]{dong2022ground} with that of Theorem \ref{thm_ground state property estimation using HSVT with Trotter}. To ensure a fair comparison, we assume that the oracle $U=e^{iH}$ in \cite{dong2022ground} is implemented via the $2k$-th order Trotter-Suzuki product formula $S_{2k}(t)$. Furthermore, as with our algorithm the expectation value $\braket{v_0|O|v_0}$ can be estimated either incoherently or coherently. This leads to different quantum circuit depths and time complexities, as shown below:
\begin{itemize}
    \item For incoherent estimation, the ground state property estimation algorithm of \cite{dong2022ground}, each coherent run requires a quantum circuit depth of 
    $$
    \widetilde{O}\left(L \big(\lambda_{\rm comm}' \Delta^{-1}\gamma^{-1}\big)^{1+o(1)}/\eps^{o(1)}\right),
    $$
    while $\widetilde{O}(\eps^{-2})$ classical repetitions are needed, overall. Here, $\lambda'_{\rm comm}=O(\sum_{k}\|H_k\|)$ scales with the nested commutators of $H_k$, but is different from the expression of $\lambda_{\rm comm}$ in Remark \ref{remark: A bound of Lambda}. Note that the circuit depth per coherent run is exponentially slower than our algorithm.

    \item For coherent estimation, iterative quantum amplitude estimation results in a total time complexity of 
    $$
    \widetilde{O}\left(L \big(\lambda_{\rm comm}' \Delta^{-1}\gamma^{-1}\big)^{1+o(1)}/\eps^{1+o(1)}\right),
    $$
    which is polynomially worse than our algorithm, which manages to attain the Heisenberg scaling of $1/\eps$.
\end{itemize}

\section{Discussion and open problems}
\label{sec:discussion}
In this work, we have introduced a new approach to Quantum Singular Value Transformation (QSVT) that eliminates block encodings while maintaining near-optimal complexity. By utilizing direct Hamiltonian evolution with higher-order Trotterization and Richardson extrapolation, we have significantly reduced the ancilla overhead, using only a single ancilla qubit instead of the standard $O(\log L)$. Our results provide a versatile and general framework for implementing polynomial transformations of Hermitian and non-Hermitian matrices, making QSVT more hardware-efficient and accessible for near-term quantum devices.

Moreover, our framework extends beyond QSVT, offering a general technique for mitigating errors in any quantum circuit that is an interleaved sequence of arbitrary unitary operations and Hamiltonian evolution. We applied this to develop near-optimal quantum algorithms for quantum linear systems and ground-state property estimation without assuming any block encoding access. Our work opens avenues for future research:

\begin{itemize}
\item {\bf Optimal cost of the interleaved sequence circuit:} For the interleaved sequence of arbitrary unitaries and Hamiltonian evolution operators considered in Sec.~\ref{sec:interleaved-sequence}, an interesting direction of future research would be to develop a method that estimates the desired expectation value with a cost that scales strictly linearly in $M$ (instead of quasi-linear, as in our case) while satisfying the three desirable features of our method: (i) commutator scaling, (ii) no ancilla overhead, (iii) not needing any oracular access to $H$. Such a method would attain optimal circuit depth for QSVT (as $M$ translates to the polynomial degree in the case of QSVT, which is a proven lower bound \cite{gilyen2019quantum}) without using block encodings or more ancilla qubits. 
%A closely related research problem is: Given a unitary decomposition $H=\sum_j \lambda_j U_j$, can we efficiently construct a block-encoding of $H$ using only 1 or $O(1)$ ancilla qubits? Solving this would allow us to directly apply the standard QSVT to design end-to-end quantum algorithms, achieving optimal complexity. However, we suspect this is not feasible. So can we prove a lower bound showing that any such block-encoding must require a large number of ancilla qubits, say $\Omega(\log L)$?

\item {\bf Generalization to CPTP maps:} Another possible direction of research would be to generalize the interleaved sequence circuit in Sec.~\ref{sec:interleaved-sequence} to incorporate completely positive trace preserving (CPTP) maps. For instance, an interleaved sequence of CPTP maps and Hamiltonian evolutions capture open quantum systems dynamics. Indeed, quantum collision models or repeated interaction maps provide a rich framework to simulate complex open quantum system dynamics \cite{cattaneo2021collision, ciccarello2022collisionreview}. Here, the environment is a sum of discrete sub-environments. Each sub-environment interacts (or \emph{collides}) with the system, one by one, for a certain time, before being traced out. The sequence of interactions can be simulated on a quantum computer, with the underlying circuit being an interleaved sequence of tracing out of the corresponding sub-environment (a CPTP map) and the evolution of the total (system+environment+interaction) Hamiltonian for some time. Thus, a generalization of our interleaved sequence circuit would allow for simulating a wide variety of open systems dynamics, such as Lindblad evolution \cite{ding2024simulating, pocrnic2023quantum} and beyond \cite{ciccarello2013collision}, with near-optimal complexity, using near-term Hamiltonian simulation procedures and very low hardware overhead. Concretely, is it possible to establish a bound on the circuit depth of interleaved sequences with CPTP maps and Hamiltonian evolution operators that match known bounds for Lindblad evolution \cite{cleve2017efficient,ding2024simulating, pocrnic2023quantum}?

% \item {\bf Lower bounds for other functions:} We established a lower bound of $\Omega(t^2)$ for 
% $f(x)=e^{\i x t}$ for any randomized quantum algorithms of computing $f(H)$ when a sampling access to the Pauli decomposition of $H$ is given. 
% This shows that any generic algorithm implementing polynomial transformations within this access model requires a quadratic dependence on the polynomial degree. However, a key open question is whether better lower bounds can be obtained for other functions. Alternatively, is it possible to rigorously prove that for any continuous function $f(x)$, the lower bound is $\Omega(\widetilde{\deg}(f)^2)$, where $\widetilde{\deg}(f)$ is the approximate degree of $f(x)$? Establishing such a result would provide a fundamental complexity-theoretic limitation on randomized quantum algorithms for implementing polynomial transformations. Notice that for the standard QSVT, a query complexity lower bound of $\Omega(\widetilde{\deg}(f))$ was proved in \cite{montanaro2024quantum}.

%\item {\bf Integration with hybrid algorithms:} Can our QSVT framework be effectively integrated into hybrid quantum-classical algorithms, such as variational quantum algorithms (VQAs) \cite{cerezo2021variational,bu2025exploring} or quantum approximate optimization algorithms (QAOAs) \cite{blekos2024review}? This would provide evidence that QSVT could be applied to more near-term quantum algorithms.

\item{\bf Interplay between physical and algorithmic errors:~} Our work shows that errors due to Trotterization can be efficiently mitigated, even when it is integrated into a larger, more general quantum circuit. Thus, is it possible to develop a theory of mitigating algorithmic errors? Moreover, quantum algorithms are also subjected to physical errors leading to noisy quantum circuits. It would be interesting to understand the interplay between these two sources of errors, and if they can be jointly mitigated while still retaining quantum advantage \cite{mohammadipour2025direct}.

\end{itemize}

Our results provide a new perspective on QSVT and quantum polynomial transformations, demonstrating that block encoding is not a necessary requirement. The framework of the interleaved-sequence circuit raises new fundamental questions in quantum complexity and algorithm design, suggesting possible connections to many different quantum algorithmic paradigms. We anticipate that this approach will lead to deeper insights into quantum algorithm design, both in near-term and fully fault-tolerant settings.

\vspace{.3cm}

\noindent {\bf Data Availability:} No new data were created during this study.

\noindent {\bf Conflict of Interest:} The authors declare that they have no financial interests.

\section*{Acknowledgments}
CS and YZ are supported by the National Key Research Project of China under Grant No. 2020YFA0712300.
SC and SH acknowledge funding from the Ministry of Electronics and Information Technology (MeitY), Government of India, under Grant No. 4(3)/2024-ITEA. SC also acknowledges support from Fujitsu Ltd, Japan, and IIIT Hyderabad via the Faculty Seed Grant. TL and XW were supported by the National Natural Science Foundation of China (Grant Numbers 62372006 and 92365117), and the Fundamental Research Funds for the Central Universities, Peking University. XW thanks the University of California, Berkeley, for its hospitality during his visit, where a part of this work was conducted. We thank Andrew M. Childs and Ronald de Wolf for valuable feedback on this work.

%\newpage

\begin{appendices}

\section[Lower bound on the number of ancilla qubits required for block encoding]{Lower bound on the ancilla requirement for block encoding an LCU}
\label{sec-app:ancilla-lb-lcu}

Consider an operator $A$ that can be written as a linear combination of $L$ unitaries, i.e., $A=\sum_{j=1}^L \lambda_j P_j$. 
Suppose that each $P_j$ can be implemented efficiently on a quantum computer, and are of the same dimension. 
In this section, we prove that in order to block encode $A$ exactly, $\Omega(\log L)$ ancilla qubits are required. 
To derive such a general result, we consider a family of quantum circuits $U$, which can possibly be block encodings of $A$. 
More specifically, we consider the following quantum circuit  (see Fig. \ref{circuit of U}):
\be
\label{eq-app:-general-unitary-be-lcu}
U = (V_1 \otimes I) c_0\text{-}P_1 (V_2 \otimes I) c_0\text{-}P_2 \cdots 
(V_L \otimes I) c_0\text{-}P_L (V_{L+1} \otimes I),
\ee
where $V_1,\ldots,V_{L+1}$ act on some $m$ ancilla qubits. 
Additionally, the controlled unitaries $c_0\text{-}P_j$ (which is controlled by 0) are defined as:
\be
\label{eq-app:controlled-unitary}
c_0\text{-}P_j := 
\ket{0^m} \bra{0^m} \otimes P_j + \sum_{k \neq 0^m} \ket{k} \bra{k} \otimes I
= \ket{0^m} \bra{0^m} \otimes (P_j-I) + I_m \otimes I.
\ee

The above quantum circuit is indeed quite general, as it includes LCU as a special case, which will be explained shortly. 
Additionally, since all $P_j$ act on the same space as $A$, it is reasonable to assume they act on the same qubit space as $A$ in \eqref{eq-app:-general-unitary-be-lcu}. 
To make the circuit more general, we can introduce some unitaries $V_1,\ldots,V_{L+1}$ acting on the ancilla qubits.

\begin{figure}[h]
\[
\Qcircuit @C=1em @R=1em {
\lstick{\ket{0^m}} & {/} \qw 
      & \gate{V_1} & \ctrlo{1} & \gate{V_2} & \ctrlo{1} & \qw 
      & \push{\kern-0.85em\mbox{$\cdots$}} & \ctrlo{1} & \gate{V_{L+1}} & \qw \\
\lstick{\ket{\psi}} & {/} \qw 
      & \qw       & \gate{P_1} & \qw       & \gate{P_2} & \qw 
      & \push{\kern-0.85em\mbox{$\cdots$}} & \gate{P_L} & \qw         & \qw
}
\]
\caption{The quantum circuit of $U$ defined via \eqref{eq-app:-general-unitary-be-lcu}.}
\label{circuit of U}
\end{figure}

Recall from Sec. \ref{subsec-prelim:blk-encoding-qsvt} that the quantum circuit for LCU is
\be \label{app:lcu}
(\texttt{PREP}^\dag \otimes I) \Big( \sum_{j=1}^L \ket{j} \bra{j} \otimes P_j \Big) (\texttt{PREP} \otimes I),
\ee
where $\texttt{PREP}$ is a unitary acting on $\ell = \lceil \log L\rceil$ ancilla qubits and maps $\ket{0^\ell}$ to $\sum_j \sqrt{\lambda_j/\lambda} \, \ket{j}$, where $\lambda=\sum_j \lambda_j$. 
Here, we assumed $\lambda_j>0$ without loss of generality. 
In the LCU circuit \eqref{app:lcu}, each $P_j$ is controlled by $\ket{j}$. 
However, we can easily find a permutation $S_j$ that maps $\ket{j}$ to $\ket{0^\ell}$, such that controlled operation cont$_j$-$P_j$ equals
$(S_j^\dag\otimes  I) c_0\text{-}P_j (S_j\otimes I)$. 
Using this observation, it is easy to see that LCU is indeed a special case of the circuit in \eqref{eq-app:-general-unitary-be-lcu}, see Fig. \ref{circuit of LCU}.

\begin{figure}[h]
\[
\Qcircuit @C=1em @R=1em {
\lstick{\ket{0^\ell}} & {/} \qw  & \gate{\texttt{PREP}} 
   & \gate{\texttt{CIRC}} & \ctrlo{1} & \gate{\texttt{CIRC}} & \ctrlo{1} & \qw & \push{\kern-0.85em\mbox{$\cdots$}}  
   & \gate{\texttt{CIRC}} & \ctrlo{1} 
   & \gate{\texttt{CIRC}} & \gate{\texttt{PREP}^\dagger} & \qw \\
\lstick{\ket{\psi}}  & {/} \qw & \qw 
   & \qw               & \gate{P_1} & \qw               & \gate{P_2} & \qw & \push{\kern-0.85em\mbox{$\cdots$}} 
   & \qw               & \gate{P_L} & \qw              & \qw                & \qw
}
\]
\caption{The quantum circuit of LCU, where $\ell = \lceil \log L\rceil$, $\texttt{PREP}~\ket{0^\ell}=\sum_{j=1}^{L}\sqrt{\lambda_{j}/\lambda} \, \ket{j}$ with $\lambda=\sum_{j}\lambda_j$, and $\texttt{CIRC}\coloneqq\sum_{j=1}^{L} \ket{j-1}\bra{j} +\ket{L}\bra{0}$ is a permutation.}
\label{circuit of LCU}
\end{figure}

In this section, we aim to prove the following theorem:

\begin{thm}
\label{thm:app-ancilla-lower-bound}
If $U$ is an exact block encoding of $A$, that is,
$$
\left(\bra{0^{\otimes m}}\otimes I\right) U \left(\ket{0^{\otimes m}}\otimes I\right)=A,
$$
then the number of ancilla qubits, $m$, satisfies $m=\Omega(\log L)$.
\end{thm}

\begin{proof}
From Eq.~\eqref{eq-app:controlled-unitary}, we have
\[
(V_j \otimes I) c_0\text{-}P_j = 
V_j \ket{0^m} \bra{0^m} \otimes (P_j-I) + V_j\otimes I.
\]
Expanding $A = \left(\bra{0^{\otimes m}}\otimes I\right) U \left(\ket{0^{\otimes m}}\otimes I\right)$, we obtain:
\bes
\begin{aligned}
& \quad \sum_{j=1}^L \lambda_j P_j \\
& =\; \bra{0^m} V_1 V_2 \cdots V_{L+1} \ket{0^m} \cdot I \\
&+ \sum_{j=1}^L \bra{0^m} V_1 \cdots V_j \ket{0^m} \bra{0^m}  V_{j+1} \cdots V_{L+1} \ket{0^m} \cdot (P_j-I)\\
&+ \sum_{1\leq j_1<j_2\leq L} \bra{0^m} V_1 \cdots V_{j_1} \ket{0^m} \bra{0^m}  V_{j_1+1} \cdots V_{j_2} \ket{0^m} \bra{0^m} V_{j_2+1} \cdots V_{L+1} \ket{0^m} \cdot (P_{j_1}-I) (P_{j_2}-I)  \\
&+ \cdots \cdots \\
% &+ \sum_{1\leq j_1<\cdots<j_\ell \leq L}
% \bra{0^m} V_1 \cdots V_{j_1} \ket{0^m} \bra{0^m}  V_{j_1+1} \cdots  V_{j_2} \ket{0^m} \bra{0^m} V_{j_2+1} \cdots V_{j_\ell} \ket{0^m} \\
% &\hspace{2.53cm} \bra{0^m} V_{j_\ell+1} \cdots V_{L+1} \ket{0^m} (P_{j_1}-I) (P_{j_2}-I) \cdots (P_{j_\ell}-I) \\
&+ \sum_{1\leq j_1<\cdots<j_\ell \leq L}
\bra{0^m} V_1 \cdots V_{j_1} \ket{0^m} \Big(\prod_{r=2}^{\ell} \bra{0^m} V_{j_{r-1}+1} \cdots V_{j_r} \ket{0^m}\Big)
\bra{0^m}  V_{j_{\ell}+1} \cdots  V_{L+1} \ket{0^m} \cdot \prod_{r=1}^{\ell} (P_{j_r}-I) \\
&+ \cdots \cdots \\
&+ \prod_{r=1}^{L+1} \bra{0^m} V_r \ket{0^m} \cdot \prod_{r=1}^L (P_r-I). \\
% &+ \bra{0^m} V_1 \ket{0^m} \bra{0^m} V_2 \ket{0^m} \bra{0^m} V_3 \ket{0^m} \cdots \bra{0^m} V_L \ket{0^m} \bra{0^m} V_{L+1} \ket{0^m}
% (P_1-I) (P_2-I) \cdots (P_L-I).
\end{aligned}
\ees
The above equality should hold for all unitaries $P_1,\ldots,P_L$, so we obtain a series of equations:
\beas
&& \prod_{r=1}^{L+1} \bra{0^m} V_r \ket{0^m} =0 \\
% && \bra{0^m} V_1 \ket{0^m} \bra{0^m} V_2 \ket{0^m} \bra{0^m} V_3 \ket{0^m} \cdots \bra{0^m} V_L \ket{0^m} \bra{0^m} V_{L+1} \ket{0^m} = 0 \\
&& \cdots \cdots \\
&& \bra{0^m} V_1 \cdots V_{j_1} \ket{0^m} \Big(\prod_{r=2}^{\ell} \bra{0^m} V_{j_{r-1}+1} \cdots V_{j_r} \ket{0^m}\Big)
\bra{0^m}  V_{j_{\ell}+1} \cdots  V_{L+1} \ket{0^m}=0 \\
% && \bra{0^m} V_1 \cdots V_{j_1} \ket{0^m} \bra{0^m}  V_{j_1+1} \cdots  V_{j_2} \ket{0^m} \bra{0^m} V_{j_2+1} \cdots V_{j_\ell} \ket{0^m} 
% \bra{0^m} V_{j_\ell+1} \cdots V_{L+1} \ket{0^m} = 0 \\
&& \cdots \cdots \\
&& \bra{0^m} V_1 \cdots V_{j_1} \ket{0^m} \bra{0^m}  V_{j_1+1} \cdots V_{j_2} \ket{0^m} \bra{0^m} V_{j_2+1} \cdots V_{L+1}\ket{0^m} = 0 \\
&& \bra{0^m} V_1 \cdots V_j \ket{0^m} \bra{0^m}  V_{j+1} \cdots V_{L+1} \ket{0^m} = \lambda_j \neq 0 \\
&& \bra{0^m} V_1V_2 \cdots V_{L+1} \ket{0^m}  =  \lambda_1+\cdots+\lambda_L.
\eeas
Apart from the last $L+1$ equations, all others lead to some orthogonalities. 
Indeed, from induction, it is easy to see that we have $L(L-1)/2$ such equations: 
\beas
&& \bra{0^m} V_2 \ket{0^m} = \bra{0^m} V_3 \ket{0^m} = \cdots = \bra{0^m} V_{L} \ket{0^m} = 0, \\
&& \bra{0^m} V_2V_3 \ket{0^m} = \bra{0^m} V_3V_4 \ket{0^m} = \cdots = \bra{0^m} V_{L-1}V_{L} \ket{0^m} = 0, \\
&& \bra{0^m} V_2V_3V_4 \ket{0^m} = \bra{0^m} V_3V_4V_5 \ket{0^m} = \cdots = \bra{0^m} V_{L-2}V_{L-1}V_{L} \ket{0^m} = 0, \\
&& \cdots \cdots \\
&& \bra{0^m} V_2\cdots V_{L-1}V_{L} \ket{0^m} = 0.
\eeas
If we consider
\beas
\mathcal{H}_2 &=& \text{Span}\{\bra{0^m} \mathcal{S}_2 \} , \quad \mathcal{S}_2 = \{V_2\} \\
\mathcal{H}_3 &=& \text{Span} \{\bra{0^m} \mathcal{S}_3 \}, \quad \mathcal{S}_3=\{V_2, V_3, V_2V_3\} \\
\mathcal{H}_4 &=& \text{Span} \{\bra{0^m} \mathcal{S}_4 \}, \quad \mathcal{S}_4=\{V_2, V_3, V_4, V_2V_3, V_3 V_4, V_2V_3V_4\}  \\
&& \cdots \cdots \\
\mathcal{H}_{j} &=& \text{Span} \{\bra{0^m} \mathcal{S}_{j} \}, \quad 
\mathcal{S}_{j}= \{V_{j},  \mathcal{S}_{j-1}, \mathcal{S}_{j-1} V_{j} \}.
\eeas
Then, we can show that $\dim(\mathcal{H}_{j}) \geq \dim(\mathcal{H}_{j-1}) + 1 $. 
This can be proved by contradiction. 
Note that $\mathcal{H}_{j-1} \subseteq \mathcal{H}_{j}$. If $\dim(\mathcal{H}_{j}) = \dim(\mathcal{H}_{j-1}) $, then $\mathcal{H}_{j} = \mathcal{H}_{j-1}$. 
This means $\bra{0^m} V_j \in \text{Span} \{\bra{0^m} \mathcal{S}_{j-1}\} = \text{Span} \{\bra{0^m} \mathcal{S}_{j-1} V_j\}$. This means $\bra{0^m} \in \mathcal{H}_{j-1}$, which is a contradiction because  $\bra{0^m}$ is orthogonal to $\mathcal{H}_{j-1}$. From this we conclude that $\dim(\mathcal{H}_{L}) \geq L-1$, which implies $m=\Omega(\log L)$.
\end{proof}

\section{Distance between two quantum states}
\label{appA:Distance between two quantum states}

First, consider that there exist two operators $P$ and $Q$ such that $\norm{P-Q}\leq \gamma$. We demonstrate that the expectation value of $O$ with respect to $P\rho P^{\dag}$ is not far off from the expectation value of $O$ with respect to $Q\rho Q^{\dag}$, for any density matrix $\rho$. More precisely, we prove
$$
\left|\Tr[O~P\rho P^{\dag}]-\Tr[O~Q\rho Q^{\dag}]\right|\leq 3\norm{O}\gamma
$$
for $\|P\|\le 1$.
This was proven in \cite{chakraborty2024implementing}, but here we provide a simpler proof. Let us recall the tracial version H\"{o}lder's inequality, which is stated below for completeness:
\begin{lem}[Tracial version of H\"{o}lder's inequality \cite{ruskai1972inequalities}]
\label{thm:holder}
Define two operators $A$ and $B$ and parameters $p,q\in [1,\infty]$ such that $1/p+1/q =1 $. Then the following holds:
$$
\Tr[A^{\dag}B]\leq \norm{A}_p \norm{B}_q.
$$
\end{lem}
Here $\norm{X}_p$ corresponds to the Schatten $p$-norm of the operator $X$. For the special case of $p=\infty$ and $q=1$, the statement of Lemma \ref{thm:holder} can be rewritten as
\begin{equation}
\label{eq:holder-special-case}
\Tr[A^{\dag}B] \leq \norm{A}_{\infty} \norm{B}_1=\norm{A} \norm{B}_1.
\end{equation}

\begin{thm}
\label{thm:distance-expectation}
Suppose $P$ and $Q$ are operators such that $\norm{P-Q}\leq \gamma$ for some $\gamma\in [0,1]$. Furthermore, let $\rho$ be any density matrix and $O$ be some Hermitian operator with spectral norm $\norm{O}$. Then, if $\norm{P}\leq 1$, the following holds:
$$
\left|\Tr[O~P\rho P^{\dag}]-\Tr[O~Q\rho Q^{\dag}]\right| \leq 3\norm{O}\gamma.
$$
\end{thm}
\begin{proof}
Using Eq.~\eqref{eq:holder-special-case}, we have
\begin{align}\label{eq:robustness-exp-inequality}
|\Tr[O~P\rho P^{\dag}]-&\Tr[O~Q\rho Q^{\dag}]|\leq \norm{O}\cdot \|P\rho P^{\dag}-Q\rho Q^{\dag}\|_1
\end{align}
For the second term in the RHS of Eq.~\eqref{eq:robustness-exp-inequality}, we can successively apply Eq.~\eqref{eq:holder-special-case} to obtain:
\begin{align*}
\bigg\|P\rho P^{\dag}-Q\rho Q^{\dag}\bigg\|_1&=\norm{P\rho P^{\dag}-P\rho Q^{\dag}+P\rho Q^{\dag}-Q\rho Q^{\dag}}_1\\
                &\leq \norm{P\rho}_1\norm{P-Q}+\norm{P-Q}\norm{\rho Q}_1\\
                &\leq \norm{P}\norm{P-Q}+\norm{Q}\norm{P-Q}~~~~~~~~~~~~~~~~(\text{As~}\norm{\rho}_1=1)\\
                &\leq \left(\|P\|+\norm{P-Q}+\norm{P}\right)\cdot \norm{P-Q}\\
                &= 2\norm{P}\norm{P-Q}+\norm{P-Q}^2.
\end{align*}
Now, substituting this upper bound back in the RHS of Eq.~\eqref{eq:robustness-exp-inequality}, we obtain
\begin{align*}
\left|\Tr[O~P\rho P^{\dag}] - \Tr[O~Q\rho Q^{\dag}]\right| &\leq \big\|O\big\| \big\|P-Q\big\|^2 + 2\big\|O\big\|\big\|P\big\|\big\|P-Q\big\|\\
&\leq \gamma^2\big\|O\big\|+2\big\|O\big\|\big\|P\big\|\gamma \nonumber
\leq 3\gamma\big\|O\big\|.
\end{align*}
This completes the proof.
\end{proof}

\section{Functions approximated by Laurent polynomials}
\label{sec-app:function-laurent}

In this section, we show that several important functions can be approximated by Laurent polynomials in $e^{\i x}$, bounded in $\mathbb{T}=\{x\in \mathbb{C}: |x|=1\}$. All these functions are implementable by our methods in Sec.~\ref{sec:qsvt-trotter}. We begin with the shifted sign function.
~\\~\\
\textbf{Shifted sign function:~}For any $\Delta\in (0, 1)$, a trigonometric approximation of $\sgn(\sin(x))$ over $[-\pi+\Delta, -\Delta]\cup [\Delta, \pi-\Delta]$ is given by \cite{wang2023quantum}. Since $x\in [-\pi,\pi]$, this is also a trigonometric approximation of $\sgn(x)$.
\begin{lem}[Lemma S3 of \cite{wang2023quantum}]
\label{lem:approx-sign}
    For any $\Delta', \varepsilon'\in (0,1)$, there exists a Laurent polynomial $P(z)$ of $\deg(P) = O(\log(1/\varepsilon')/\Delta')$ such that 
    \begin{align*}
        &|P(z)|\le 1 && \text{ for all } z\in \mathbb{T}, \\
        &|P(e^{\i x}) -\sgn(x)|\le \varepsilon' && \text{ for all } x\in [-\pi+\Delta', -\Delta']\cup [\Delta', \pi-\Delta'].
    \end{align*}
\end{lem}
Then, we obtain the following:
\begin{lem}[Trigonometric approximation of the shifted sign function]
\label{lem:approx-shift-sign}
    For any $\Delta, \varepsilon\in (0,1)$ and $\mu\in[-1,1]$ such that $\Delta \le 2\min\{|\mu-1|, |\mu+1|\}$, there exists a Laurent polynomial $P(z)$ of $\deg(P) = O(\log(1/\varepsilon)/\Delta)$ such that 
    \begin{align*}
        &|P(z)|\le 1 && \text{ for all } z\in \mathbb{T}, \\&P(e^{\i x}) \in [1-\varepsilon, 1] && \text{ for all } x\in [-1,\mu-\Delta/2 ], \\
        &P(e^{\i x}) \in [0, \varepsilon] && \text{ for all } x\in [\mu+\Delta/2,1 ].
    \end{align*}
\end{lem}
\begin{proof}
    Let $P_1(z)$ be the Laurent polynomial in Lemma~\ref{lem:approx-sign} with parameters $(\Delta', \varepsilon')=(\min\{\pi-2, \Delta/2\}, 2\varepsilon)$. Then $P_1(e^{\i x})$ is an $2\varepsilon$-approximation of $\sgn(x)$ for all $x\in [-2, -\Delta/2]\cup[\Delta/2, 2]$. The Laurent polynomial $P(e^{\i x})=(P_1(e^{\i (x-\mu)})+1)/2$ satisfies the desired properties.
\end{proof}
%The shifted sign function has been used in the fuzzy biselection problem for ground state energy estimation (Problem \ref{prob:fuzzy bisection}).

Next, we construct the filter function used for solving linear systems in Theorem \ref{thm:qls-our-method}.
\begin{lem}[Trigonometric approximation of the filter function]
\label{lem:tri-approx-filter}
    For any $\Delta, \varepsilon\in (0,1)$, there exists a Laurent polynomial $P(z)$ of degree $O(\log(1/\varepsilon)/\Delta)$ such that $P(1) \in [1-\varepsilon, 1]$ and \begin{align*}
        &|P(z)|\le 1 && \text{ for all } z\in \mathbb{T},  \\
        &P(e^{\i x}) \in [-\varepsilon/2, \varepsilon/2] && \text{ for all } x\in [-1,-\Delta]\cup [\Delta, 1].
    \end{align*}
\end{lem}
\begin{proof}
    Let $P_1(z)$ be the Laurent polynomial in Lemma~\ref{lem:approx-sign} with parameters $(\Delta', \varepsilon')=(\min\{\pi-2, \Delta/2\}, \varepsilon)$. Then the Laurent polynomial $P(e^{\i x}) = (P_1(e^{\i(x+\Delta/2))})+P_1(e^{\i(-x+\Delta/2)})/2$ satisfies the desired properties. 
\end{proof}
~\\
\textbf{Rectangle function:~}Combining two approximate shifted sign functions yields an approximate rectangle function. 
\begin{lem}[Trigonometric approximation of the the rectangle function]
\label{lem-approx-rec}
For any $\delta, \varepsilon\in (0,1/2)$ and $t\in [-1,1]$, there exists a Laurent polynomial $P(z)$ of degree $ O(\log(1/\varepsilon)/\delta)$ such that $|P(z)|\le 1$ for all $z\in \mathbb{T}$ and 
    \begin{align*}
        &P(e^{\i x}) \in [1-\varepsilon, 1] && \text{ for all } x\in [-t+\delta,t-\delta], \\
        &P(e^{\i x}) \in [0, \varepsilon] && \text{ for all } x\in [-1,-t-\delta]\cup[t+\delta,1].
    \end{align*}
\end{lem}
\begin{proof}
    Let $P_1(z)$ be the Laurent polynomial in Lemma~\ref{lem:approx-sign} with $(\Delta', \varepsilon') = (\min\{\pi-2, \delta\}, \varepsilon)$ and $P(e^{\i x})=(P_1(e^{\i (x-t)})+P_1(e^{\i (-x-t)}))/2$ satisfies the desired properties.
\end{proof}
    
~\\
\textbf{Inverse function:~}For the inverse function, we can transform the polynomial approximation of $1/x$ constructed in \cite{childs2017qls} to a trigonometric approximation using Proposition~\ref{lem:gqsp-poly2tri}. 

\begin{lem}[Trigonometric approximations of the inverse function]
    For any $\kappa > 1$ and $\varepsilon\in (0,1/2)$, there exists a Laurent polynomial $P(z)$ of degree $O(\kappa \log^2(\kappa/\varepsilon))$ and $B=O(\kappa\log(\kappa/\varepsilon))$ such that \begin{align*}
        &|P(z)|\le 1 && \text{ for all } z\in \mathbb{T},\\
        &|P(e^{\i \pi x/3})-1/(Bx)| \le \varepsilon && \text{ for all }x \in [-1,-1/\kappa]\cup [1/\kappa, 1].
    \end{align*}
\end{lem} 
\begin{proof}
    \cite[Lemmas 17-19]{childs2017qls} constructed an $O(\kappa\log(\kappa/\varepsilon))$-degree polynomial $g(x)$ that is bounded by $B = O(\kappa \log(\kappa/\varepsilon))$ in $[-1,1]$ and $(B\varepsilon/2)$-close to $1/x$ in $[-1,-1/(2\kappa)]\cup[1/(2\kappa), 1]$. Then $\tilde{g}(x)\coloneqq g(x/2)/B$ is $\varepsilon/2$-close to $1/(Bx)$ in $[-2,-1/\kappa]\cup[1/\kappa,2]$ and bounded by $1$ in $[-2,2]$. Therefore, we can apply Proposition~\ref{lem:gqsp-poly2tri} to $\tilde{g}(x)$ with parameters $(\varepsilon, \delta)=(\varepsilon/2, 1/2)$, which gives us a Laurent polynomial $P(z)$ of degree $O(\kappa \log^2(\kappa/\varepsilon))$ such that 
    \begin{align*}
         &|P(z)|\le 1 && \text{for all } z\in \mathbb{T},\\
        &|P(e^{\i \pi x/3})-\tilde{g}(x)|\le \varepsilon/2 && \text{for all } x \in [-1,1].
    \end{align*}
    The Laurent polynomial $P(e^{\i \pi x/3})$ is $\varepsilon$-close to $1/(Bx)$ in $[-1,1/\kappa]\cup[1/\kappa,1]$ as $\tilde{g}(x)$ is $\varepsilon/2$-close $1/(Bx)$ in this interval. 
\end{proof}
~\\
\textbf{Exponential function:~}We can also construct the trigonometric polynomial approximation of the exponential function used in the quantum Gibbs sampler \cite{van2017quantum}.
\begin{lem}[Trigonometric approximations of the exponential function]
For any $\beta>1$ and $\varepsilon\in  (0,1/2)$, there exists a Laurent polynomial $P(z)$ of degree $O(\beta\log(1/\varepsilon))$ such that \begin{align*}
         &|P(z)|\le 1 && \text{for all } z\in \mathbb{T},\\
        &|P(e^{\i \pi x/(2(1+1/\beta))})-e^{\beta(x-1)-1}|\le \varepsilon && \text{for all } x \in [-1,1].
    \end{align*}
\end{lem}
\begin{proof}
    The function $e^{\beta(x-1)-1}$ has Taylor series \begin{align*}
        e^{\beta(x-1)-1}=e^{-\beta-1}\sum_{j=0}^{\infty} \frac{\beta^j}{j!} x^j,
    \end{align*}
    which we denote by  $\sum_{j=0}^{\infty} a_j x^j$. Let $\delta = 1/\beta$ and $B = \sum_{j=0}^{\infty} a_j(1+\delta)^j=e^{-\beta-1+\beta(1+1/\beta)}=1$.
    Then applying Proposition~\ref{lem:bounded-1-norm-gqsp} with $f(x)=\tilde{f}(x)=e^{\beta(x-1)-1}$ gives the claimed results. 
\end{proof}

\end{appendices}

\bibliographystyle{unsrt}
\bibliography{qsvt.bib}

\end{document}